\documentclass[11pt,english]{article}
\usepackage[T1]{fontenc}
\usepackage[latin9]{inputenc}
\usepackage{array}
\usepackage{amsmath}
\usepackage{graphicx}
\usepackage{amssymb}
\usepackage{array}
\usepackage{float}

\newtheorem{theorem}{Theorem}
\newtheorem{lemma}{Lemma}

\newtheorem{definition}{Definition}

\newtheorem{remark}{Remark}

\newenvironment{proof}{\begin{trivlist}
\item[]{\bf Proof. }}{\hspace*{\fill}$\rule{.3\baselineskip}{.35\baselineskip}$
\end{trivlist}}

\newcommand{\R}{\mathbb{R}}
\newcommand{\C}{\mathbb{C}}
\newcommand{\Z}{\mathbb{Z}}
\newcommand{\N}{\mathbb{N}}

\begin{document}

\title{\textbf{Multi-site breathers in Klein--Gordon lattices: }\\
 \textbf{ stability, resonances, and bifurcations}}

\author{Dmitry Pelinovsky and Anton Sakovich \\
 {\small Department of Mathematics, McMaster University, Hamilton
ON, Canada, L8S 4K1} }

\date{}
\maketitle

\begin{abstract}
We prove a general criterion of spectral stability of multi-site
breathers in the discrete Klein--Gordon equation with a small coupling
constant. In the anti-continuum limit, multi-site breathers represent
excited oscillations at different sites of the lattice separated by
a number of ``holes" (sites at rest). The criterion
describes how the stability or instability of a multi-site breather
depends on the phase difference and distance between the excited oscillators.
Previously, only multi-site breathers with adjacent excited sites were
considered within the first-order perturbation theory. We show that
the stability of multi-site breathers with one-site holes changes
for large-amplitude oscillations in soft nonlinear potentials.
We also discover and study a symmetry-breaking (pitchfork) bifurcation
of one-site and multi-site breathers in soft quartic potentials near
the points of 1:3 resonance.
\end{abstract}

\section{Introduction}

Space-localized and time-periodic breathers in nonlinear Hamiltonian
lattices have been studied extensively in the past twenty years. These
breathers model the oscillatory dynamics of particles due to the external
forces and the interaction with other particles. A particularly simple
model is the Klein--Gordon lattice, which is expressed by the discrete
Klein--Gordon equation,
\begin{equation}
\ddot{u}_{n}+V'(u_{n})=\epsilon(u_{n+1}-2u_{n}+u_{n-1}),\quad
n\in\mathbb{Z},\label{KGlattice}
\end{equation}
where $t\in\R$ is the evolution time, $u_{n}(t)\in\R$ is the displacement
of the $n$-th particle, $V:\R\to\R$ is a smooth on-site potential
for the external forces, and $\epsilon\in\R$ is the coupling constant
of the linear interaction between neighboring particles. For the sake of clarity,
we will assume that the potential $V$ is symmetric, but a generalization
can be formulated for non-symmetric potentials $V$. We will also
assume that $V'(u)$ can be expanded in the power series near $u=0$
by
\begin{equation}
V'(u)=u\pm u^{3}+{\cal O}(u^{5})\quad
\mbox{{\rm as}}\quad u\to0.
\label{potential-expansion}
\end{equation}
The plus and minus signs are referred to as the hard and soft potentials,
respectively.

A simplification of analysis of the Klein--Gordon lattice was proposed
by MacKay and Aubry \cite{MA94} in the anti-continuum limit of small
coupling constant $\epsilon\to0$. This limit inspired many researchers
to study existence, stability, and global dynamics of space-localized
and time-periodic breathers \cite{Aubry}. Since all oscillators are
uncoupled at $\epsilon=0$, one can construct time-periodic breathers
localized at different sites of the lattice.
{\em Such time-periodic space-localized solutions supported on
a finite number of lattice sites at the anti-continuum limit are
called the multi-site breathers}. All these multi-site breathers are uniquely
continued with respect to the (small) coupling constant $\epsilon$
if the period of oscillations at different lattice sites is identical and the oscillations
are synchronized either in-phase or anti-phase.

Spectral stability of multi-site breathers, which are continued from the anti-continuum
limit $\epsilon=0$, was considered by Morgante {\em et al.} \cite{MJKA02}
with the help of numerical computations. These numerical computations
suggested that spectral stability of small-amplitude multi-site breathers in
the discrete Klein--Gordon equation (\ref{KGlattice}) is similar
to the spectral stability of multi-site solitons in the discrete nonlinear Schr\"{o}dinger
(DNLS) equation.

The DNLS approximation for small-amplitude and slowly varying oscillations
relies on the asymptotic solution,
\begin{equation}
u_{n}(t) = \epsilon^{1/2} \left[ a_{n}(\epsilon t)e^{it} +
\bar{a}_{n}(\epsilon t) e^{-it} \right] + \mathcal{O}_{l^{\infty}}(\epsilon^{3/2}),
\label{envelope-approximation}
\end{equation}
where $\epsilon>0$ is assumed to be small, $\tau=\epsilon t$ is the
slow time, and $a_{n}(\tau) \in \C$ is an envelope amplitude
of nearly harmonic oscillations with the linear frequency $\omega=1$. Substitution
of (\ref{envelope-approximation}) to (\ref{KGlattice}) yields the
DNLS equation to the leading order in $\epsilon$,
\begin{equation}
2i\dot{a}_{n}=a_{n+1}-2a_{n}+a_{n-1}\mp3|a_{n}|^{2}a_{n},\quad n\in\mathbb{Z}.
\label{NLSlattice}
\end{equation}
The hard and soft potentials (\ref{potential-expansion}) result
in the defocusing and focusing cubic nonlinearities of the DNLS equation
(\ref{NLSlattice}), respectively. Existence and continuous approximations
of small-amplitude breathers in the discrete Klein--Gordon
and DNLS equations were justified recently by Bambusi {\em et al.}
\cite{Bambusi,Bambusi1}. The problem of bifurcation of small-amplitude breathers
in Klein--Gordon lattices in connection to homoclinic bifurcations
in the DNLS equations was also studied by James {\em et al.} \cite{James1}.

Multi-site solitons of the DNLS equation (\ref{NLSlattice}) can be
constructed similarly to the multi-site breathers in the discrete Klein--Gordon
equation (\ref{KGlattice}). The time-periodic solutions are given by
$a_n(\tau) = A_n e^{-i \omega \tau}$, where $\omega \in \R$ is a frequency of oscillations
and $\{ A_n \}_{n \in \Z}$ is a real-valued sequence of amplitudes
decaying at infinity as $|n| \to \infty$.
In the anti-continuum limit (which corresponds here to the limit
$|\omega| \to \infty$ \cite{Pelin-book}), the multi-site solitons
are supported on a finite number of lattice sites. The oscillations
are in-phase or anti-phase, depending on the sign difference between
the amplitudes $\{ A_n \}_{n \in \Z}$ on the excited sites of the lattice.

Numerical results of \cite{MJKA02} can be summarized as follows.
{\em In the case of the focusing nonlinearity, the only stable multi-site
solitons of the DNLS equation (\ref{NLSlattice}) near the anti-continuum
limit correspond to the anti-phase oscillations on the excited sites of
the lattice.} This conclusion does not depend on the number of ``holes"
(oscillators at rest) between the excited sites at the anti-continuum
limit. The stable oscillations in the case of the defocusing nonlinearity
can be recovered from the stable anti-phase oscillations in the focusing case
using the staggering transformation,
\begin{equation}
a_{n}(\tau)=(-1)^{n}\bar{b}_{n}(\tau)e^{2i\tau},
\label{stag-transformation}
\end{equation}
which changes the DNLS equation (\ref{NLSlattice}) to the form,
\begin{equation}
2i\dot{b}_{n}=b_{n+1}-2b_{n}+b_{n-1}\pm3|b_{n}|^{2}b_{n},\quad n\in\mathbb{Z}.
\label{NLSlattice-stag}
\end{equation}
Consequently, we have the following statement. {\em In the case
of the defocusing nonlinearity, the only stable multi-site solitons of
the DNLS equation (\ref{NLSlattice}) with adjacent excited sites
near the anti-continuum limit correspond to the in-phase oscillations
on the excited sites of the lattice.}

The numerical observations of \cite{MJKA02} were rigorously proved
for the DNLS equation (\ref{NLSlattice}) by Pelinovsky {\em et
al.} \cite{PKF1}. Further details on the spectrum of a linearized operator
associated with the multi-site solitons near the anti-continuum limit
of the DNLS equation are obtained in our previous work \cite{PelSak}.

Similar conclusions on the spectral stability
of breathers in the discrete Klein--Gordon equation (\ref{KGlattice})
were reported in the literature under some simplifying assumptions.
Archilla {\em et al.} \cite{Archilla} used the perturbation theory
for spectral bands to consider two-site, three-site, and generally
multi-site breathers. Theorem 6 in \cite{Archilla}
states that {\em in-phase multi-site breathers are stable for hard
potentials and anti-phase breathers are stable for soft potentials
for $\epsilon>0$}. The statement of this result misses however that
the corresponding computations are justified for multi-site breathers
with adjacent excited sites: no ``holes" (oscillators
at rest) in the limiting configuration at $\epsilon=0$ are allowed.
More recently, Koukouloyannis and Kevrekidis \cite{KK09} recovered
exactly the same conclusion using the averaging theory for Hamiltonian
systems in action--angle variables developed earlier by MacKay
{\em et al.} \cite{MacKay1,MacKay2}.
To justify the use of the first-order perturbation theory,
the multi-site breathers were considered to have
adjacent excited sites and no holes. The equivalence between these
two approaches was addressed by Cuevas {\em et al.} \cite{ArchKK11}.

It is the goal of our paper to rigorously prove the stability criterion
for all multi-site breathers, including breathers with holes
between excited sites in the anti-continuum limit. We will use
perturbative arguments for characteristic exponents of the Floquet
monodromy matrices. To be able to work with the higher-order perturbation
theory, we will combine these perturbative arguments with the theory
of tail-to-tail interactions of individual breathers in lattice differential
equations. Although the tail-to-tail interaction theory is well-known
for continuous partial differential equations \cite{Sandstede}, it
is the first time to our knowledge when this theory is extended to
nonlinear lattices.

Multi-site breathers with holes have been recently considered by Yoshimura
\cite{Yoshimura} in the context of the diatomic Fermi-Pasta-Ulam
lattice near the anti-continuum limit. In order to separate
variables $n$ and $t$ and to perform computations
using the discrete Sturm theorem (similar to the one used in the context
of NLS lattices in \cite{PKF1}), the interaction potential was assumed
to be nearly homogeneous of degree four and higher. Similar work was
performed for the Klein--Gordon lattices with a purely anharmonic interaction
potential \cite{Yoshimura-2}. Compared to this work, our treatment is valid
for a non-homogeneous on-site potential $V$ satisfying expansion (\ref{potential-expansion})
and for the quadratic interaction potential.

Within our work, we have discovered new important details on the spectral
stability of multi-site breathers, which were missed in the previous
works \cite{Archilla,KK09,MJKA02}. In the case of soft potentials,
breathers of the discrete Klein--Gordon equation (\ref{KGlattice})
can not be continued far away from the small-amplitude limit described
by the DNLS equation (\ref{NLSlattice}) because of the resonances
between the nonlinear oscillators at the excited sites and the linear
oscillators at the sites at rest. Branches of
breather solutions continued from the anti-continuum limit
above and below the resonance are disconnected. In addition,
these resonances change the stability
conclusion. In particular, the anti-phase oscillations may become
unstable in soft nonlinear potentials even if the coupling constant
is sufficiently small.

Another interesting feature of soft potentials is the symmetry-breaking
(pitchfork) bifurcation of one-site and multi-site breathers that
occur near the point of resonances. In symmetric potentials, the first
non-trivial resonance occurs near $\omega=\frac{1}{3}$, that is,
at 1:3 resonance. We analyze this bifurcation by using asymptotic
expansions and reduction of the discrete Klein--Gordon equation (\ref{KGlattice})
to a normal form, which coincides with the nonlinear
Duffing oscillator perturbed by a small harmonic forcing.
It is interesting that the normal form equation for 1:3 resonance which
we analyze here is different from
the normal form equations considered in the previous studies of 1:3
resonance \cite{Broer,Simo,Simo-2}. The difference is explained by the
fact that we are looking at bifurcations of periodic solutions far from the
equilibrium points, whereas the standard normal form equations for 1:3
resonance are derived in a neighborhood of equilibrium points. Note
that an analytical study of bifurcations of small breather solutions close
to a point of 1:3 resonance for a diatomic Fermi--Pasta--Ulam lattice
was performed by James \& Kastner \cite{James2}.

The paper is organized as follows. Existence of space-localized and
time-periodic breathers near the anti-continuum limit is reviewed in
Section 2. Besides the persistence results based on the implicit
function arguments as in \cite{MA94}, we also develop
a new version of the tail-to-tail interaction theory for multi-site
breathers in the discrete Klein--Gordon equation (\ref{KGlattice}).
The main result on spectral stability of multi-site breathers for
small coupling constants is formulated and proved in Section 3.
Section 4 illustrates the existence and spectral stability
of multi-site breathers in soft potentials numerically. Section
5 is devoted to studies of the symmetry-breaking (pitchfork) bifurcation
using asymptotic expansions and normal forms for the 1:3 resonance.
Section 6 summarizes our findings. Appendix A
compares Floquet theory with the spectral band theory and Hamiltonian averaging to
show equivalence of our conclusions with those reported earlier in
\cite{Archilla,KK09}.

\section{Existence of multi-site breathers near the anti-continuum limit}

In what follows, we will use bold-faced notations for vectors in discrete
space $l^{p}(\mathbb{Z})$ defined by their norms
\[
\|{\bf u}\|_{l^{p}}:=\left(\sum_{n\in\mathbb{Z}}|u_{n}|^{p}\right)^{1/p},\quad p\geq1.
\]
Components of ${\bf u}$ are denoted by $u_{n}$ for $n\in\mathbb{Z}$.
These components can be functions of $t$, in which case they are
considered in Hilbert--Sobolev spaces $H_{{\rm per}}^{s}(0,T)$ of
$T$-periodic functions equipped with the norm,
\[
\|f\|_{H_{{\rm per}}^{s}}:=\left(\sum_{m\in\mathbb{Z}}(1+m^{2})^{s}|c_{m}|^{2}\right)^{1/2},\quad s\geq0,
\]
where the set of coefficients $\{c_{m}\}_{m\in\Z}$ defines
the Fourier series of a $T$-periodic function $f$,
\[
f(t)=\sum_{m\in\mathbb{Z}}c_{m}\exp\left(\frac{2\pi imt}{T}\right),\quad t\in[0,T].
\]

We consider space-localized and time-periodic breathers
${\bf u}\in l^{2}(\mathbb{Z},H_{{\rm per}}^{2}(0,T))$
of the discrete Klein--Gordon equation (\ref{KGlattice}) with smooth
even $V$ and $\epsilon>0$. Parameter $T>0$ represents the fundamental
period of the time-periodic breathers. Accounting for symmetries,
we shall work in the restriction of $H_{{\rm per}}^{s}(0,T)$ to the
space of even $T$-periodic functions,
\[
H_{e}^{s}(0,T) = \left\{ f\in H_{{\rm per}}^{s}(0,T): \quad
f(-t)=f(t),\quad t\in\R\right\} ,\quad s\geq0.
\]

At $\epsilon=0$, we have an arbitrary family of multi-site breathers,
\begin{equation}
{\bf u}^{(0)}(t)=\sum_{k\in S}\sigma_{k}\varphi(t){\bf e}_{k},
\label{limiting-breather}
\end{equation}
where ${\bf e}_{k}$ is the unit vector in $l^{2}(\Z)$, $S\subset\Z$
is a finite set of excited sites of the lattice, $\sigma_{k}\in\{+1,-1\}$
encodes the phase factor of the $k$-th oscillator, and $\varphi\in H_{{\rm per}}^{2}(0,T)$
is an even solution of the nonlinear oscillator equation at the energy level
$E$,
\begin{equation}
\ddot{\varphi}+V'(\varphi)=0 \quad\Rightarrow\quad
E=\frac{1}{2}\dot{\varphi}^{2}+V(\varphi).
\label{nonlinear-oscillator}
\end{equation}
The unique even solution $\varphi(t)$ satisfies the initial condition,
\begin{equation}
\label{ic-oscillator}
\varphi(0)=a,\quad\dot{\varphi}(0)=0,
\end{equation}
where $a$ is the smallest positive root of $V(a)=E$ for a fixed
value of $E$. Period of oscillations $T$ is uniquely defined by the energy level
$E$,
\begin{equation}
T=\sqrt{2}\int_{-a}^{a}\frac{d\varphi}{\sqrt{E-V(\varphi)}}.
\label{period-energy}
\end{equation}
Because $\varphi(t)$ is $T$-periodic, we have
\begin{eqnarray}
\label{id-oscillator1}
\partial_{E}\varphi(T) & = & a'(E) = \frac{1}{V'(a)}, \\
\label{id-oscillator3}
\partial_{E}\dot{\varphi}(T) & = & -\ddot{\varphi}(T)T'(E) = V'(a)T'(E).
\end{eqnarray}

Our main example of the nonlinear potential $V$ is the truncation of the
expansion (\ref{potential-expansion}) at the first two terms:
\begin{equation}
V'(u)=u\pm u^{3}.\label{potential-numerics}
\end{equation}
The dependence $T$ versus $E$ is computed numerically from (\ref{period-energy})
and is shown on Figure \ref{fig1} together with the phase portraits
of the system (\ref{nonlinear-oscillator}). For the hard potential
with the plus sign, the period $T$ is a decreasing function of $E$
in $(0,2\pi)$, whereas for the soft potential with the minus sign,
the period $T$ is an increasing function of $E$ with $T>2\pi$.

\begin{figure}
\begin{tabular}{cc}
\includegraphics[width=0.47\textwidth]{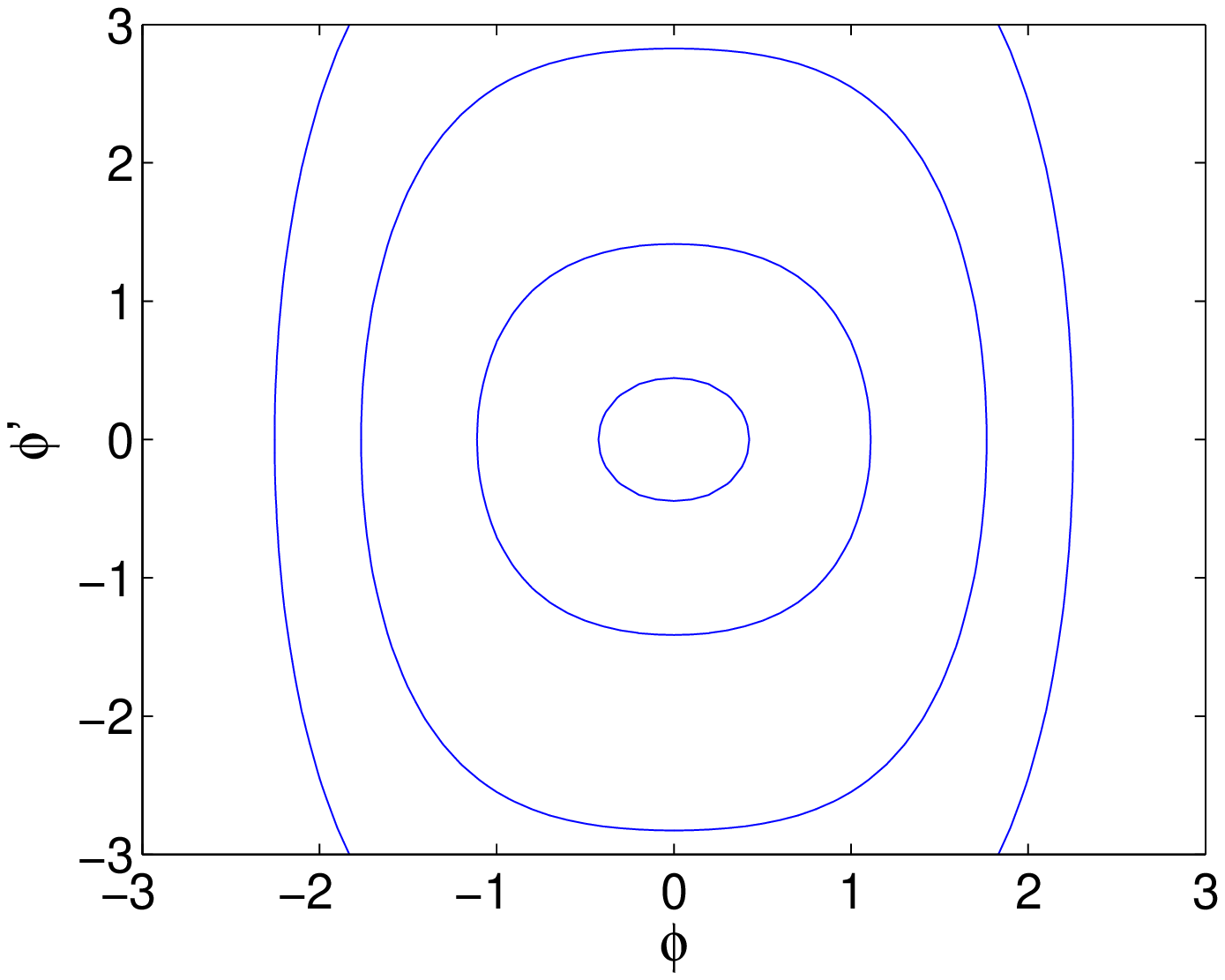}  & \includegraphics[width=0.47\textwidth]{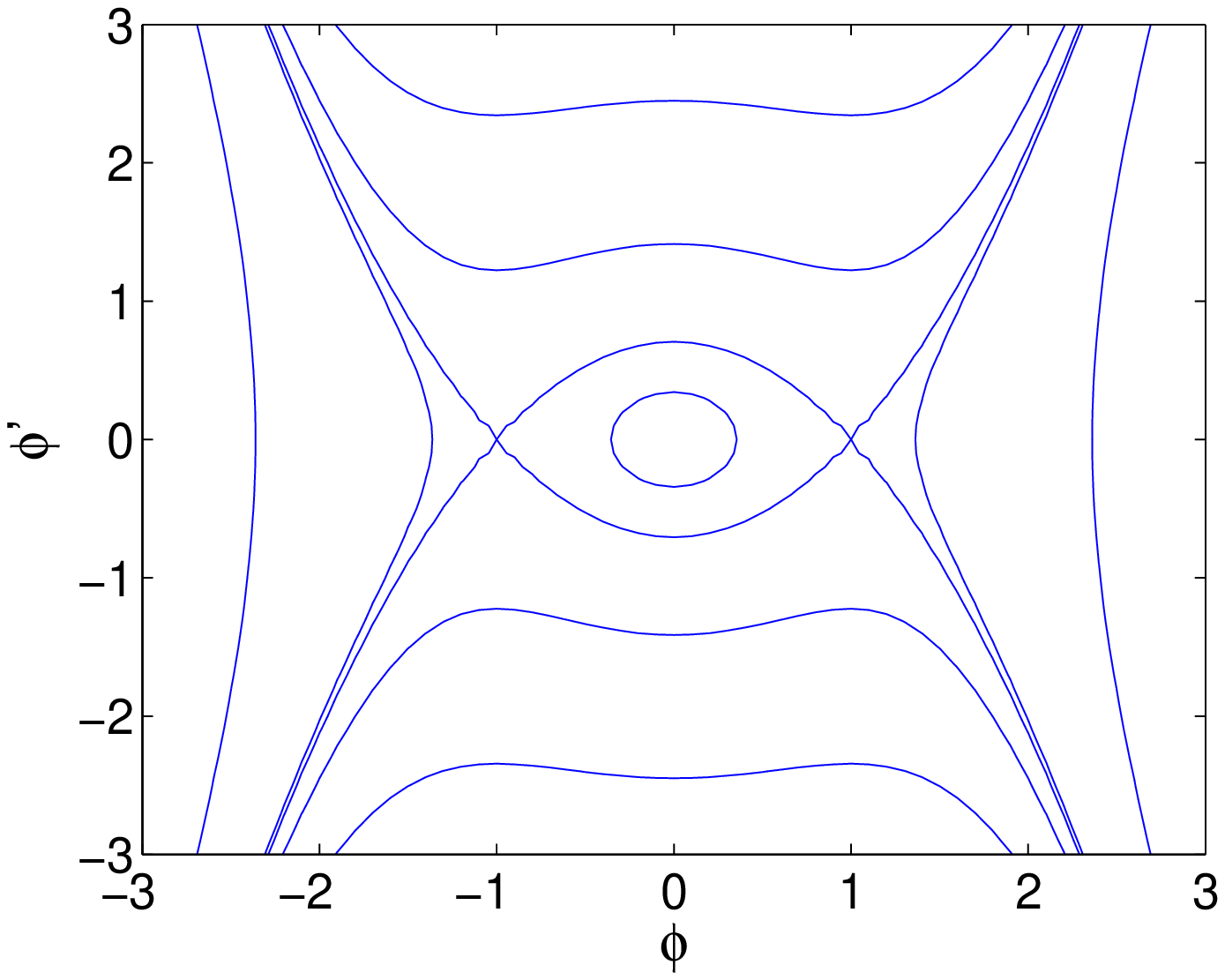}\tabularnewline
\multicolumn{2}{c}{\includegraphics[width=0.47\textwidth]{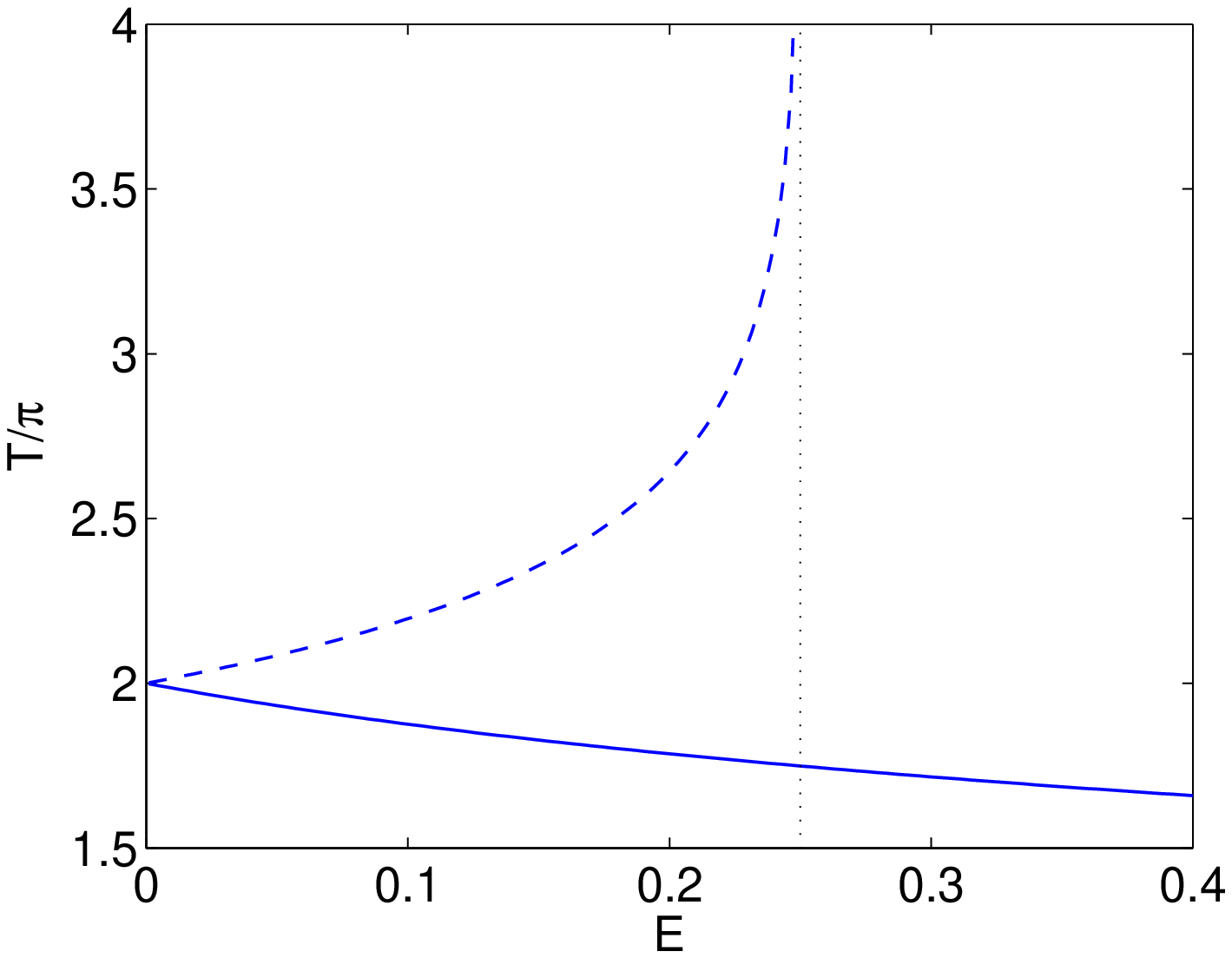}}\tabularnewline
\end{tabular} \caption{Top: the phase plane $(\varphi,\dot{\varphi})$ for the hard (left)
and soft (right) potentials (\ref{potential-numerics}). Bottom: the
period $T$ versus energy $E$ for the hard (solid) and soft (dashed)
potentials.}
\label{fig1}
\end{figure}

\begin{remark}
All nonlinear oscillators at the excited sites of
$S\subset\Z$ in the limiting configuration (\ref{limiting-breather})
have the same period $T$. Two oscillators at the $j$-th and $k$-th
sites are said to be in-phase if $\sigma_{j}\sigma_{k}=1$ and anti-phase
if $\sigma_{j}\sigma_{k}=-1$.
\end{remark}

Persistence of the limiting configuration (\ref{limiting-breather})
as a space-localized and time-periodic breather of the discrete Klein--Gordon
equation (\ref{KGlattice}) for small values of $\epsilon$ is established
by MacKay \& Aubry \cite{MA94}. The following theorem gives the relevant
details of the theory that are useful in our analysis.

\begin{theorem}
Fix the period $T$ and the solution $\varphi\in H_e^{2}(0,T)$
of the nonlinear oscillator equation (\ref{nonlinear-oscillator})
with an even $V\in C^{\infty}(\R)$ satisfying (\ref{potential-expansion})
and assume that $T\neq2\pi n$, $n\in\N$ and $T'(E)\neq0$. Define
${\bf u}^{(0)}$ by the representation (\ref{limiting-breather})
with fixed $S\subset\Z$ and $\{\sigma_{k}\}_{k\in S}$. There are
$\epsilon_{0}>0$ and $C>0$ such that for all $\epsilon\in(-\epsilon_{0},\epsilon_{0})$,
there exists a unique solution ${\bf u}^{(\epsilon)}\in l^{2}(\mathbb{Z},H_e^{2}(0,T))$
of the discrete Klein--Gordon equation (\ref{KGlattice}) satisfying
\begin{equation}
\|{\bf u}^{(\epsilon)}-{\bf u}^{(0)}\|_{l^{2}(\mathbb{Z},H_{{\rm per}}^{2}(0,T))}\leq C |\epsilon|.
\end{equation}
Moreover, the map $\R\ni\epsilon\mapsto{\bf u}^{(\epsilon)}\in l^{2}(\mathbb{Z},H_e^{2}(0,T))$
is $C^{\infty}$ for all $\epsilon\in(-\epsilon_{0},\epsilon_{0})$.
\label{theorem-continuum}
\end{theorem}

\begin{proof}
Thanks to the translational invariance of the nonlinear
oscillator equation (\ref{nonlinear-oscillator}), we fix an even
$\varphi$ according to the initial condition (\ref{ic-oscillator})
and consider $H_e^s(0,T)$, the restriction of $H_{{\rm per}}^{s}(0,T)$
to even functions. Under the condition
$T'(E)\neq0$, operator
\[
L_{e}=\partial_{t}^{2}+V''(\varphi(t)):H_{e}^{2}(0,T) \to L_{e}^{2}(0,T)
\]
is invertible, because the only eigenvector $\dot{\varphi}$ of
$L=\partial_{t}^{2}+V''(\varphi(t)):H_{{\rm per}}^{2}(0,T)\to L_{{\rm per}}^{2}(0,T)$
is odd in $t$. Similarly, operator
\[
L_{0}=\partial_{t}^{2}+1:H_{{\rm per}}^{2}(0,T)\to L_{{\rm per}}^{2}(0,T)
\]
is invertible if $T\neq2\pi n$, $n\in\N$.

Substituting ${\bf u}={\bf u}^{(0)}+{\bf w}$, where ${\bf u}^{(0)}$
is the limiting breather (\ref{limiting-breather}) to the discrete
Klein--Gordon equation (\ref{KGlattice}), we obtain the coupled system
of differential-difference equations
\begin{equation}
\label{system-ift1}
L_{e}w_{n}=\epsilon(u_{n+1}^{(0)}-2u_{n}^{(0)}+u_{n-1}^{(0)})+N_{n}({\bf w},\epsilon),\quad n\in S
\end{equation}
and
\begin{equation}
\label{system-ift2}
L_{0}w_{n}=\epsilon(u_{n+1}^{(0)}-2u_{n}^{(0)}+u_{n-1}^{(0)})+N_{n}({\bf w},\epsilon),\quad n\in\Z\backslash S,
\end{equation}
where
\[
N_{n}({\bf w},\epsilon) = \epsilon(w_{n+1}-2w_{n}+w_{n-1}) + Q_n(w_n)
\]
and
\[
Q_n(w_n) = V'(u_{n}^{(0)})+V''(u_{n}^{(0)})w_{n}-V'(u_{n}^{(0)}+w_{n}).
\]
Since $V \in C^{\infty}(\R)$, the nonlinear function $Q_n(w_n)$ is $C^{\infty}$ for all $w_n \in \R$.
Because the discrete Laplacian is a bounded operator from $l^2(\Z)$ to $l^2(\Z)$
and $H_{\rm per}^{2}(0,T)$ forms a Banach algebra with respect to pointwise multiplication, we
conclude that the vector field ${\bf N}({\bf w},\epsilon) : l^{2}(\mathbb{Z},H_{{\rm per}}^{2}(0,T))
\times \R \to l^{2}(\mathbb{Z},H_{{\rm per}}^{2}(0,T))$ is a $C^{\infty}$ map.
Moreover, for any ${\bf w}$ in a ball in $l^{2}(\mathbb{Z},H_{{\rm per}}^{2}(0,T))$
centered at $0$ with radius $\delta>0$, there are constants $C_{\delta},D_{\delta}>0$
such that
\[
\|{\bf N}({\bf w},\epsilon)\|_{l^{2}(\mathbb{Z},H_{{\rm per}}^{2}(0,T))}\leq C_{\delta}\left(\epsilon\|{\bf w}\|_{l^{2}(\mathbb{Z},H_{{\rm per}}^{2}(0,T))}+\|{\bf w}\|_{l^{2}(\mathbb{Z},H_{{\rm per}}^{2}(0,T))}^{2}\right)
\]
and
\begin{eqnarray*}
\|{\bf N}({\bf w}_{1},\epsilon)-{\bf N}({\bf w}_{2},\epsilon)\|_{l^{2}(\mathbb{Z},H_{{\rm per}}^{2}(0,T))}\leq D_{\delta}\left(\epsilon+\|{\bf w}_{1}\|_{l^{2}(\mathbb{Z},H_{{\rm per}}^{2}(0,T))}+\|{\bf w}_{2}\|_{l^{2}(\mathbb{Z},H_{{\rm per}}^{2}(0,T))}\right)\\
\times\|{\bf w}_{1}-{\bf w}_{2}\|_{l^{2}(\mathbb{Z},H_{{\rm per}}^{2}(0,T))}.
\end{eqnarray*}
Thanks to the invertibility of the linearized operators $L_{e}$
and $L_{0}$ on $L_{e}^{2}$, the result of the theorem follows from
the Implicit Function Theorem (Theorem 4.E in \cite{Zeidler})
and the map $\R\ni\epsilon\mapsto{\bf u}^{(\epsilon)}\in l^{2}(\mathbb{Z},H_e^{2}(0,T))$
is $C^{\infty}$ for all small $\epsilon$.
\end{proof}

\begin{remark}
Although persistence of other breather configurations,
where oscillators are neither in-phase nor anti-phase, can not be
apriori excluded, we restrict our studies to the most important and
physically relevant breather configurations covered by Theorem \ref{theorem-continuum}.
\end{remark}

We shall now introduce the concept of the fundamental breather for
the set $S=\{0\}$. Multi-site breathers for small $\epsilon>0$ can
be approximated by the superposition of fundamental breathers at a
generic set $S$ of excited sites up to and including the order, at which
the tail-to-tail interactions of these breathers occur.

\begin{definition}
Let ${\bf u}^{(\epsilon)}\in l^{2}(\mathbb{Z},H_e^{2}(0,T))$
be the solution of the discrete Klein--Gordon equation (\ref{KGlattice})
for small $\epsilon>0$ defined by Theorem \ref{theorem-continuum}
for a given ${\bf u}^{(0)}(t)=\varphi(t){\bf e}_{0}$. This solution
is called the fundamental breather and we denote it by $\mbox{\boldmath\ensuremath{\phi}}^{(\epsilon)}$.
\label{definition-breather}
\end{definition}

By Theorem \ref{theorem-continuum}, we can use the Taylor series
expansion,
\begin{equation}
\mbox{\boldmath\ensuremath{\phi}}^{(\epsilon)}=\mbox{\boldmath\ensuremath{\phi}}^{(\epsilon,N)}+{\cal O}_{l^{2}(\mathbb{Z},H_{{\rm per}}^{2}(0,T))}(\epsilon^{N+1}), \quad
\mbox{\boldmath\ensuremath{\phi}}^{(\epsilon,N)} =
\sum_{k=0}^{N}\frac{\epsilon^{k}}{k!}\frac{d^{k}}{d\epsilon^{k}} \mbox{\boldmath\ensuremath{\phi}}^{(\epsilon)}\biggr|_{\epsilon=0},
\label{decomposition-phi}
\end{equation}
up to any integer $N\geq0$. Thanks to the discrete translational
invariance of the lattice, the fundamental breather can be centered
at any site $j\in\Z$. Let $\tau_{j}:l^{2}(\Z)\to l^{2}(\Z)$ be the
shift operator defined by
\[
(\tau_{j}{\bf u})_{n}=u_{n-j},\quad n\in\Z.
\]
If $\mbox{\boldmath\ensuremath{\phi}}^{(\epsilon)}$ is centered
at site $0$, then $\tau_{j}\mbox{\boldmath\ensuremath{\phi}}^{(\epsilon)}$
is centered at site $j\in\Z$. The simplest multi-site breather is
given by the two excited nodes at $j\in\Z$ and $k\in\Z$ with $j\neq k$.

\begin{lemma}
Let ${\bf u}^{(0)}(t)=\sigma_{j}\varphi(t){\bf e}_{j}+\sigma_{k}\varphi(t){\bf e}_{k}$
with $j\neq k$ and $N=|j-k|\geq1$. Let ${\bf u}^{(\epsilon)}\in l^{2}(\mathbb{Z},H_e^{2}(0,T))$
be the corresponding solution of the discrete Klein--Gordon equation
(\ref{KGlattice}) for small $\epsilon>0$ defined by Theorem \ref{theorem-continuum}.
Let $\{\varphi_{m}\}_{m=1}^{N}\in H_{e}^{2}(0,T)$ be defined recursively
by
\begin{equation}
L_{0}\varphi_{m}:=(\partial_{t}^{2}+1)\varphi_{m}=\varphi_{m-1},\quad m=1,2,...,N,
\label{inhomogen-1}
\end{equation}
starting with $\varphi_{0}=\varphi$, and let $\psi_{N} \in H_e^2(0,T)$ be defined
by
\begin{equation}
L_{e}\psi_{N}:=(\partial_{t}^{2}+V''(\varphi(t)))\psi_{N}=\varphi_{N-1}.
\label{inhomogen-2}
\end{equation}
Then, we have
\begin{equation}
{\bf u}^{(\epsilon)}=\sigma_{j}\tau_{j}\mbox{\boldmath\ensuremath{\phi}}^{(\epsilon,N)}+
\sigma_{k}\tau_{k}\mbox{\boldmath\ensuremath{\phi}}^{(\epsilon,N)}+
\epsilon^{N}\left(\sigma_{j}{\bf e}_{k}+\sigma_{k}{\bf e}_{j}\right)(\psi_{N}-\varphi_{N})+
{\cal O}_{l^{2}(\mathbb{Z},H_{{\rm per}}^{2}(0,T))}(\epsilon^{N+1}).\label{expansion-u}
\end{equation}
\label{lemma-tail-interaction}
\end{lemma}

\begin{proof}
By Theorem \ref{theorem-continuum}, the limiting configuration
${\bf u}^{(0)}(t)=\sigma_{j}\varphi(t){\bf e}_{j}+\sigma_{k}\varphi(t){\bf e}_{k}$
with two excited sites generates a $C^{\infty}$ map, which can be
expanded up to the $N+1$-order,
\begin{equation}
{\bf u}^{(\epsilon)}=\sum_{k=0}^{N}\frac{\epsilon^{k}}{k!}\frac{d^{k}}{d\epsilon^{k}}{\bf u}^{(\epsilon)}\biggr|_{\epsilon=0}+{\cal O}_{l^{2}(\mathbb{Z},H_{{\rm per}}^{2}(0,T))}(\epsilon^{N+1}).
\label{decomposition-u}
\end{equation}
Substituting (\ref{decomposition-u}) into (\ref{KGlattice}) generates
a sequence of equations at each order of $\epsilon$, which we consider
up to and including the terms of order $N$.

The central excited site at $n=0$ in the fundamental
breather $\mbox{\boldmath\ensuremath{\phi}}^{(\epsilon)}$
generates fluxes, which reach sites $n=\pm m$ at the $m$-th order.
Because $\mbox{\boldmath\ensuremath{\phi}}^{(\epsilon,m)}$
is compactly supported on $\{-m,-m+1,...,m\}$ and all sites with
$n\neq0$ contain no oscillations at the $0$-th order, we have
\begin{equation}
\mbox{\boldmath\ensuremath{\phi}}_{\pm m}^{(\epsilon,m)}=\epsilon^{m}\varphi_{m},
\label{fluxes}
\end{equation}
where $\{\varphi_{m}\}_{m=1}^{N}\in H_{e}^{2}(0,T)$ are computed from
the linear inhomogeneous equations (\ref{inhomogen-1}) starting with
$\varphi_{0}=\varphi$. Note that equations (\ref{inhomogen-1}) are
uniquely solvable because $T\neq2\pi n$, $n\in\N$.

For definiteness, let us assume that $j=0$ and $k=N\geq1$. The fluxes
from the excited sites $n=0$ and $n=N$ meet at the $N/2$-th order
at the middle site $n=N/2$ if $N$ is even or they overlap at the
$(N+1)/2$-th order at the two sites $n=(N-1)/2$ and $n=(N+1)/2$
if $N$ is odd. In either case, because of the expansion (\ref{potential-expansion}),
the nonlinear superposition of these fluxes affects terms at the order $3N/2$-th
or $3(N+1)/2$-th orders, that is, beyond the $N$-th order of the
expansion (\ref{expansion-u}). Therefore, the nonlinear superposition
of fluxes in higher orders of $\epsilon$ will definitely be beyond
the $N$-th order of the expansion (\ref{expansion-u}).

Up to the $N$-th order, all correction terms are combined together
as a sum of correction terms from the decomposition (\ref{decomposition-phi})
centered at the $j$-th and $k$-th sites, that is, we have
\begin{equation}
\label{superposition-u}
{\bf u}(\epsilon)=\sigma_{j}\tau_{j} \mbox{\boldmath\ensuremath{\phi}}^{(\epsilon,N-1)}(\epsilon) + \sigma_{k}\tau_{k}\mbox{\boldmath\ensuremath{\phi}}^{(\epsilon,N-1)}(\epsilon)+
{\cal O}_{l^{2}(\mathbb{Z},H_{{\rm per}}^{2}(0,T))}(\epsilon^{N}).
\end{equation}
At the $N$-th order, the flux from $j$-th site arrives to the $k$-th
site and vice versa. Therefore, besides the $N$-th order correction
terms from the decomposition (\ref{decomposition-phi}), we have additional
terms $\epsilon^{N}\left(\sigma_{j}{\bf e}_{k}+\sigma_{k}{\bf e}_{j}\right) \psi_N$
at the sites $n=j$ and $n=k$. Thanks to the linear superposition
principle, these additional terms are given by solutions of the inhomogeneous
equations (\ref{inhomogen-2}), which are uniquely solvable in $H_{e}^{2}(0,T)$
because $T'(E)\neq 0$. We also have to subtract
$\epsilon^{N}\left(\sigma_{j}{\bf e}_{k}+\sigma_{k}{\bf e}_{j}\right) \varphi_N$
from the $N$-th order of $\sigma_{j}\tau_{j}\mbox{\boldmath\ensuremath{\phi}}^{(\epsilon,N)}+
\sigma_{k}\tau_{k}\mbox{\boldmath\ensuremath{\phi}}^{(\epsilon,N)}$,
because these terms were computed under the assumption that the $k$-th site
contained no oscillations at the order $0$ for $\sigma_j \tau_j \mbox{\boldmath $\phi$}^{(\epsilon,N)}$
and vice versa. Combined all together, the expansion (\ref{expansion-u})
is justified up to terms of the $N$-th order.
\end{proof}

\section{Stability of multi-site breathers near the anti-continuum limit}
\label{sec:stability}

Let ${\bf u}\in l^{2}(\mathbb{Z},H_e^{2}(0,T))$ be a multi-site
breather in Theorem \ref{theorem-continuum} and $\epsilon>0$ be
a small parameter of the discrete Klein--Gordon equation (\ref{KGlattice}).
When we study stability of breathers, we understand the spectral stability,
which is associated with the linearization of the discrete Klein--Gordon
equation (\ref{KGlattice}) by using a perturbation ${\bf w}(t)$
in the decomposition ${\bf u}(t)+{\bf w}(t)$. Neglecting quadratic
and higher-order terms in ${\bf w}$, we obtain the linearized discrete
Klein--Gordon equation,
\begin{equation}
\ddot{w}_{n}+V''(u_{n})w_{n}=\epsilon(w_{n+1}-2w_{n}+w_{n-1}),\quad n\in\mathbb{Z}.
\label{KGlattice-lin}
\end{equation}
Because ${\bf u}(t+T)={\bf u}(t)$, an infinite-dimensional analogue
of the Floquet theorem applies and the Floquet monodromy
matrix ${\cal M}$ is defined by ${\bf w}(T)={\cal M}{\bf w}(0)$. We say that
the breather is stable if all eigenvalues of ${\cal M}$, called Floquet
multipliers, are located on the unit circle and it is unstable if
there is at least one Floquet multiplier outside the unit disk. Because
the linearized system (\ref{KGlattice-lin}) is Hamiltonian,
Floquet multipliers come in pairs $\mu_{1}$ and $\mu_{2}$ with $\mu_{1}\mu_{2}=1$.

For $\epsilon=0$, Floquet multipliers can be computed explicitly
because ${\cal M}$ is decoupled into a diagonal combination of $2$-by-$2$
matrices $\{M_{n}\}_{n\in\Z}$, which are computed from solutions
of the linearized equations
\begin{equation}
\ddot{w}_{n}+w_{n}=0,\quad n\in\mathbb{Z}\backslash S
\label{KGlattice-lin-1}
\end{equation}
and
\begin{equation}
\ddot{w}_{n}+V''(\varphi)w_{n}=0,\quad n\in S.
\label{KGlattice-lin-2}
\end{equation}

The first problem (\ref{KGlattice-lin-1}) admits the exact solution,
\begin{equation}
w_{n}(t)=w_{n}(0)\cos(t)+\dot{w}_{n}(0)\sin(t)\quad\Rightarrow\quad M_{n}=\left[\begin{array}{cc}
\cos(T) & \sin(T)\\
-\sin(T) & \cos(T)\end{array}\right],\quad n\in\mathbb{Z}\backslash S.
\label{exact-solution-1}
\end{equation}
Each $M_{n}$ for $n\in\mathbb{Z}\backslash S$ has two Floquet multipliers
at $\mu_{1}=e^{iT}$ and $\mu_{2}=e^{-iT}$. If $T\neq2\pi n$, $n \in \N$,
the Floquet multipliers $\mu_1$ and $\mu_2$ are located on the unit circle bounded away
from the point $\mu = 1$.

The second problem (\ref{KGlattice-lin-2}) also admits the exact
solution,
\begin{equation}
w_{n}(t)=\frac{\dot{w}_{n}(0)}{\ddot{\varphi}(0)}\dot{\varphi}(t) + \frac{w_{n}(0)}{\partial_{E}\varphi(0)}\partial_{E}\varphi(t),\quad n\in S,
\label{exact-solution-2}
\end{equation}
where $\varphi(t)$ is a solution of the nonlinear oscillator equation
(\ref{nonlinear-oscillator}) with the initial condition (\ref{ic-oscillator}).
Using identities (\ref{id-oscillator1})--(\ref{id-oscillator3}), we obtain,
\[
M_{n}=\left[\begin{array}{cc} 1 & 0\\
T'(E) [V'(a)]^{2} & 1\end{array}\right],\quad n\in S.
\]
Note that $V'(a)\neq0$ (or $T$ is infinite). If $T'(E)\neq0$,
each $M_{n}$ for $n\in S$ has the Floquet multiplier $\mu = 1$
of geometric multiplicity one and algebraic multiplicity two.

We conclude that if $T\neq2\pi n$, $n\in\N$ and $T'(E)\neq0$, the
limiting multi-site breather (\ref{limiting-breather}) at the anti-continuum
limit $\epsilon=0$ has an infinite number of semi-simple Floquet
multipliers at $\mu_{1}=e^{iT}$ and $\mu_{2}=e^{-iT}$ bounded away
from the Floquet multiplier $\mu = 1$ of algebraic multiplicity $2|S|$
and geometric multiplicity $|S|$.

Semi-simple multipliers on the unit circle are structurally stable
in Hamiltonian dynamical systems (Chapter III in \cite{Yakub}).
Under perturbations in the Hamiltonian, Floquet multipliers of the
same Krein signature do not move off the unit circle unless they coalesce with
Floquet multipliers of the opposite Krein signature \cite{Bridges}.
Therefore, the instability of the multi-site breather may only arise
from the splitting of the Floquet multiplier $\mu = 1$ of algebraic multiplicity $2|S|$
for $\epsilon \neq 0$.

To consider Floquet multipliers, we can introduce the characteristic
exponent $\lambda$ in the decomposition ${\bf w}(t)={\bf W}(t)e^{\lambda t}$.
If $\mu=e^{\lambda T}$ is the Floquet multiplier of the monodromy
operator ${\cal M}$, then ${\bf W}\in l^{2}(\Z,H_{{\rm per}}^{2}(0,T))$
is a solution of the eigenvalue problem,
\begin{equation}
\ddot{W}_{n}+V''(u_{n})W_{n}+2\lambda\dot{W}_{n}+\lambda^{2}W_{n}=\epsilon(W_{n+1}-2W_{n}+W_{n-1}),\quad n\in\mathbb{Z}.
\label{KGlattice-BVP}
\end{equation}
In particular, Floquet multiplier $\mu=1$ corresponds to the characteristic
exponent $\lambda=0$. The generalized eigenvector ${\bf Z} \in l^{2}(\Z,H_{{\rm per}}^{2}(0,T))$
of the eigenvalue problem (\ref{KGlattice-BVP}) for $\lambda=0$ solves the inhomogeneous problem,
\begin{equation}
\ddot{Z}_{n}+V''(u_{n})Z_{n}=\epsilon(Z_{n+1}-2Z_{n}+Z_{n-1})-2\dot{W}_{n},\quad n\in\mathbb{Z},
\label{KGlattice-BVP-generalized}
\end{equation}
where ${\bf W}$ is the eigenvector of (\ref{KGlattice-BVP}) for
$\lambda=0$. To normalize ${\bf Z}$ uniquely, we add a constraint that ${\bf Z}$
is orthogonal to ${\bf W}$ with respect to the inner product
$$
\langle {\bf W}, {\bf Z} \rangle_{l^2(\mathbb{Z},L^2_{\rm per}(0,T))} :=
\sum_{n \in \Z} \int_0^T \bar{W}_n(t) Z_n(t) dt.
$$

At $\epsilon=0$, the eigenvector ${\bf W}$ of the eigenvalue problem
(\ref{KGlattice-BVP}) for $\lambda=0$ is spanned by the linear combination
of $|S|$ fundamental solutions,
\begin{equation}
{\bf W}^{(0)}(t)=\sum_{k\in S}c_{k}\dot{\varphi}(t){\bf e}_{k}.
\label{limiting-eigenvector}
\end{equation}
The generalized eigenvector ${\bf Z}$ is spanned
by the linear combination of $|S|$ fundamental solutions,
\begin{equation}
{\bf Z}^{(0)}(t)=-\sum_{k\in S}c_{k} v(t) {\bf e}_{k}, \quad
v := 2 L_{e}^{-1}\ddot{\varphi},
\label{limiting-generalized-eigenvector}
\end{equation}
where $L_{e}=\partial_{t}^{2}+V''(\varphi(t)):H_{e}^{2}(0,T)\to L_{e}^{2}(0,T)$
is invertible and $\ddot{\varphi}\in L_{e}^{2}(0,T)$ (see the proof of
Theorem \ref{theorem-continuum}). Note that $\langle \dot{\varphi}, v \rangle_{L^2_{\rm per}(0,T)} = 0$
because $\dot{\varphi}$ is odd and $v$ is even in $t$.

Because of the translational invariance in $t$, we note that if ${\bf u}=\mbox{\boldmath\ensuremath{\phi}}^{(\epsilon)}$
is the fundamental breather in Definition \ref{definition-breather},
then ${\bf W} = \partial_{t} \mbox{\boldmath\ensuremath{\phi}}^{(\epsilon)} \equiv \mbox{\boldmath\ensuremath{\theta}}^{(\epsilon)} \in l^{2}(\Z,H_{{\rm per}}^{2}(0,T))$ is the eigenvector
of the eigenvalue problem (\ref{KGlattice-BVP})
for $\lambda=0$ and small $\epsilon>0$ and there exists a generalized
eigenvector ${\bf Z}\equiv\mbox{\boldmath\ensuremath{\mu}}^{(\epsilon)} \in l^{2}(\Z,H_{{\rm per}}^{2}(0,T))$
of the inhomogeneous problem (\ref{KGlattice-BVP-generalized}), which
exists because $\partial_{t}\mbox{\boldmath\ensuremath{\theta}}^{(\epsilon)}$
has the opposite parity in $t$ compared to $\mbox{\boldmath\ensuremath{\theta}}^{(\epsilon)}$.

By Taylor series expansions (\ref{decomposition-phi}), for any integer $N \geq 0$,
we have
\begin{equation}
\mbox{\boldmath\ensuremath{\theta}}^{(\epsilon)} =
\mbox{\boldmath\ensuremath{\theta}}^{(\epsilon,N)} +
{\cal O}_{l^{2}(\mathbb{Z},H_{{\rm per}}^{2}(0,T))}(\epsilon^{N+1}), \quad
\mbox{\boldmath\ensuremath{\mu}}^{(\epsilon)} =
\mbox{\boldmath\ensuremath{\mu}}^{(\epsilon,N)} +
{\cal O}_{l^{2}(\mathbb{Z},H_{{\rm per}}^{2}(0,T))}(\epsilon^{N+1}),
\end{equation}
where $\mbox{\boldmath\ensuremath{\theta}}^{(\epsilon,N)}$ and
$\mbox{\boldmath\ensuremath{\mu}}^{(\epsilon,N)}$ are polynomials in $\epsilon$ of
degree $N$. It follows from equations (\ref{limiting-eigenvector}) and (\ref{limiting-generalized-eigenvector})
that
\begin{equation}
\mbox{\boldmath\ensuremath{\theta}}^{(0)}=\dot{\varphi}(t){\bf e}_{0},\quad\mbox{\boldmath\ensuremath{\mu}}^{(0)}=-v(t){\bf e}_{0}.
\end{equation}

This formalism sets up the scene for the perturbation theory, which
is used to prove the main result on spectral stability of multi-site
breathers. We start with a simple multi-site breather configuration
with equal distances between excited sites and then upgrade this result
to multi-site breathers with non-equal distances between excited sites.

\begin{lemma}
Under assumptions of Theorem \ref{theorem-continuum},
let ${\bf u}^{(0)}(t)=\sum_{j=1}^{M}\sigma_{j}\varphi(t){\bf e}_{jN}$
with fixed $M,N\in\N$ and ${\bf u}^{(\epsilon)}\in l^{2}(\mathbb{Z},H_e^{2}(0,T))$
be the corresponding solution of the discrete Klein--Gordon equation
(\ref{KGlattice}) for small $\epsilon>0$ defined by Theorem \ref{theorem-continuum}.
Let $\{\varphi_{m}\}_{m=0}^{N}$ be defined by Lemma \ref{lemma-tail-interaction}
starting with $\varphi_{0}=\varphi$. Then the eigenvalue problem (\ref{KGlattice-BVP})
for small $\epsilon>0$ has $2M$ small eigenvalues,
\[
\lambda=\epsilon^{N/2}\Lambda+{\cal O}(\epsilon^{(N+1)/2}),
\]
where $\Lambda$ is an eigenvalue of the matrix eigenvalue problem
\begin{equation}
-\frac{T^{2}(E)}{T'(E)}\Lambda^{2}{\bf c}=K_{N}\mathcal{S}{\bf c},\quad{\bf c}\in\C^{M}.
\label{reduced-eigenvalue}
\end{equation}
Here the numerical coefficient $K_N$ is given by
\[
K_{N}=\int_{0}^{T}\dot{\varphi}\dot{\varphi}_{N-1}dt
\]
and the matrix $\mathcal{S}\in\mathbb{M}^{M\times M}$ is given by
\[
\mathcal{S} = \left[ \begin{array}{ccccccc} -\sigma_1 \sigma_2 & 1 & 0 & \ldots & 0 & 0 \\
1 & -\sigma_2(\sigma_1+\sigma_3) & 1 & \ldots & 0 & 0 \\
0 & 1 & -\sigma_3(\sigma_2+\sigma_4) & \ldots & 0 & 0 \\
\vdots & \vdots & \vdots & \ddots & \vdots & \vdots \\
0 & 0 & 0 & \ldots & -\sigma_{M-1} (\sigma_{M-2}+\sigma_M) & 1 \\
0 & 0 & 0 & \ldots & 0 & -\sigma_{M} \sigma_{M-1} \end{array}
\right].
\]
\label{lemma-eigenvalue}
\end{lemma}

\begin{proof}
At $\epsilon=0$, the eigenvalue problem (\ref{KGlattice-BVP})
admits eigenvalue $\lambda=0$ of geometric multiplicity $M$ and
algebraic multiplicity $2M$, which is isolated from the rest of the
spectrum. Perturbation theory in $\epsilon$ applies thanks to the
smoothness of ${\bf u}^{(\epsilon)}$ in $\epsilon$ and $V'$ in
$u$. Perturbation expansions (so-called Puiseux series, see
Chapter 2 in \cite{Kato} and recent work \cite{Welters})
are smooth in powers of $\epsilon^{1/2}$ thanks
to the Jordan block decomposition at $\epsilon=0$.

We need to find out how the eigenvalue $\lambda = 0$
of algebraic multiplicity $2M$ split for small $\epsilon > 0$.
Therefore, we are looking for the eigenvectors of the eigenvalue
problem (\ref{KGlattice-BVP}) in the subspace associated with
the eigenvalue $\lambda=0$ using the substitution
$\lambda=\epsilon^{N/2}\tilde{\lambda}$ and the decomposition
\begin{eqnarray}
\label{decomposition-W}
{\bf W} = \sum_{j=1}^{M} c_{j}  \left( \tau_{jN}
\mbox{\boldmath\ensuremath{\theta}}^{(\epsilon,N)} - \epsilon^N ({\bf e}_{(j-1)N} + {\bf e}_{(j+1)N})
\dot{\varphi}_N \right) + \epsilon^{N/2} \tilde{\lambda} \sum_{j=1}^{M} c_{j} \tau_{jN}
\mbox{\boldmath\ensuremath{\mu}}^{(\epsilon,N_*)} + \epsilon^{N}\tilde{{\bf W}},
\end{eqnarray}
where $N_* = N/2$ if $N$ is even and $N_* = (N-1)/2$ if $N$ is odd, whereas
$\tilde{\bf W}$ is the remainder term at the $N$-th order in $\epsilon$. The decomposition
formula (\ref{decomposition-W}) follows from the superposition (\ref{expansion-u})
up to the $N$-th order in $\epsilon$. The terms $\epsilon^N \sum_{j=1}^{M} c_{j}
({\bf e}_{(j-1)N} + {\bf e}_{(j+1)N}) \dot{\psi}_N$ from the superposition (\ref{expansion-u})
are to be accounted at the equation for $\tilde{\bf W}$. Note that our convention in
writing (\ref{decomposition-W}) is to drop the boundary terms with
${\bf e}_{0N}$ and ${\bf e}_{(M+1)N}$.

Substituting (\ref{decomposition-W})
to (\ref{KGlattice-BVP}), all equations
are satisfied up to the $N$-th order. At the $N$-th order,
we divide (\ref{KGlattice-BVP}) by $\epsilon^{N}$ and
collect equations at the excited sites $n=jN$ for $j\in\{1,2,...,M\}$,
\begin{eqnarray}
\ddot{\tilde{W}}_{jN}+V''(\varphi)\tilde{W}_{jN} & = & (c_{j+1} +c_{j-1})
\dot{\varphi}_{N-1} - \sigma_{j}(\sigma_{j+1} +\sigma_{j-1}) c_{j}
V'''(\varphi)\psi_{N}\dot{\varphi}\nonumber \\
& \phantom{t} & +\tilde{\lambda}^{2} c_j (2\dot{v}-\dot{\varphi})+\mathcal{O}(\epsilon^{1/2}),
\label{inhom-eq-W}
\end{eqnarray}
where we admit another convention that $\sigma_0 = \sigma_{M+1} = 0$ and $c_0 = c_{M+1} = 0$.
In the derivation of equations (\ref{inhom-eq-W}),
we have used the fact that the term $\dot{\varphi}_{N-1}$ comes from the fluxes from
$n=(j+1)N$ and $n=(j-1)N$ sites generated by the derivatives of
the linear inhomogeneous equations (\ref{inhomogen-1}) and the term
$\sigma_{j}(\sigma_{j+1} +\sigma_{j-1}) c_{j}  V'''(\varphi)\psi_{N}\dot{\varphi}$ comes from the expansion
(\ref{expansion-u}) of the nonlinear potential $V''(u_{jN})$.

Expanding $\tilde{\lambda}=\Lambda+{\cal O}(\epsilon^{1/2})$ and
projecting the system of linear inhomogeneous equations (\ref{inhom-eq-W})
to $\dot{\varphi}\in H_{{\rm per}}^{2}(0,T)$, the kernel of $L=\partial_{t}^{2}+V''(\varphi) :
H_{{\rm per}}^{2}(0,T) \to L_{{\rm per}}^{2}(0,T)$,
we obtain the system of difference equations,
\begin{eqnarray*}
\Lambda^{2}c_{j} \int_{0}^{T} \left(\dot{\varphi}^{2}+2 v \ddot{\varphi}\right)dt =
(c_{j+1} + c_{j-1}) \int_{0}^{T}\dot{\varphi}\dot{\varphi}_{N-1}dt
- \sigma_{j}(\sigma_{j+1} + \sigma_{j-1})c_{j}\int_{0}^{T}V'''(\varphi)\psi_{N}\dot{\varphi}^{2}dt,
\end{eqnarray*}
where the integration by parts is used to simplify the left-hand
side. Differentiating the linear inhomogeneous equation (\ref{inhomogen-2})
and projecting it to $\dot{\varphi}$, we infer that
\begin{eqnarray*}
\int_{0}^{T}V'''(\varphi)\psi_{N}\dot{\varphi}^{2}dt=\int_{0}^{T}\dot{\varphi}\dot{\varphi}_{N-1}dt\equiv K_{N}.
\end{eqnarray*}

The system of difference equations yields the matrix eigenvalue problem
(\ref{reduced-eigenvalue}) if we can verify that
\begin{eqnarray*}
\int_{0}^{T}\left(\dot{\varphi}^{2} + 2 v \ddot{\varphi}\right)dt=-\frac{T^{2}(E)}{T'(E)}.
\end{eqnarray*}
To do so, we note that $v\equiv2L_{e}^{-1}\ddot{\varphi}$ is even
in $t\in\R$, so that it is given by the exact solution,
\[
v(t)=t\dot{\varphi}(t)+C\partial_{E}\varphi(t),
\]
where $C\in\R$. From the condition of $T$-periodicity for $v(t)$
and $\dot{v}(t)$, we obtain
\[
v(0)=v(T)=Ca'(E),\quad\dot{v}(0)=0=\dot{v}(T)=T\ddot{\varphi}(0)-CT'(E)\ddot{\varphi}(0),
\]
hence $C=T(E)/T'(E)$ and
\begin{eqnarray*}
\int_{0}^{T}\left(\dot{\varphi}^{2}+2 v \ddot{\varphi}\right)dt & = &
2C\int_{0}^{T}\ddot{\varphi}\partial_{E}\varphi dt=-C\int_{0}^{T}
\left(\dot{\varphi}\partial_{E}\dot{\varphi} + V'(\varphi)\partial_{E}\varphi\right)dt\\
& = & -C\int_{0}^{T}\frac{\partial}{\partial E}\left(\frac{1}{2}\dot{\varphi}^{2}+
V(\varphi)\right)dt = -C T(E) =-\frac{T^{2}(E)}{T'(E)},
\end{eqnarray*}
where equation (\ref{nonlinear-oscillator}) has been used.
Finally, the matrix eigenvalue problem (\ref{reduced-eigenvalue})
defines $2 M$ small eigenvalues that bifurcate from $\lambda = 0$
for small $\epsilon > 0$. The proof of the lemma is complete.
\end{proof}

We shall now count eigenvalues of the matrix eigenvalue problem (\ref{reduced-eigenvalue})
to classify stable and unstable configurations of multi-site breathers
near the anti-continuum limit.

\begin{lemma}
Let $n_{0}$ be the numbers of negative elements in
the sequence $\{\sigma_{j}\sigma_{j+1}\}_{j=1}^{M-1}$. If $T'(E)K_{N}>0$,
the matrix eigenvalue problem (\ref{reduced-eigenvalue}) has exactly
$n_{0}$ pairs of purely imaginary eigenvalues $\Lambda$ and $M-1-n_{0}$
pairs of purely real eigenvalues $\mu$ counting their multiplicities,
in addition to the double zero eigenvalue. If $T'(E)K_{N}<0$, the
conclusion changes to the opposite.
\label{lemma-count}
\end{lemma}

\begin{proof}
We shall prove that the matrix $\mathcal{S}$ has exactly
$n_{0}$ positive and $M-1-n_{0}$ negative eigenvalues counting their
multiplicities, in addition to the simple zero eigenvalue. If this
is the case, the assertion of the lemma follows from the correspondence
$\Lambda^{2}=-\frac{T'(E)K_{N}}{T^{2}(E)}\gamma$, where $\gamma$
is an eigenvalue of $\mathcal{S}$.

Setting $c_{j}=\sigma_{j}b_{j}$, we rewrite the eigenvalue problem
$\mathcal{S}{\bf c}=\gamma{\bf c}$ as the difference equation,
\begin{equation}
\sigma_{j}\sigma_{j+1}(b_{j+1}-b_{j})+\sigma_{j}\sigma_{j-1}(b_{j-1}-b_{j})=\gamma b_{j},
\quad j\in\{1,2,...,M\},
\label{difference-eigenvalue}
\end{equation}
subject to the conditions $\sigma_{0}=\sigma_{M+1}=0$. Therefore,
$\gamma=0$ is always an eigenvalue with the eigenvector ${\bf b} = [1,1,...,1] \in \R^M$.
The coefficient matrix in (\ref{difference-eigenvalue}) coincides
with the one analyzed by Sandstede in Lemma 5.4 and Appendix C \cite{Sand}.
This correspondence yields the assertion on the number of eigenvalues
of $\mathcal{S}$.
\end{proof}

Before generalizing the results of Lemmas \ref{lemma-eigenvalue} and
\ref{lemma-count} to other multi-site breathers, we consider two
examples, which are related to the truncated potential (\ref{potential-numerics}).
We shall use the Fourier cosine series for the solution $\varphi\in H_{e}^{2}(0,T)$,
\begin{equation}
\varphi(t)=\sum_{n\in\N}c_{n}(T)\cos\left(\frac{2\pi nt}{T}\right),
\end{equation}
for some square summable set $\{c_{n}(T)\}_{n\in\N}$. Because of
the symmetry of $V$, we have $\varphi(T/4)=0$, which imply that
$c_{n}(T)\equiv0$ for all even $n\in\N$. Solving the linear inhomogeneous equations
(\ref{inhomogen-1}), we obtain
\begin{equation}
\varphi_{k}(t) = \sum_{n\in\N_{{\rm odd}}}
\frac{T^{2k} c_{n}(T)}{\left(T^{2}-4\pi^{2} n^{2}\right)^{k}}
\cos\left(\frac{2\pi nt}{T}\right),\quad k\in\N.
\end{equation}
Using Parseval's equality, we compute the constant $K_{N}$ in Lemma
\ref{lemma-eigenvalue},
\begin{equation}
K_{N}=\int_{0}^{T}\dot{\varphi}(t)\dot{\varphi}_{N-1}(t)dt =
4\pi^{2}\sum_{n\in\N_{{\rm odd}}}
\frac{T^{2N-3} n^{2}|c_{n}(T)|^{2}}{\left(T^{2}-4\pi^{2}n^{2} \right)^{N-1}}.
\label{formula-K-N}
\end{equation}

For the hard potential with $V'(u)=u+u^{3}$, we know from Figure
\ref{fig1} that the period $T(E)$ is a decreasing function of $E$
from $T(0)=2\pi$
to $\lim_{E\to\infty}T(E)=0$. Since $T'(E)<0$ and $T(E)<2\pi$,
we conclude that $T'(E)K_{N}<0$ if $N$ is odd and $T'(E)K_{N}>0$
if $N$ is even. By Lemma \ref{lemma-count}, if $N$ is odd, the
only stable configuration of the multi-site breathers is the one with
all equal $\{\sigma_{j}\}_{j=1}^{M}$ (in-phase breathers),
whereas if $N$ is even, the only stable configuration of the multi-site breathers
is the one with all alternating $\{\sigma_{j}\}_{j=1}^{M}$ (anti-phase breathers).
This conclusion is shown in the first line of Table I.

For the soft potential with $V'(u)=u-u^{3}$, we know from Figure
\ref{fig1} that the period $T(E)$ is an increasing function of $E$ from $T(0)=2\pi$
to $\lim_{E\to E_{0}}T(E)=\infty$, where $E_{0}=\frac{1}{4}$. If
$T(E)$ is close to $2\pi$, then the first positive term in the
series (\ref{formula-K-N}) dominates and $K_{N}>0$ for all $N\in\N$.
At the same time, $T'(E)>0$ and Lemma \ref{lemma-count} implies
that the only stable configuration of the multi-site breathers is
the one with all alternating $\{\sigma_{j}\}_{j=1}^{M}$ (anti-phase breathers).
The conclusion holds for any $T > 2\pi$ if $N$ is odd, because $K_N > 0$ in this case.

This precise conclusion is obtained in Theorem 3.6 of \cite{PKF1}
in the framework of the DNLS equation (\ref{NLSlattice}). It is also
in agreement with perturbative arguments in \cite{Archilla,KK09},
which are valid for $N=1$ (all excited sites are adjacent on the
lattice). To elaborate this point further, we show in Appendix \ref{sec-equivalence}
the equivalence between the matrix eigenvalue problem (\ref{reduced-eigenvalue})
with $N=1$ and the criteria used in \cite{Archilla,KK09}.

For even $N\in\N$, we observe a new phenomenon, which arise for the
soft potentials with larger values of $T(E)>2\pi$.
To be specific, we restrict our consideration of multisite breathers
with the period $T$ in the interval $(2\pi,6\pi)$. Similar
results can be obtained in the intervals $(6\pi,10\pi)$, $(10 \pi,14 \pi)$,
and so on. For even $N\in\N$, there exists a period $T_N \in (2\pi,6\pi)$ such that
the constant $K_{N}$ in (\ref{formula-K-N}) changes sign
from $K_N > 0$ for $T \in (2\pi,T_N)$ to $K_{N}<0$ for $T \in (T_N,6\pi)$.
When it happens, the conclusion on stability of the multi-site breather change directly
to the opposite: the only stable configuration of the multi-site breathers
is the one with all equal $\{\sigma_{j}\}_{j=1}^{M}$ (in-phase breathers).
This conclusion is shown in the second line of Table I.

\medskip{}

\begin{center}
\begin{tabular}{|>{\centering}m{25mm}|>{\centering}m{25mm}|>{\centering}m{50mm}|}
\hline
 & $N$ odd  & $N$ even\tabularnewline
\hline
hard potential

$V'(u)=u+u^{3}$

$0 < T < 2\pi$ & in-phase  & anti-phase\tabularnewline
\hline
soft potential

$V'(u)=u-u^{3}$

$2 \pi < T < 6 \pi$ & anti-phase & $2\pi<T<T_{N}$ anti-phase

$T_{N}<T<6\pi$ in-phase \tabularnewline
\hline
\end{tabular}
\par\end{center}

\textbf{Table I:} Stable multi-site breathers in hard and soft potentials.
The stability threshold $T_{N}$ corresponds to the zero of $K_{N}$ for $T \in (2\pi,6\pi)$.

\medskip{}

We conclude this section with the stability theorem for general multi-site
breathers. For the sake of clarity, we formulate the theorem in the
case when $T'(E)>0$ and all $K_{N}>0$, which arises
for the soft potential with odd $N$. Using Lemma \ref{lemma-count},
the count can be adjusted to the cases of $T'(E) < 0$ and/or $K_{N}<0$.

\begin{theorem}
Let $\{n_{j}\}_{j=1}^{M}\subset\Z$ be an increasing
sequence with $M\in\N$. Let ${\bf u}^{(\epsilon)}\in l^{2}(\mathbb{Z},H_e^{2}(0,T))$
be a solution of the discrete Klein--Gordon equation (\ref{KGlattice})
in Theorem \ref{theorem-continuum} with
${\bf u}^{(0)}(t)=\sum_{j=1}^{M}\sigma_{j}\varphi(t){\bf e}_{n_{j}}$
for small $\epsilon>0$. Let $\{\varphi_{m}\}_{m=0}^{\infty}$ be defined
by the linear equations (\ref{inhomogen-1}) starting with $\varphi_{0}=\varphi$.

Define $\{N_{j}\}_{j=1}^{M-1}$ and $\{K_{N_{j}}\}_{j=1}^{M-1}$ by
$N_{j}=n_{j+1}-n_{j}$ and $K_{N_{j}}=\int_{0}^{T}\dot{\varphi}\dot{\varphi}_{N_{j}-1}dt$.
Assume $T'(E)>0$ and $K_{N_{j}}>0$ for all $N_{j}$.

Let $n_{0}$ be the numbers of negative elements in the sequence $\{\sigma_{j}\sigma_{j+1}\}_{j=1}^{M-1}$.
The eigenvalue problem (\ref{KGlattice-BVP}) at the discrete breather
${\bf u}^{(\epsilon)}$ has exactly $n_{0}$ pairs of purely imaginary
eigenvalues $\lambda$ and $M-1-n_{0}$ pairs of purely real eigenvalues
$\lambda$ counting their multiplicities, in addition to the double
zero eigenvalue.
\label{theorem-eigenvalue}
\end{theorem}

\begin{proof}
The limiting configuration
${\bf u}^{(0)}(t)=\sum_{j=1}^{M}\sigma_{j}\varphi(t){\bf e}_{n_{j}}$
defines clusters of excited sites with equal distances
$N_{j}$ between the adjacent excited sites.

According to Lemma \ref{lemma-eigenvalue}, splitting of $M$ double
Jordan blocks associated to the decompositions (\ref{limiting-eigenvector})
and (\ref{limiting-generalized-eigenvector}) occurs into different
orders of the perturbation theory, which are determined by the set
$\{ N_j \}_{j =1}^{M-1}$. At each order of the perturbation
theory, the splitting of eigenvalues associated with one cluster with
equal distance between the adjacent excited sites
obeys the matrix eigenvalue problem (\ref{reduced-eigenvalue}),
which leaves exactly one double eigenvalue at zero and yields symmetric
pairs of purely real or purely imaginary eigenvalues, in accordance
to the count of Lemma \ref{lemma-count}.

The double zero eigenvalue corresponds to the eigenvector ${\bf W}$
and the generalized eigenvector ${\bf Z}$ generated by the translational
symmetry of the multi-site breather bifurcating from a particular
cluster of excited sites in the limiting configuration ${\bf u}^{(0)}$.
The splitting of the double zero eigenvalues associated with the
cluster happens at the higher orders in $\epsilon$, when the fluxes
from adjacent clusters reach each others. Since the end-point fluxes from
the multi-site breathers are equivalent to the fluxes (\ref{fluxes})
generated from the fundamental breathers, they still obey Lemma \ref{lemma-tail-interaction}
and the splitting of double zero eigenvalue associated with different
clusters still obeys Lemma \ref{lemma-eigenvalue}.

At the same time, the small pairs of real and imaginary eigenvalues
arising at a particular order in $\epsilon$ remain at the real and
imaginary axes in higher orders of the perturbation theory because
their geometric and algebraic multiplicity coincides, thanks to the
fact that these eigenvalues are related to the eigenvalues of the
symmetric matrix $\mathcal{S}$ in the matrix eigenvalue problem (\ref{reduced-eigenvalue}).

Avoiding lengthy algebraic proofs, these arguments yield the assertion
of the theorem.
\end{proof}

\section{Numerical results \label{sec:numerics}}

\label{sec-numerics}

We illustrate our analytical results on existence and stability of
discrete breathers near the anti-continuum limit by using numerical
approximations. The discrete Klein--Gordon equation (\ref{KGlattice})
can be truncated at a finite system of differential equations by applying
the Dirichlet conditions at the ends.

\subsection{Three-site model}

The simplest model which allows gaps in the initial configuration
$\mathbf{u}^{(0)}$ is the one restricted to three lattice sites,
e.g. $n\in\{-1,0,1\}$. We choose the soft potential $V'(u)=u-u^{3}$
and rewrite the truncated discrete Klein--Gordon equation as a system
of three Duffing oscillators with linear coupling terms,
\begin{equation}
\left\{ \begin{array}{l}
\ddot{u}_{0}+u_{0}-u_{0}^{3}=\epsilon(u_{1}-2u_{0}+u_{-1}),\\
\ddot{u}_{\pm1}+u_{\pm1}-u_{\pm1}^{3}=\epsilon(u_{0}-2u_{\pm1}).\end{array}\right.
\label{eq:dKG_trunc}
\end{equation}

A fast and accurate approach to construct $T$-periodic solutions
for this system is the shooting method. The idea is to find $\mathbf{a}\in\R^{3}$
such that the solution $\mathbf{u}(t)\in C^{1}(\R_{+},\R^{3})$ with
initial conditions $\mathbf{u}(0)=\mathbf{a}$, $\dot{\mathbf{u}}(0)=0$
satisfy the conditions of $T$-periodicity, $\mathbf{u}(T)={\bf a}$,
$\dot{\mathbf{u}}(T)=0$. However, these constraints would generate
an over-determined system of equations on ${\bf a}$. To set up the
square system of equations, we can use the symmetry $t\mapsto-t$
of system (\ref{eq:dKG_trunc}). If we add the constraint $\dot{\mathbf{u}}(T/2)=0$,
then even solutions of system (\ref{eq:dKG_trunc}) satisfy $\mathbf{u}(-T/2)=\mathbf{u}(T/2)$
and $\dot{\mathbf{u}}(-T/2)=-\dot{\mathbf{u}}(T/2)=0$, that is,
these solutions are $T$-periodic. Hence, the values of ${\bf a}\in\R^{3}$ become
the roots of the vector ${\bf F}({\bf a})=\dot{\mathbf{u}}(T/2)\in\R^{3}$.

We now construct a periodic solution $\mathbf{u}$ to system (\ref{eq:dKG_trunc})
that corresponds to the anti-continuum limit $\mathbf{u}^{(0)}$ as
follows. Using the initial data $\mathbf{u}^{(0)}(0)$ as an initial
guess for the shooting method for a fixed value of $T$, we continue
the initial displacement $\mathbf{u}(0)$ with respect to the coupling
constant $\epsilon>0$. After that, we fix a value of $\epsilon$
and use the shooting method again to continue the initial displacement
$\mathbf{u}(0)$ with respect to period $T$.

Let us apply this method to determine initial conditions for the fundamental
breather,
\begin{equation}
u_{0}^{(0)}=\varphi,\quad u_{\pm1}^{(0)}=0,
\label{fund-breat}
\end{equation}
and for a two-site breather with a hole,
\begin{equation}
u_{0}^{(0)}=0,\quad u_{\pm1}^{(0)}=\varphi.
\label{hole-breat}
\end{equation}
In both cases, we can use the symmetry $u_{-1}(t)=u_{1}(t)$ to reduce
the dimension of the shooting method to two unknowns $a_{0}$ and
$a_{1}$.

Figure \ref{fig:breather_branches} shows solution branches for these
two breathers on the period--amplitude plane by plotting $T$
versus $a_{0}$ and $a_{1}$ for $\epsilon = 0.01$. For $2\pi<T<6\pi$,
solution branches are close to the limiting solutions (dotted line),
in agreement with Theorem \ref{theorem-continuum}. However, a new
phenomenon emerges near $T=6\pi$: both breather solutions experience
a pitchfork bifurcation and two more solution branches split off the
main solution branch. The details of the pitchfork bifurcation for the
fundamental solution branch are shown on the insets of Figure \ref{fig:breather_branches}.

Let $T_S$ be the period at the point of the pitchfork bifurcation.
We may think intuitively that $T_S$ should approach to the point of
$1 : 3$ resonance for small $\epsilon$, that is, $T_S \to 6 \pi$ as $\epsilon \to 0$.
We have checked numerically that this conjecture is in fact false and
the value of $T_S$ gets larger as $\epsilon$ gets smaller.
This property of the pitchfork bifurcation is analyzed
in Section \ref{sec-resonance} below (see Remark \ref{remark-5} and Figure \ref{fig:delta_vs_eps}).

Figure \ref{fig:breather_branches} also shows two branches of solutions
for $T>6\pi$ with negative values of $a_{1}$ for positive values of $a_0$
and vice versa. One of
the two branches is close to the breathers at the anti-continuum limit,
as prescribed by Theorem \ref{theorem-continuum}. We note that the
breather solutions prescribed by Theorem \ref{theorem-continuum}
for $T<6\pi$ and $T>6\pi$ belong to different solution branches.
This property is also analyzed in Section \ref{sec-resonance} below
(see Remark \ref{remark-4} and Figure \ref{fig:a_vs_delta}).

Figure \ref{fig:solutions} shows the fundamental breather before
$(T=5\pi)$ and after $(T=5.8\pi)$ pitchfork bifurcation. The symmetry
condition $\mathbf{u}(T/4)=0$ for the solution at the main branch
is violated for two new solutions that bifurcate from the main branch.
Note that the two new solutions bifurcating for $T > T_S$ look different on
the graphs of $a_0$ and $a_1$ versus $T$. Nevertheless, these two solutions
are related to each other by the symmetry of the system (\ref{eq:dKG_trunc}).
If $\mathbf{u}(t)$ is one solution of the system (\ref{eq:dKG_trunc}),
then $-\mathbf{u}(t+T/2)$ is another solution of the same system.
If $\mathbf{u}(T/4) \neq 0$, then these two solutions are different from each other.

Let us now illustrate the stability result of Theorem \ref{theorem-eigenvalue}
using the fundamental breather (\ref{fund-breat}) and the breather
with a hole (\ref{hole-breat}). We draw a conclusion on linearized
stability of the breathers by testing whether their Floquet multipliers,
found from the monodromy matrix associated with the linearization
of system (\ref{eq:dKG_trunc}), stay on the unit circle.

Figure \ref{fig:FM_010} shows the real part of Floquet multipliers
versus the breather's period for the fundamental breather (left) and
the new solution branches (right) bifurcating from the fundamental breather
due to the pitchfork bifurcation. Because Floquet multipliers
are on the unit circle for all periods below the bifurcation value $T_S$,
the fundamental breather remains stable for these periods, in agreement
with Theorem \ref{theorem-eigenvalue}. Once the bifurcation occurs,
one of the Floquet multiplier becomes real and unstable (outside the
unit circle). Two new stable solutions appear during the bifurcation and
have the identical Floquet multipliers because of the aforementioned symmetry
between the new solutions. These solutions become
unstable for periods slightly larger than the bifurcation value $T_S$, because
of the period-doubling bifurcation associated with Floquet multipliers
at $-1$.

\begin{figure}
\begin{centering}
\begin{tabular}{cc}
\includegraphics[width=0.47\textwidth]{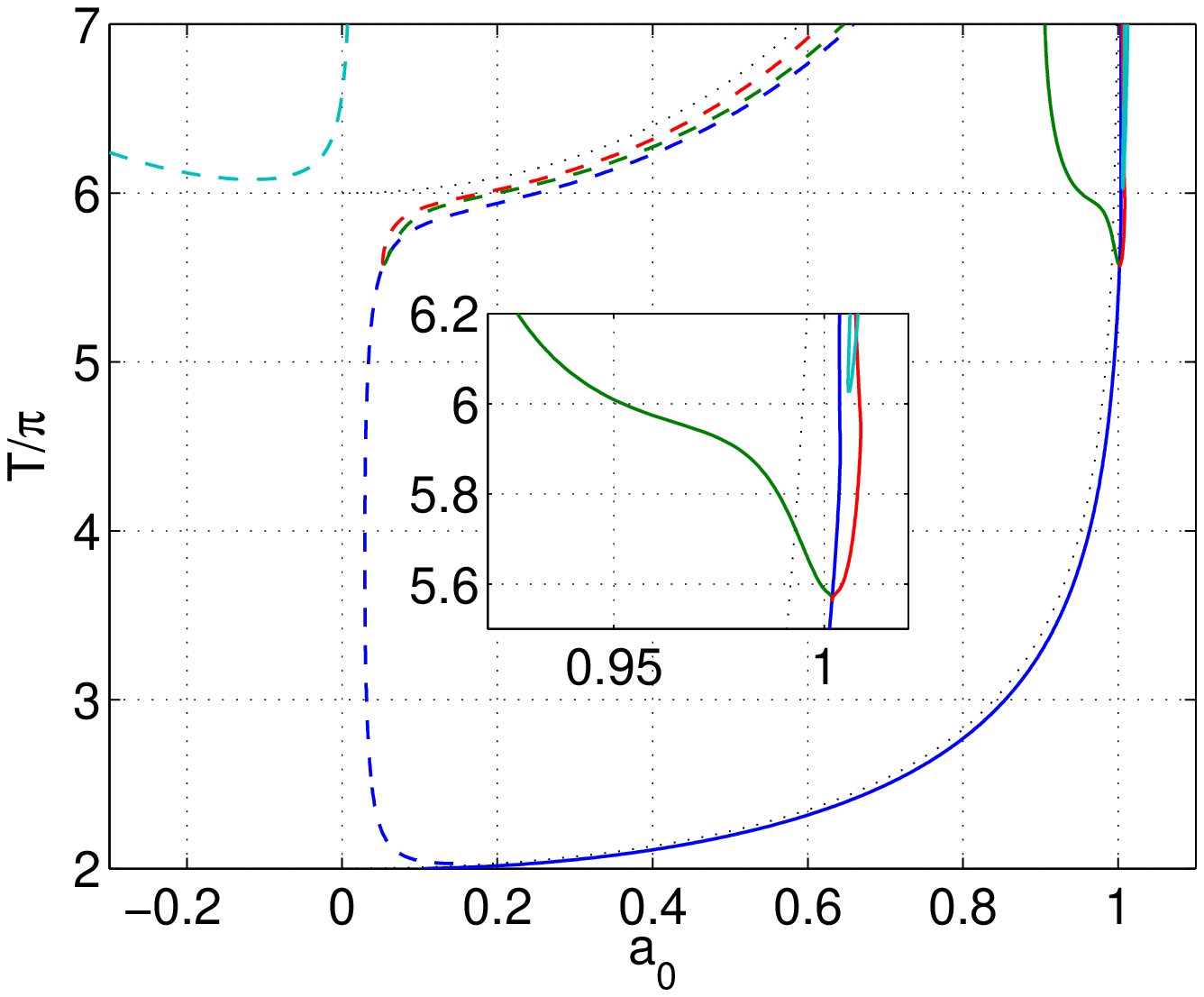}  & \includegraphics[width=0.47\textwidth]{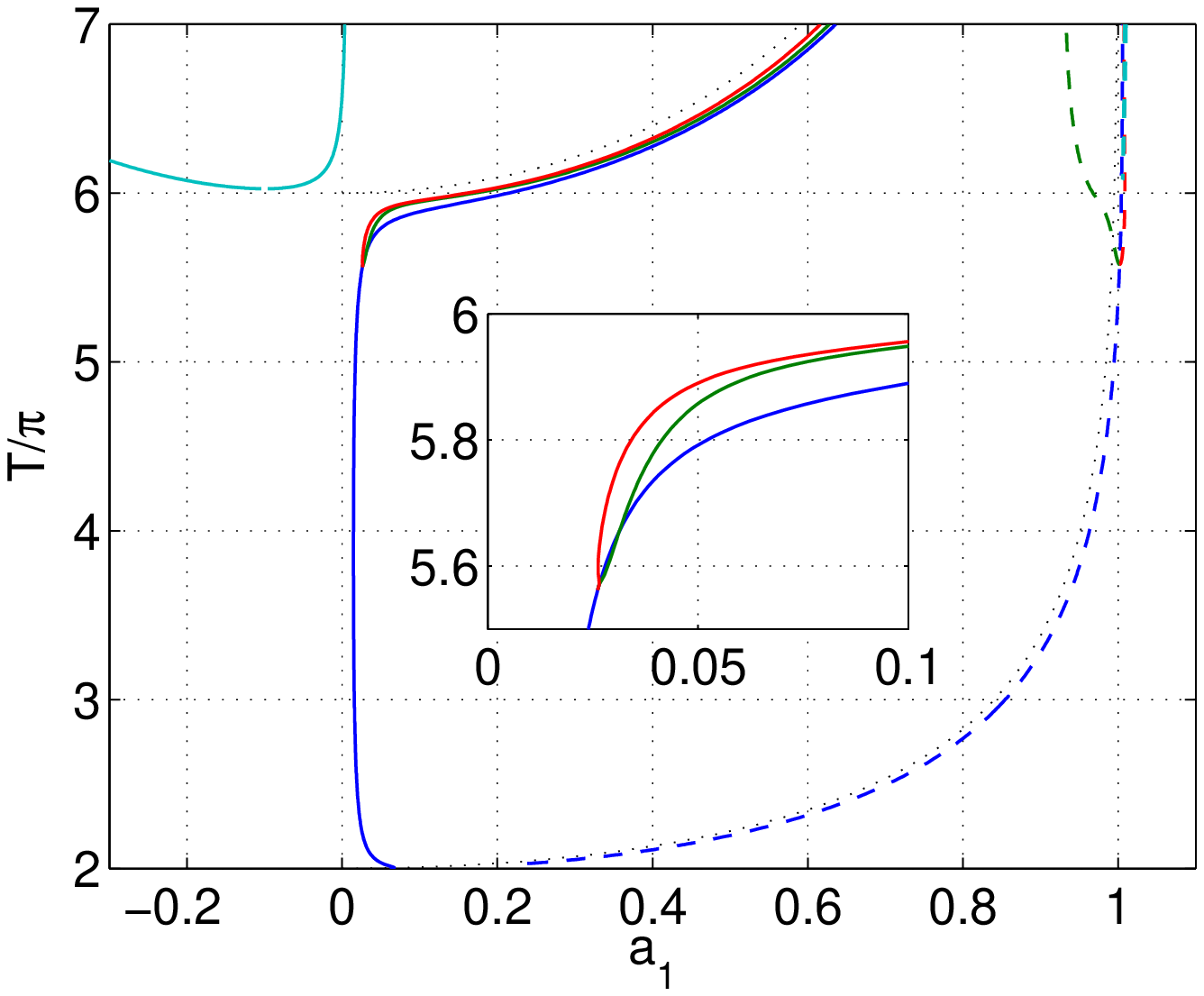}\tabularnewline
\end{tabular}
\par\end{centering}
\caption{The initial displacements $a_0$ and $a_1$ for the $T$-periodic solutions
to system (\ref{eq:dKG_trunc}) with $\epsilon=0.01$.
The solid and dashed lines correspond to the fundamental (\ref{fund-breat})
and two-site (\ref{hole-breat}) breathers respectively. The dotted lines correspond to the $T$-periodic
solutions to equation (\ref{nonlinear-oscillator}).
The insets show the pitchfork bifurcation of the fundamental breather.
\label{fig:breather_branches}}
\end{figure}
\begin{figure}
\begin{centering}
\begin{tabular}{cc}
\includegraphics[width=0.47\textwidth]{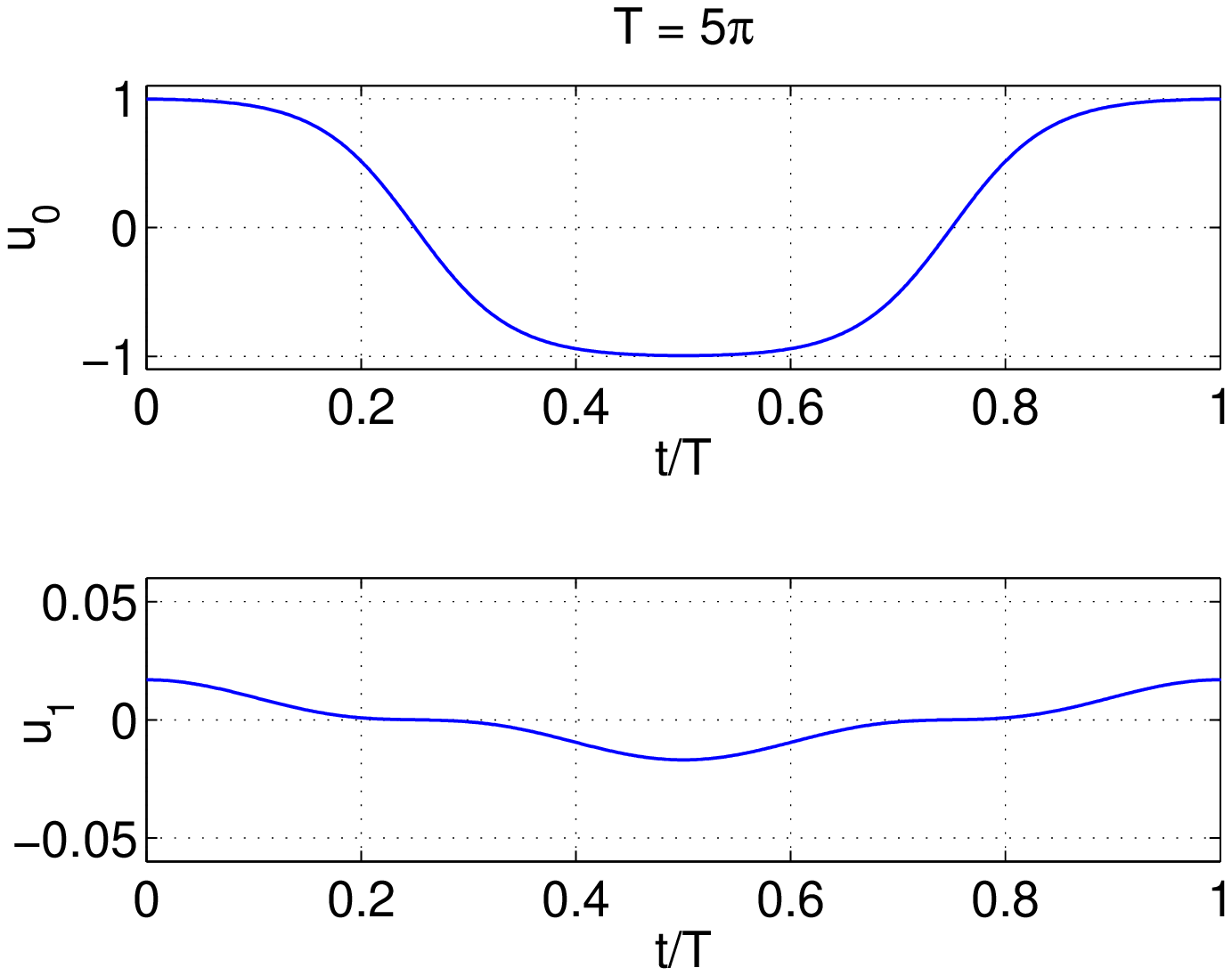}  & \includegraphics[width=0.47\textwidth]{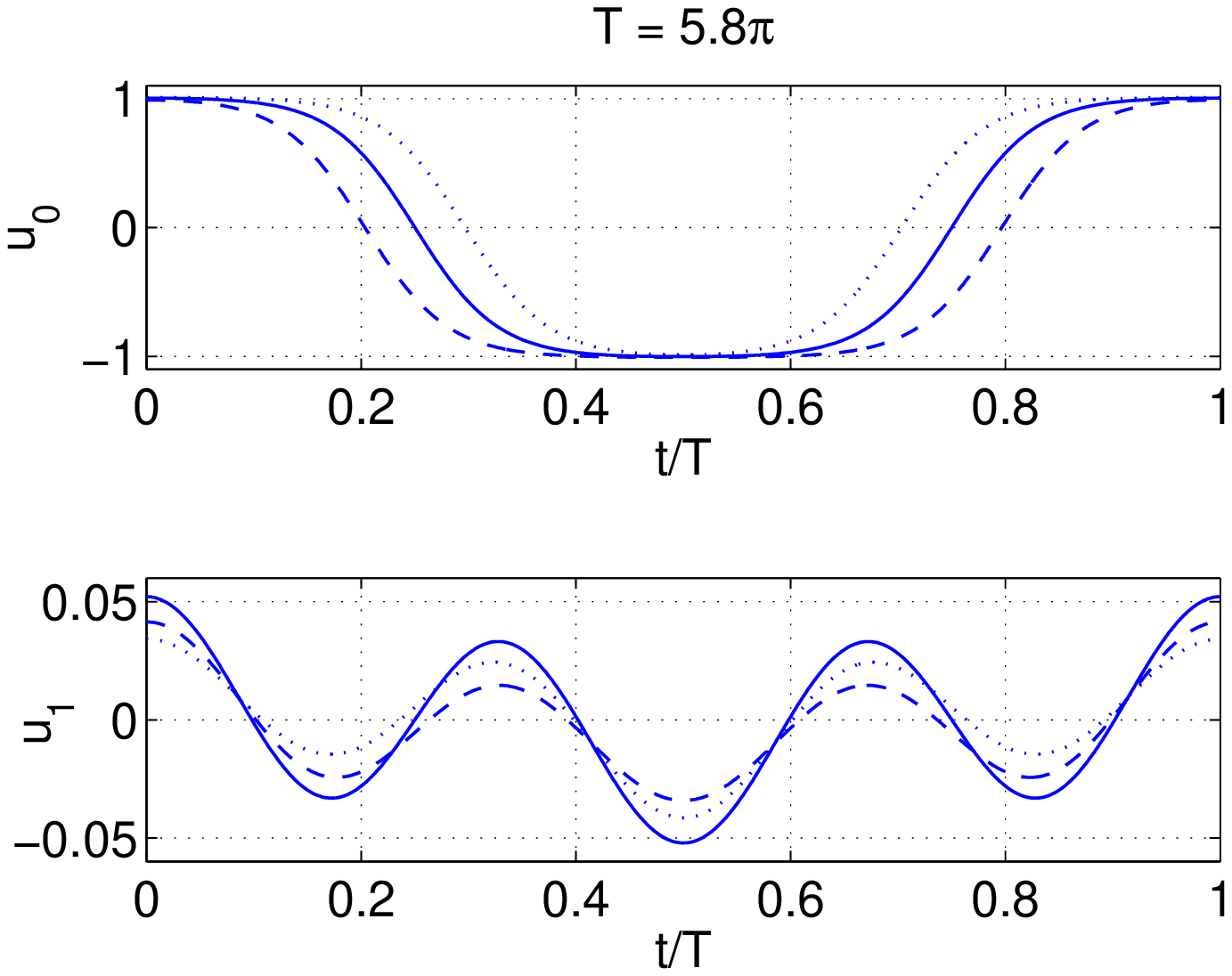}\tabularnewline
\end{tabular}
\par\end{centering}
\caption{Fundamental breathers for system (\ref{eq:dKG_trunc})
before (left) and after (right) the symmetry-breaking bifurcation
at $\epsilon=0.01$. \label{fig:solutions}}
\end{figure}

\begin{figure}
\begin{centering}
\begin{tabular}{cc}
\includegraphics[width=0.47\textwidth]{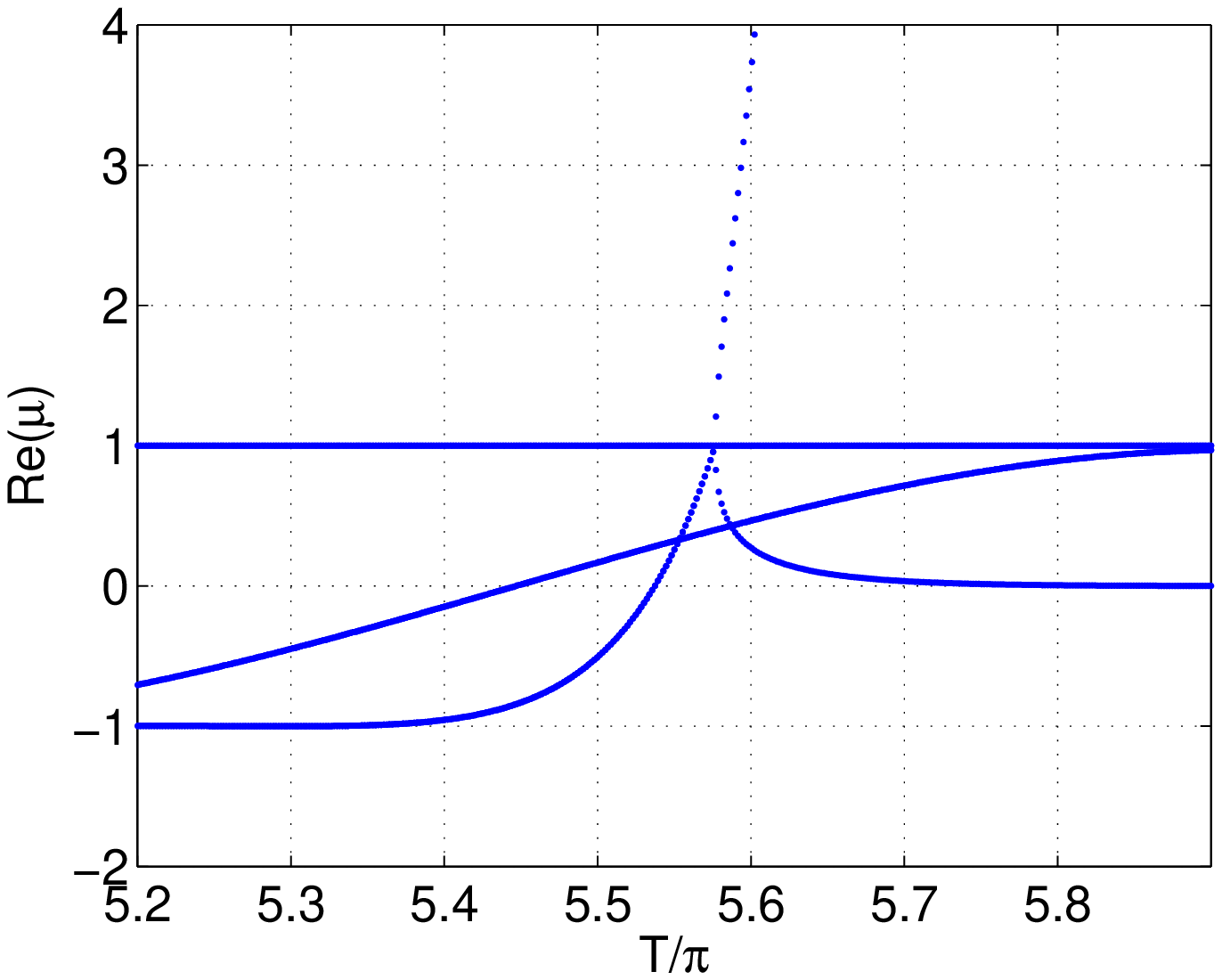} & \includegraphics[width=0.47\textwidth]{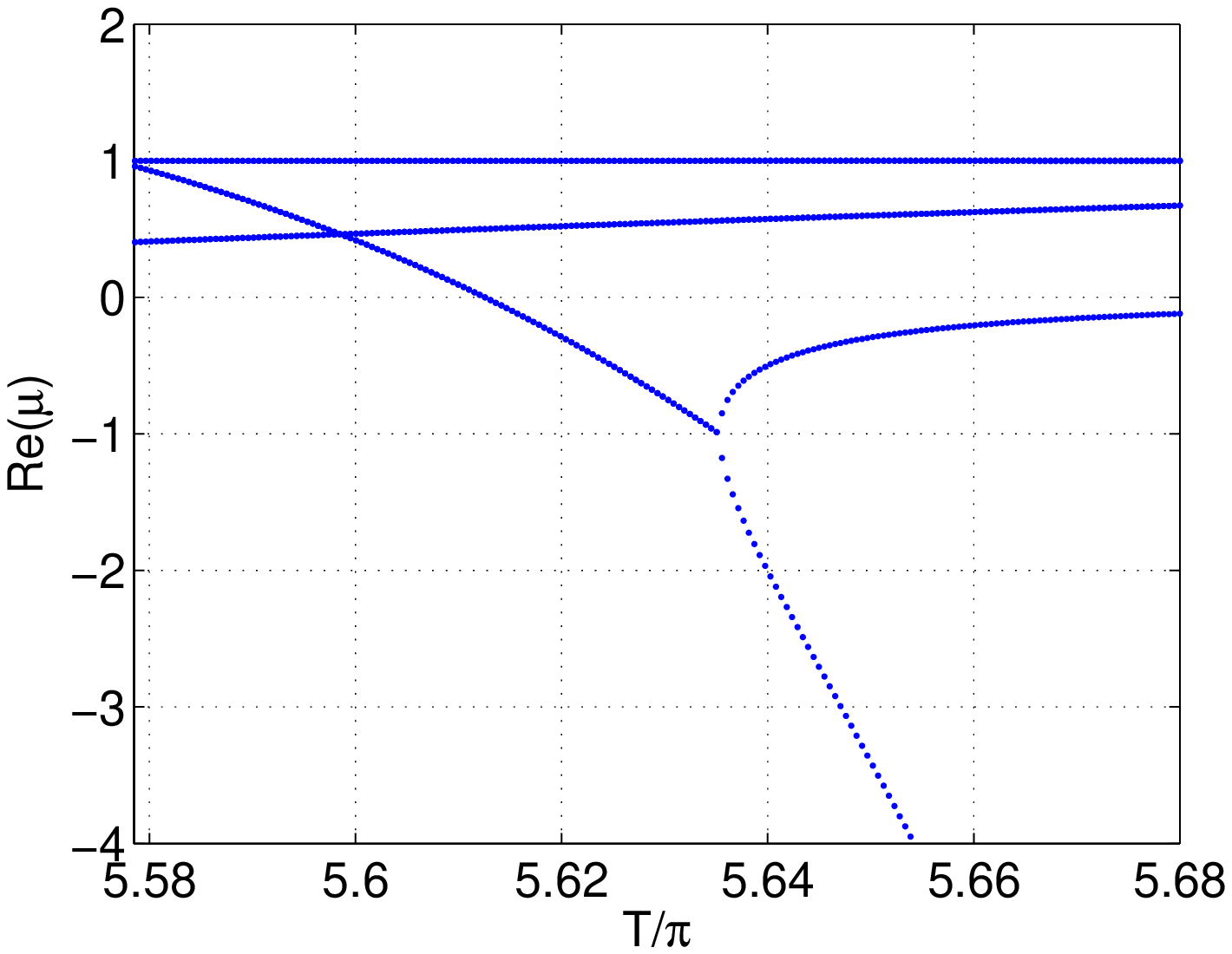}\tabularnewline
\end{tabular}
\par\end{centering}
\caption{Real parts of Floquet multipliers $\mu$ for the fundamental breather at
$\epsilon=0.01$ near the bifurcation for the main branch (left) and
side branches (right).
\label{fig:FM_010}}
\end{figure}
\begin{figure}
\begin{centering}
\begin{tabular}{cc}
\includegraphics[width=0.47\textwidth]{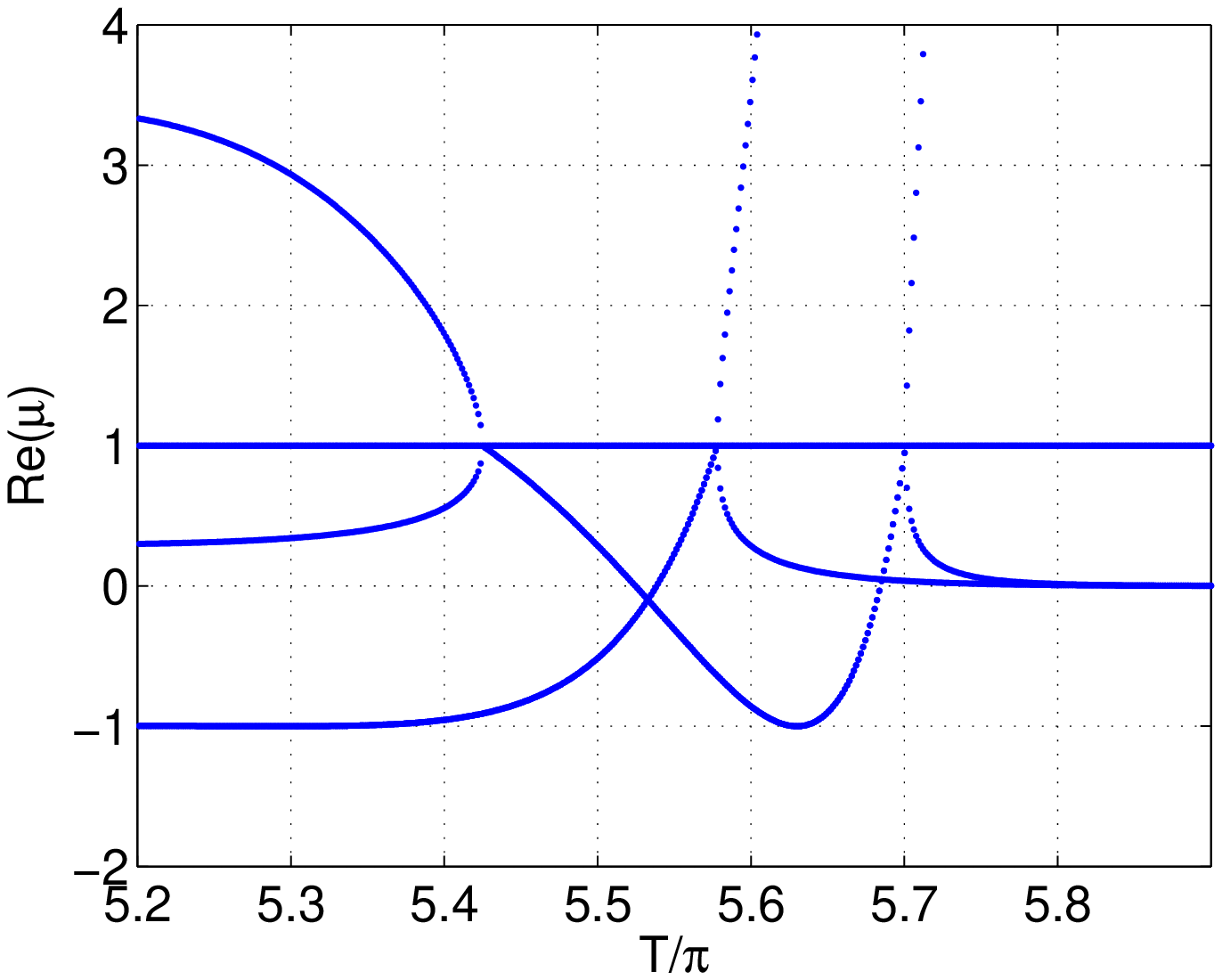} & \includegraphics[width=0.47\textwidth]{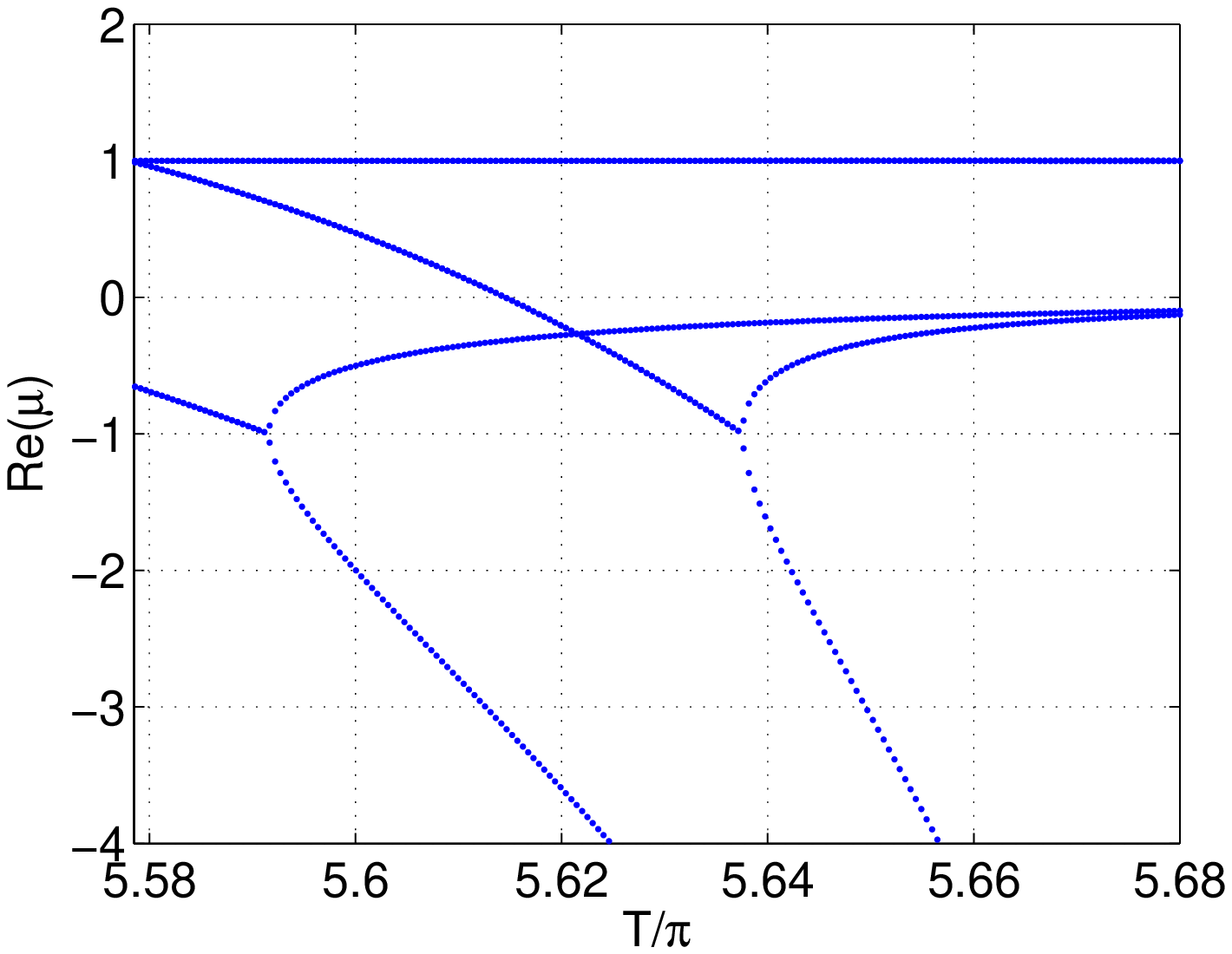}\tabularnewline
\end{tabular}
\par\end{centering}
\caption{Real parts of Floquet multipliers $\mu$ for the two-site breather with a
hole at $\epsilon=0.01$ near the bifurcation for the main branch
(left) and side branches (right). \label{fig:FM_101}}
\end{figure}

We perform similar computations for the two-site breather with the
central hole (\ref{hole-breat}). Figure \ref{fig:FM_101} (left) shows that
at the coupling $\epsilon=0.01$ the breather is unstable for periods
$2\pi < T < T_{*}^{(\epsilon)}$ and stable for periods $T \gtrapprox T_{*}^{(\epsilon)}$
with $T_{*}^{(\epsilon)}\approx5.425\pi$ for $\epsilon = 0.01$. This can be compared using
the change of stability predicted by Theorem \ref{theorem-eigenvalue}.
According to equation (\ref{formula-K-N}), $K_{2}$ changes sign
from positive to negative at $T_{N = 2} \approx5.476\pi$. Since $T'(E)$
is positive for the soft potential, Theorem \ref{theorem-eigenvalue}
predicts that in the anti-continuum limit the two-site breather is
unstable for $2 \pi < T < T_{N = 2}$ and stable for $T_{N = 2} < T < 6\pi$.
This change of stability agrees with Figure \ref{fig:FM_101} where we note that
$|T_{*}^{(\epsilon)}-T_{N = 2}|\approx0.05\pi$ at $\epsilon=0.01$.

At $T\approx5.6\pi$ and $T\approx5.7\pi$, two bifurcations occur
for the two-site breather with the central hole
and unstable multipliers bifurcate from the unit multiplier for larger
values of $T$. The behavior of Floquet multipliers is similar to
the one on Figure \ref{fig:FM_010} (left) and it marks two
consequent pitchfork bifurcations for the two-site breather with the hole.
The first bifurcation
is visible on Figure \ref{fig:breather_branches} in the space of
symmetric two-site breathers with $u_{-1}(t)=u_{1}(t)$.
The Floquet multipliers for the side branches of these symmetric two-site breathers is
shown on Figure \ref{fig:FM_101} (right), where we can see two consequent period-doubling
bifurcations in comparison with one such bifurcation on Figure \ref{fig:FM_010} (right). The second
bifurcation is observed in the space of asymmetric two-site breathers
with $u_{-1}(t)\neq u_{1}(t)$.

We display the two pitchfork bifurcations
on the top panel of Figure \ref{fig:asym_breathers}. One can see for the second
bifurcation that the value of $a_{0}$ is the same for both breathers splitting
of the main solution branch. Although the values of $a_{-1}$ and $a_1$
look same for the second bifurcation, dashed and dotted lines indicate that $a_1$
is greater than $a_{-1}$ at one asymmetric branch and vice versa for the other one.
The bottom panels of Figure \ref{fig:asym_breathers} show
the asymmetric breathers with period $T=5.75\pi$ that appear as a result
of the second pitchfork bifurcation.

\subsection{Five-site model}

We can now truncate the discrete Klein--Gordon equation (\ref{KGlattice})
at five lattice sites, e.g. at $n\in\{-2,-1,0,1,2\}$. The fundamental
breather (\ref{fund-breat}) and the breather with a central hole
(\ref{hole-breat}) are continued in the five-site lattice subject to
the symmetry conditions $u_{n}(t)=u_{-n}(t)$ for $n=1,2$. We would
like to illustrate that increasing the size of the lattice does not
qualitatively change the previous existence and stability results,
in particular, the properties of the pitchfork bifurcations.

Figure \ref{fig:gap_five-node} gives
analogues of Figure \ref{fig:breather_branches} for the fundamental
breather and the breather with a hole. The associated
Floquet multipliers are shown on Figure \ref{fig:FMs_gap},
in full analogy with  Figures \ref{fig:FM_010} and \ref{fig:FM_101}.
We can see that both existence and stability results are analogous
between the three-site and five-site lattices.

\begin{figure}
\begin{centering}
\begin{tabular}{ccc}
\includegraphics[width=0.32\textwidth]{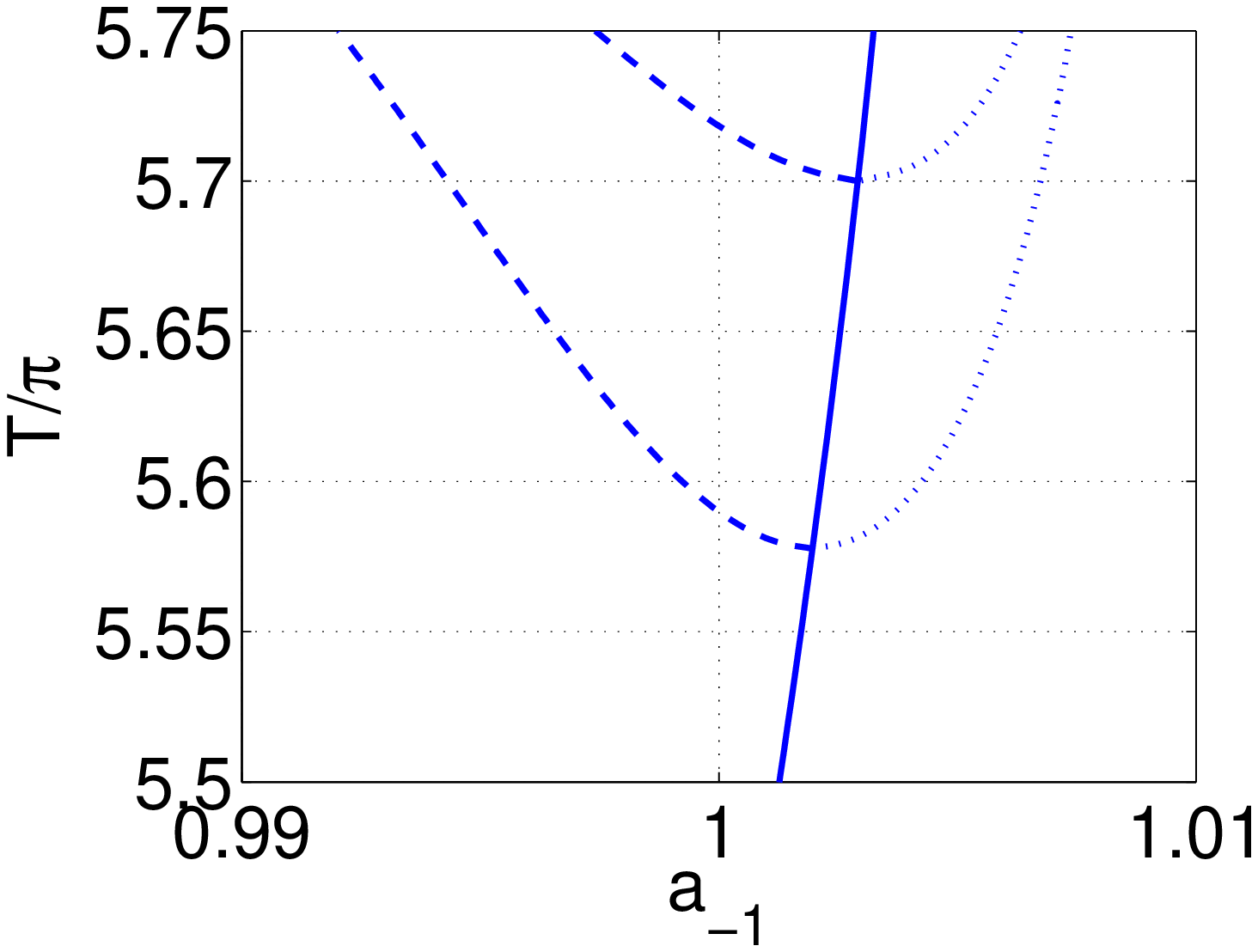}  &
\includegraphics[width=0.32\textwidth]{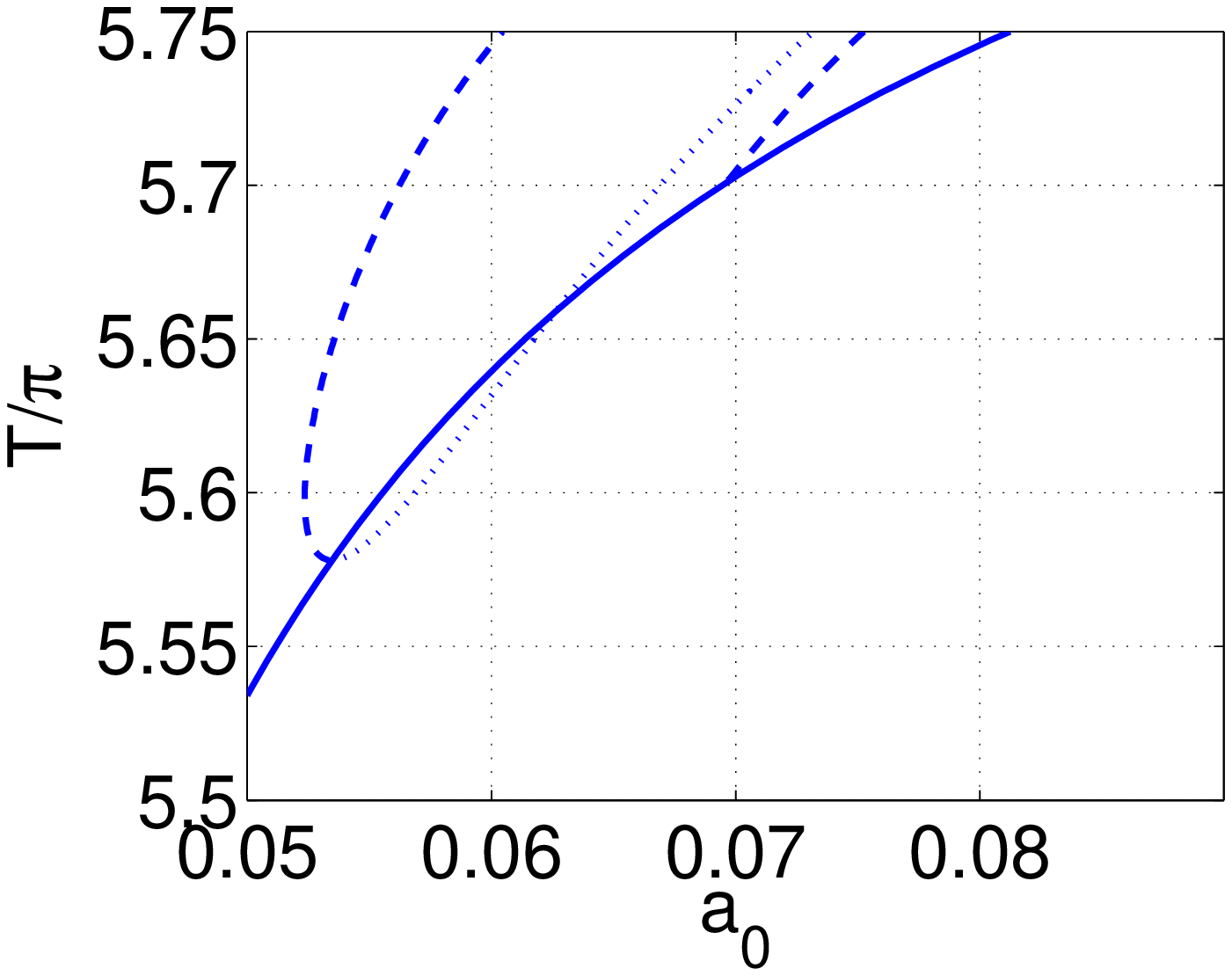}  & \includegraphics[width=0.32\textwidth]{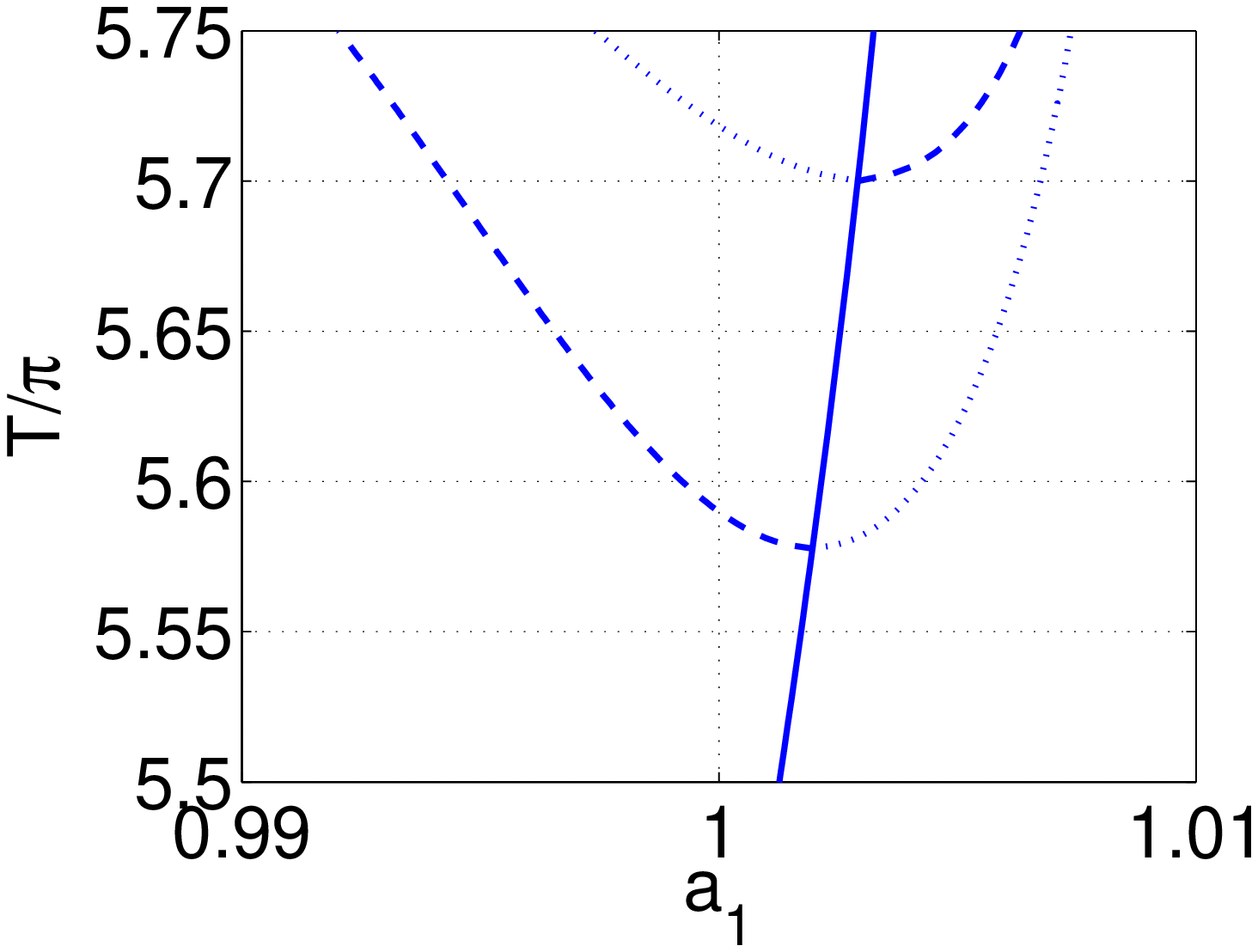}\tabularnewline
\includegraphics[width=0.29\textwidth]{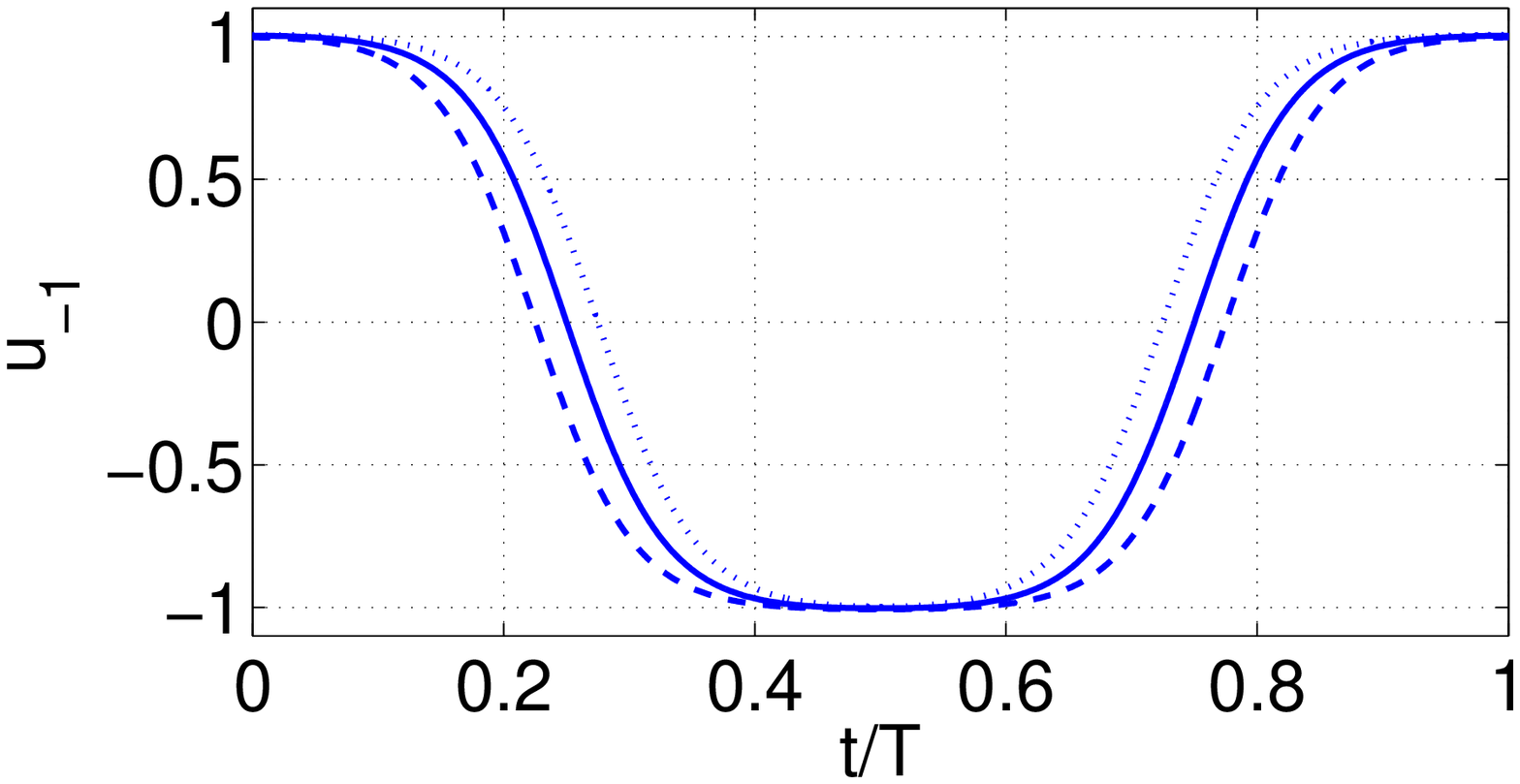} &
\includegraphics[width=0.29\textwidth]{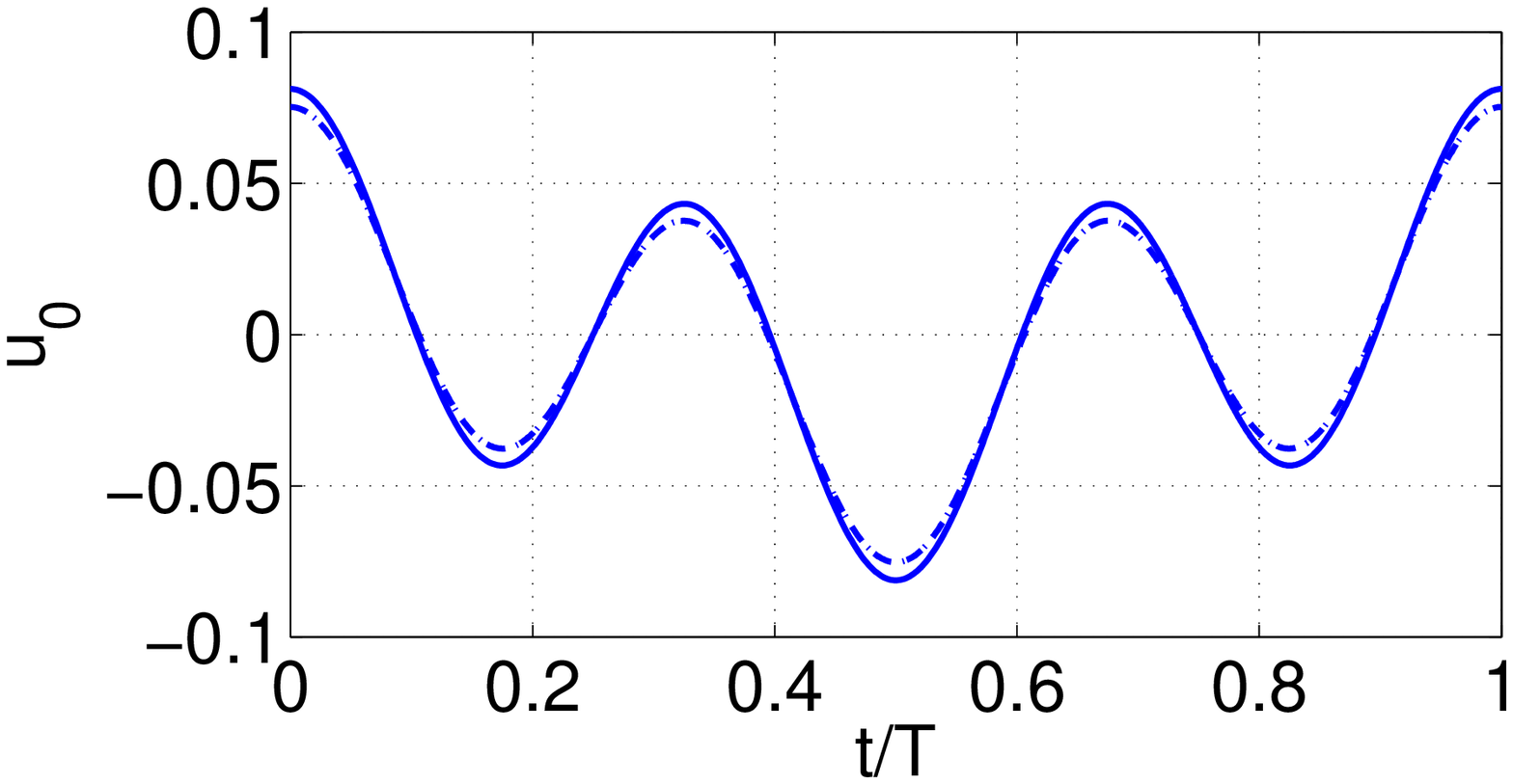} & \includegraphics[width=0.29\textwidth]{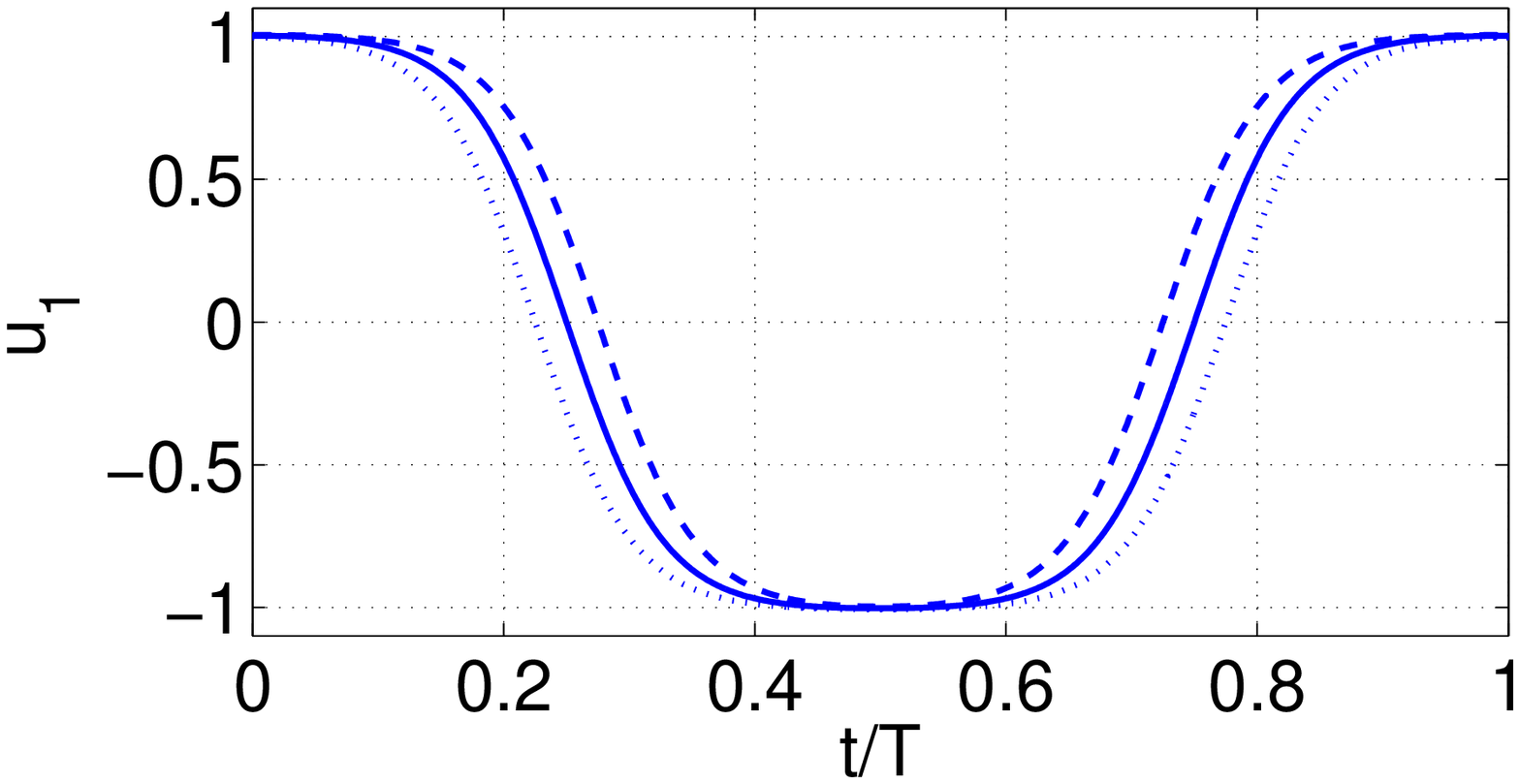}\tabularnewline
\end{tabular}
\par\end{centering}
\caption{Top: The initial displacements $a_{-1}$, $a_0$, and $a_1$
for the $T$-periodic breather with a hole
on the three-site lattice with $\epsilon=0.01$.
Bottom: Asymmetric breathers with period $T=5.75\pi$ on the three-site
lattice with $\epsilon=0.01$.
\label{fig:asym_breathers}}
\end{figure}
\begin{figure}
\begin{centering}
\begin{tabular}{ccc}
\includegraphics[width=0.32\textwidth]{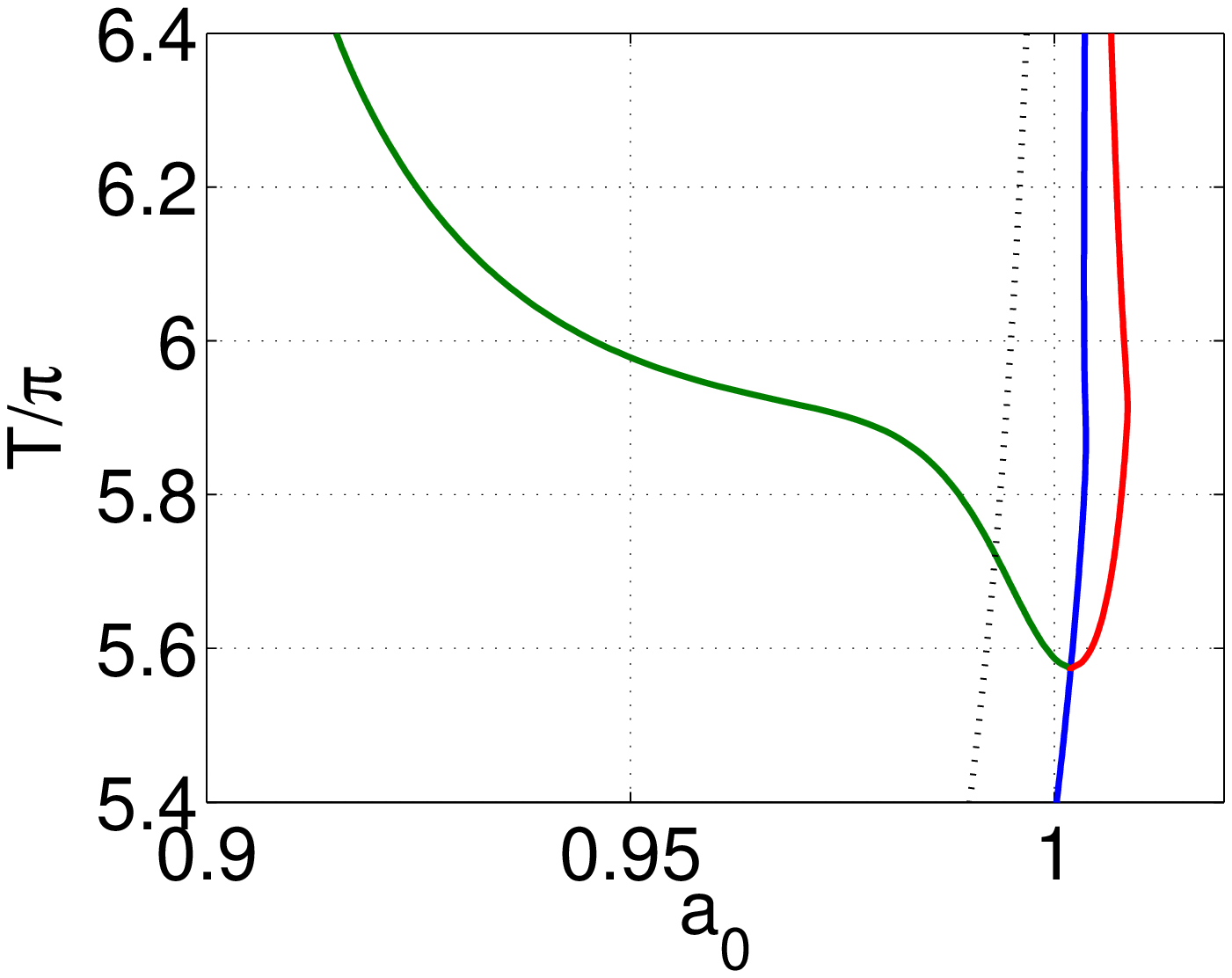}  &
\includegraphics[width=0.32\textwidth]{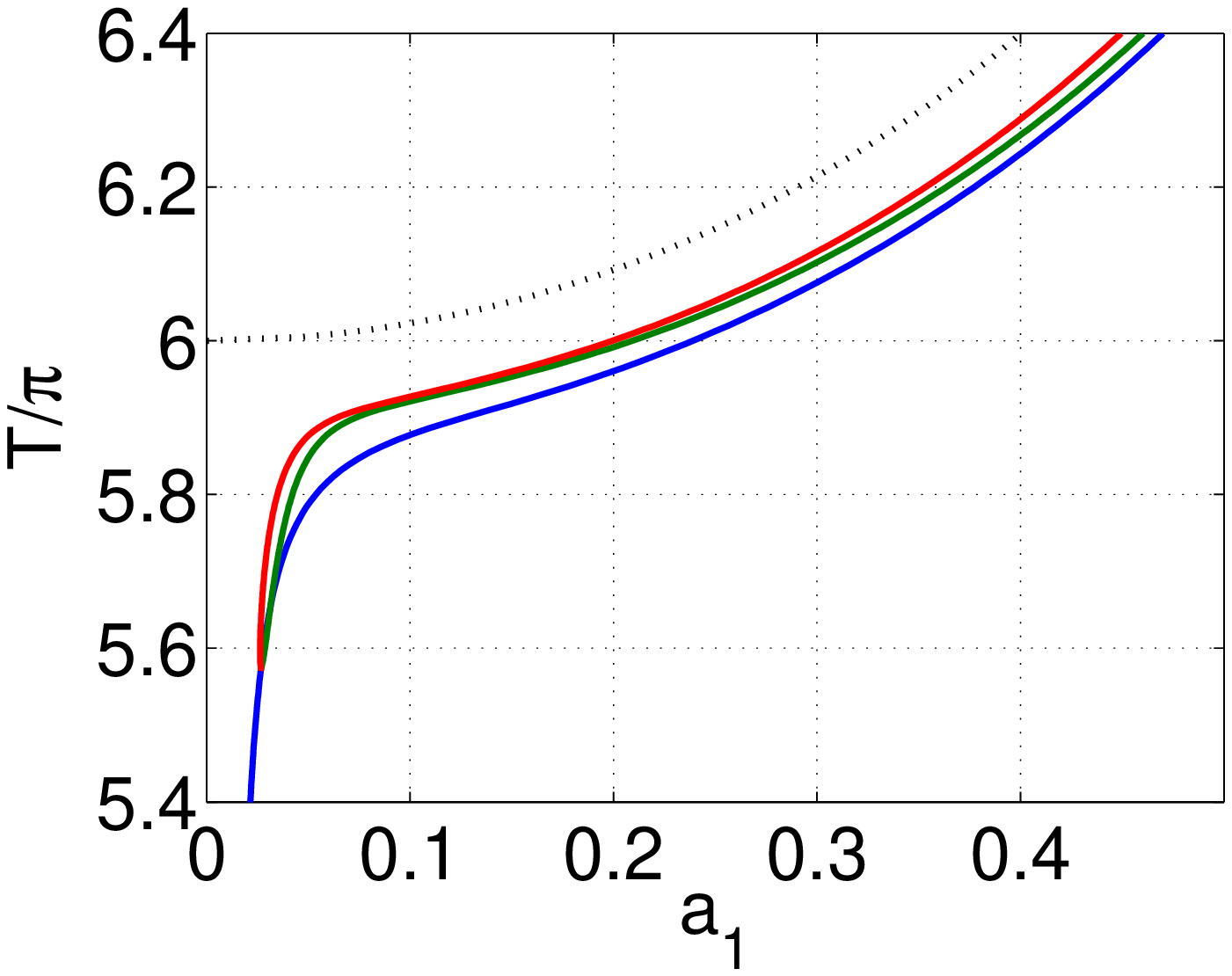}  & \includegraphics[width=0.32\textwidth]{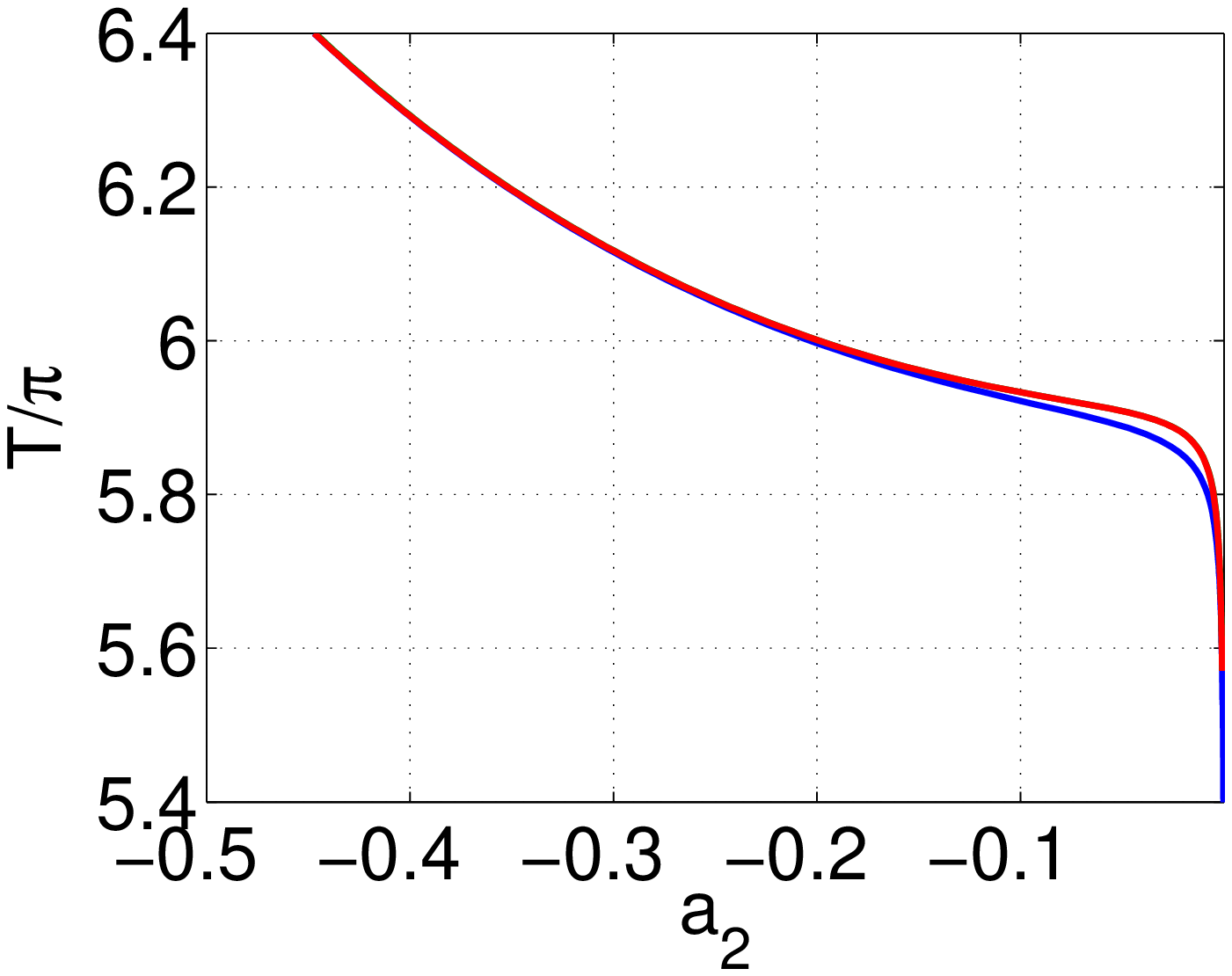}\tabularnewline
\includegraphics[width=0.32\textwidth]{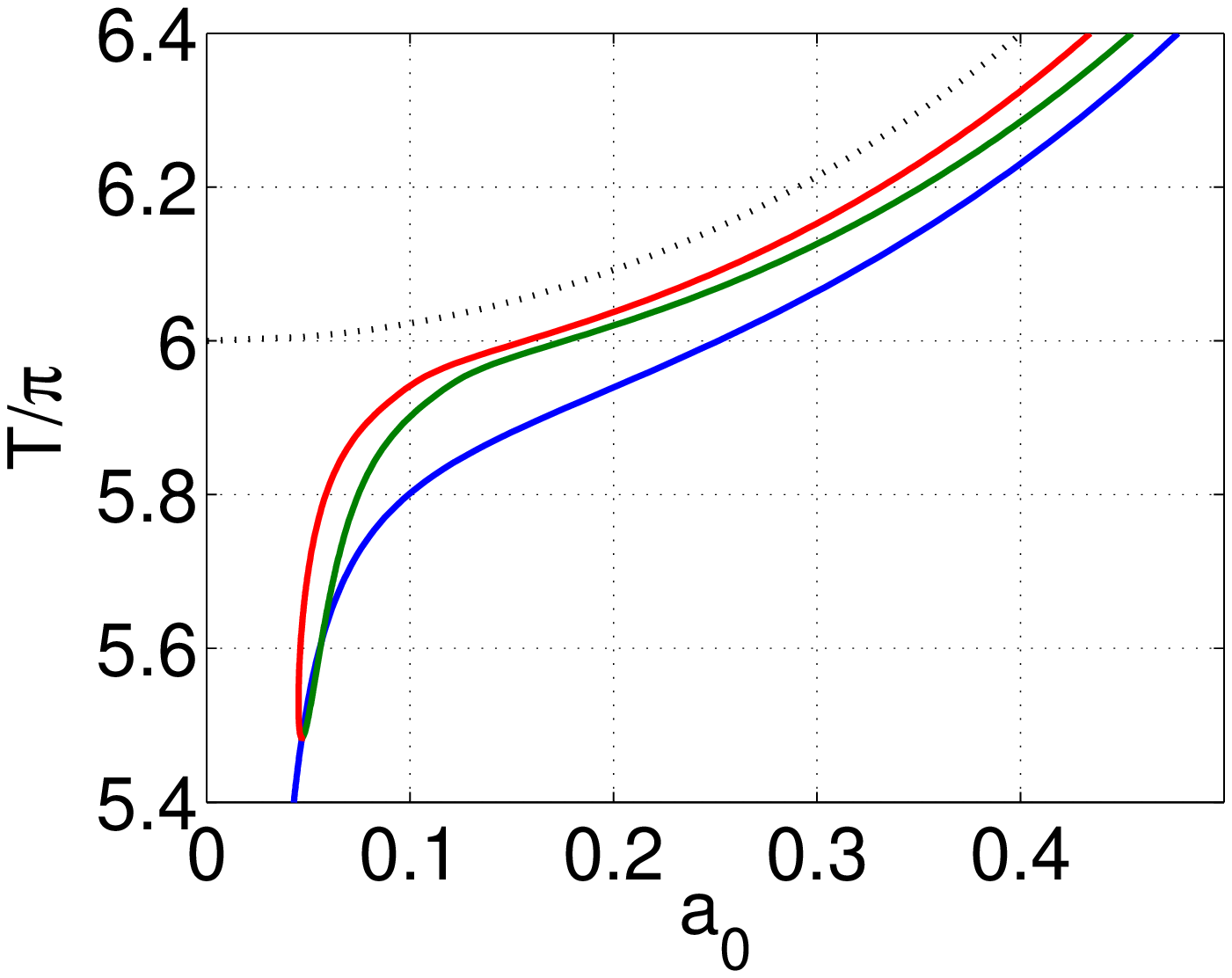}  &
\includegraphics[width=0.32\textwidth]{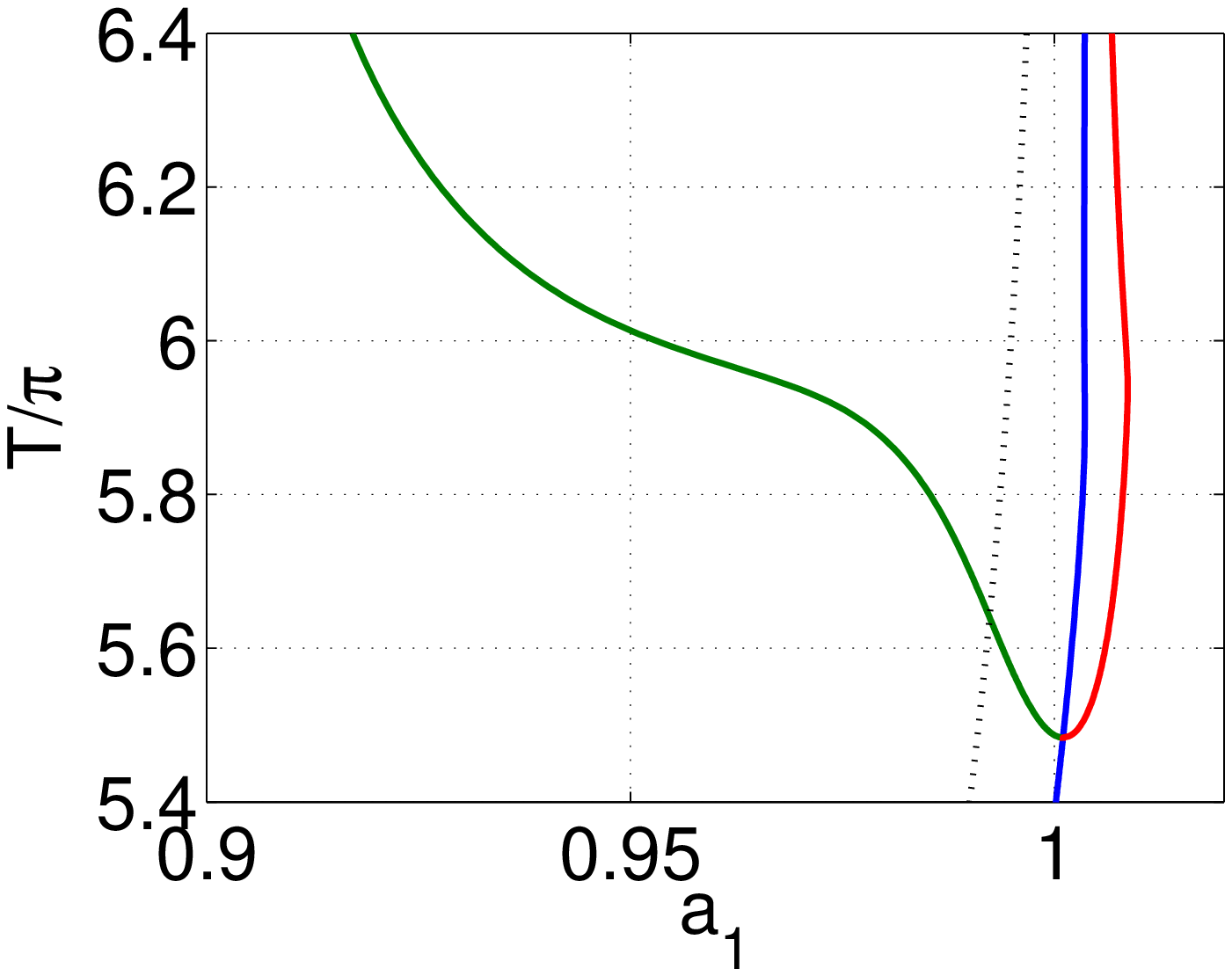}  & \includegraphics[width=0.32\textwidth]{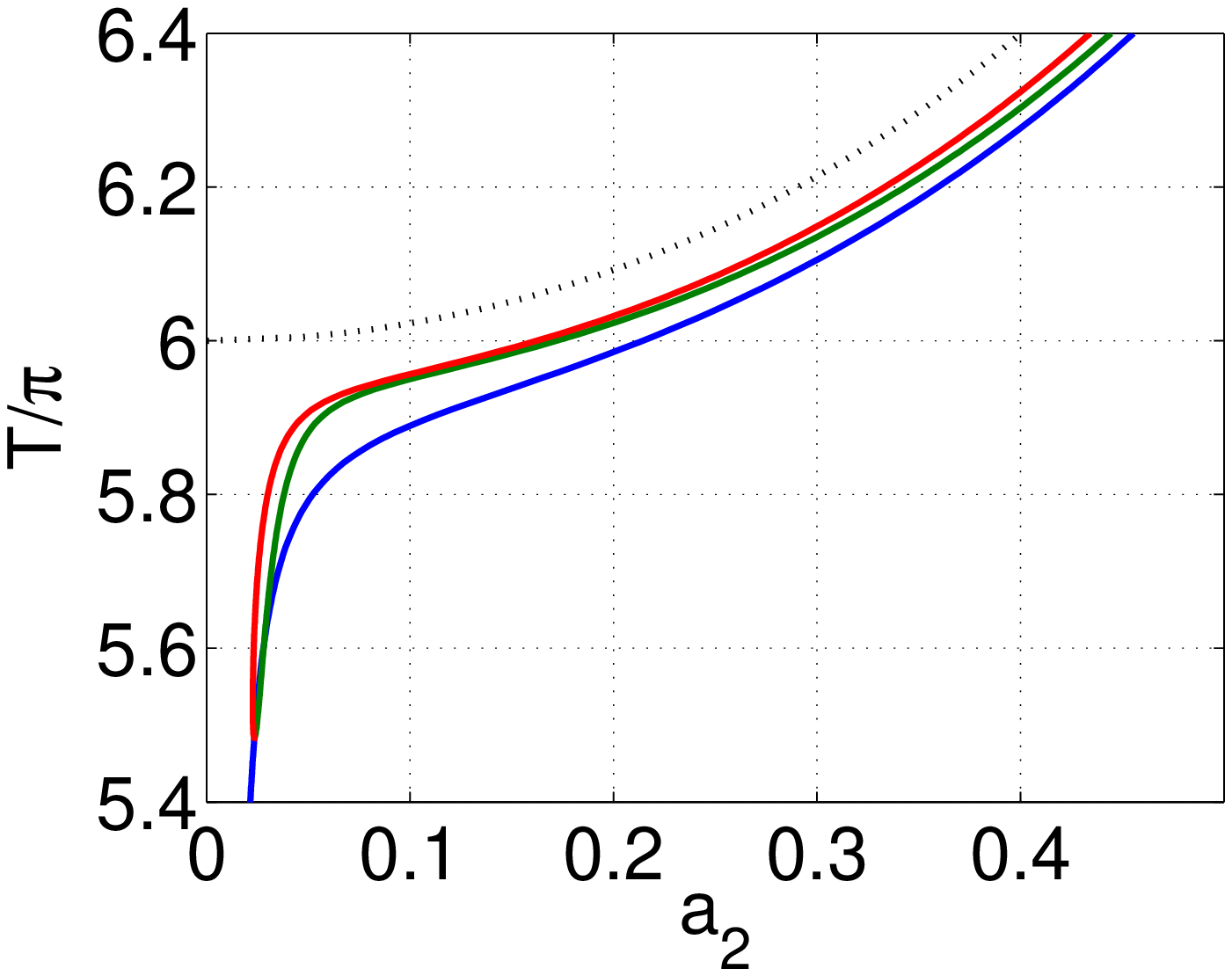}\tabularnewline
\end{tabular}
\par\end{centering}
\caption{Top: The initial displacements $a_0$, $a_1$, and $a_2$ for the $T$-periodic
fundamental breather of the five-site
lattice with $\epsilon=0.01$.
Bottom: The same for the two-site breather with a hole.
The dotted lines correspond to the $T$-periodic
solutions to equation (\ref{nonlinear-oscillator}).
\label{fig:gap_five-node}}
\end{figure}

\begin{figure}
\begin{centering}
\begin{tabular}{cc}
\includegraphics[width=0.47\textwidth]{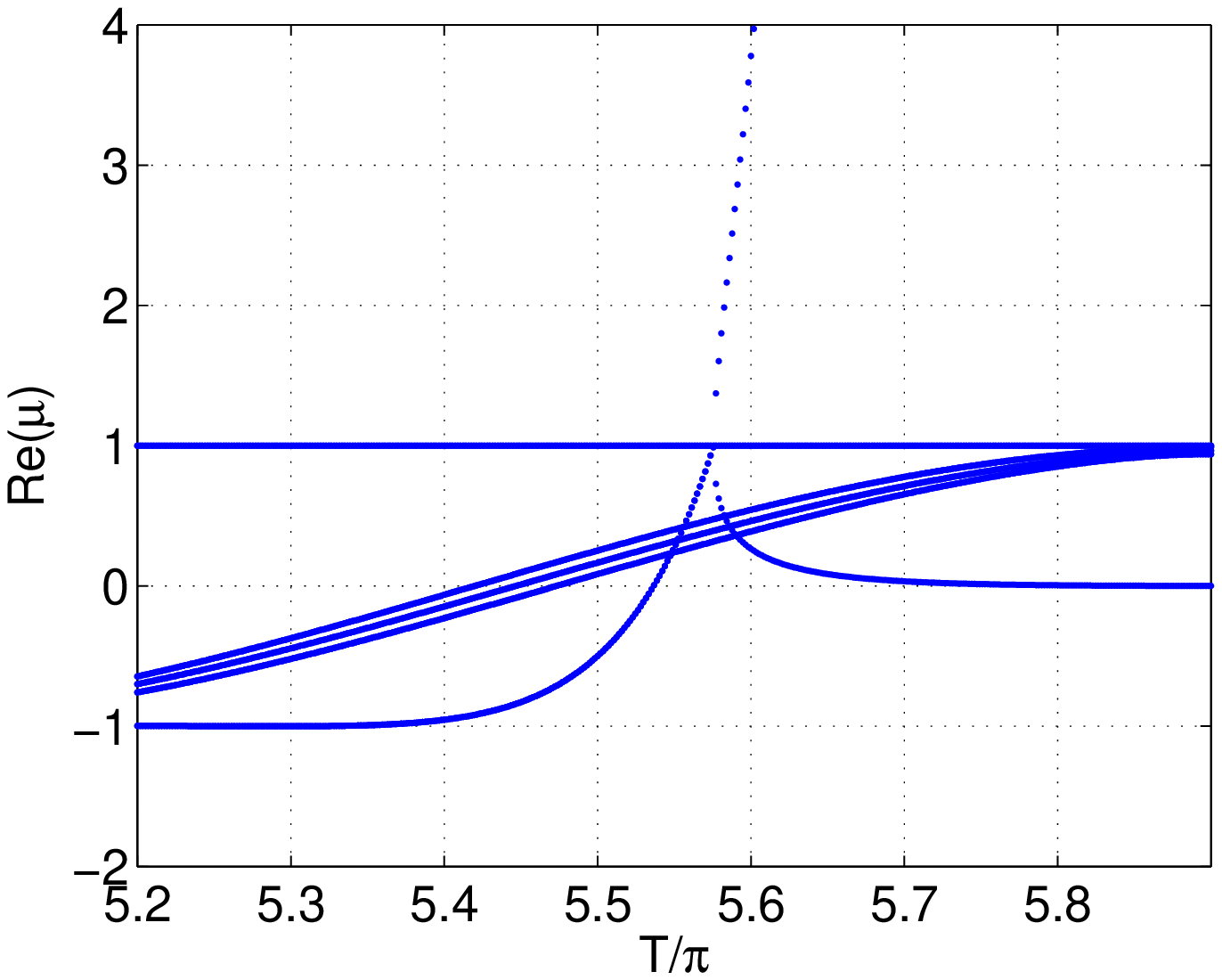} &
\includegraphics[width=0.47\textwidth]{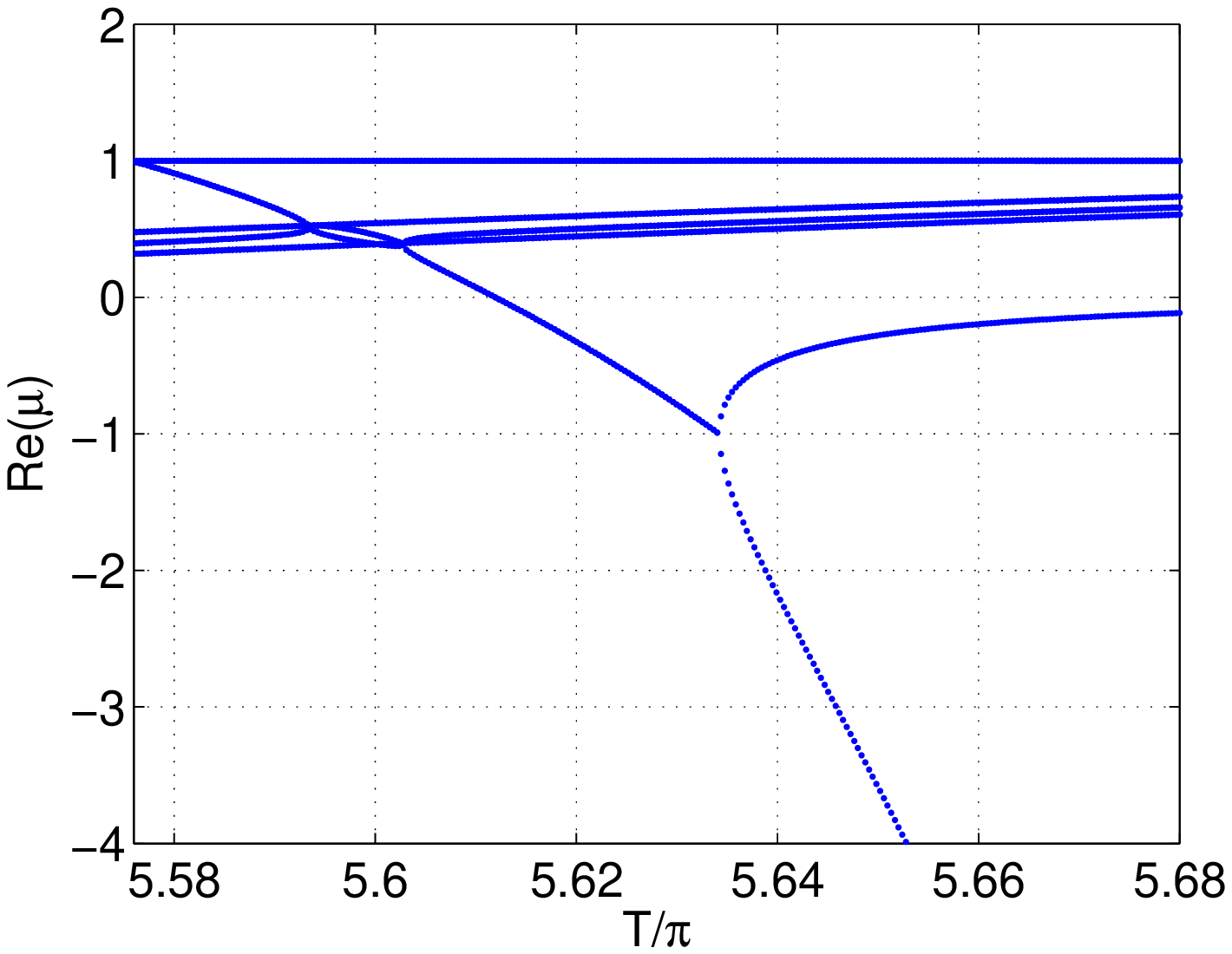}\tabularnewline
\includegraphics[width=0.47\textwidth]{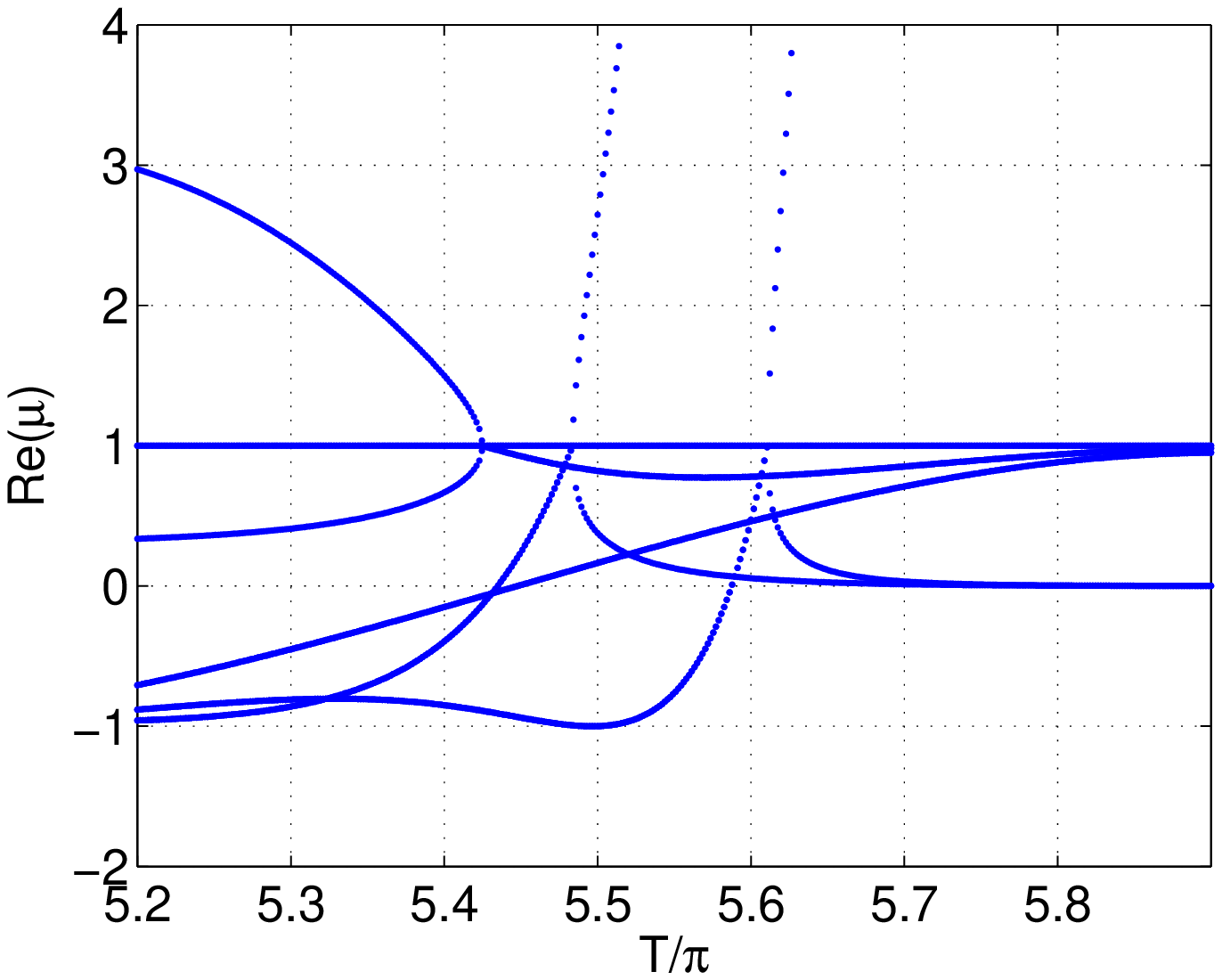} &
\includegraphics[width=0.47\textwidth]{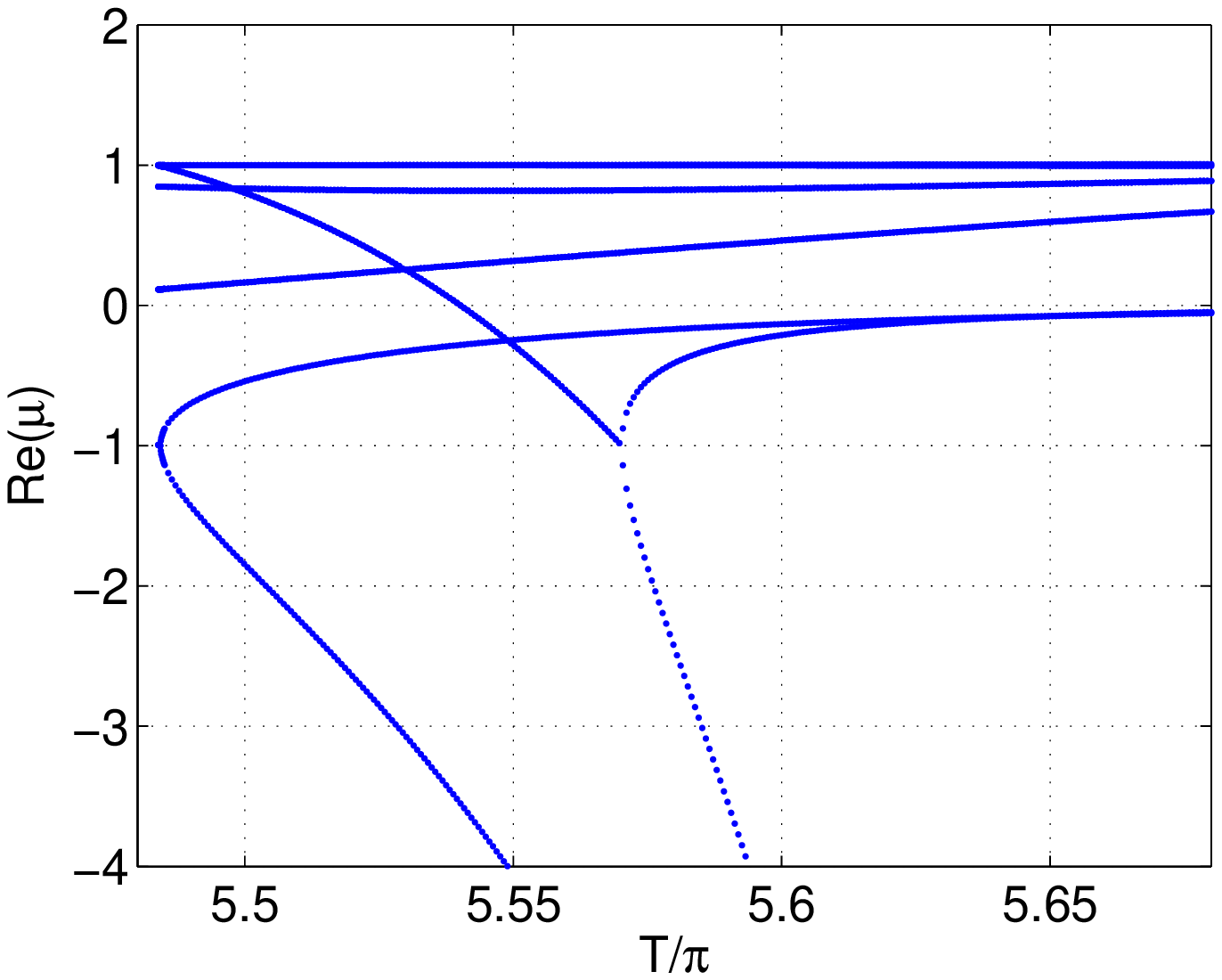}\tabularnewline
\end{tabular}
\par\end{centering}
\caption{Top: Real parts of Floquet multipliers $\mu$ for the fundamental breather
near the bifurcation for the main branch (left) and side branches (right).
Bottom: The same for the two-site breather with a hole.
\label{fig:FMs_gap}}
\end{figure}

\section{Pitchfork bifurcation near 1:3 resonance \label{sec:sym-breaking}}
\label{sec-resonance}

We study here the symmetry-breaking (pitchfork) bifurcation of the fundamental
breather. This bifurcation illustrated on Figure \ref{fig:solutions}
occurs for soft potentials near the point of 1:3 resonance, when the
period $T$ is close to $6\pi$. We point out that the period $T_S$
of the pitchfork bifurcation is close to $6\pi$ for small but
finite values of $\epsilon$. As we have discovered numerically, $T_S$ gets larger
as $\epsilon$ gets smaller. This property indicates that the asymptotic analysis of this bifurcation
is not uniform with respect to two small parameters $\epsilon$ and $T - 6\pi$,
which we explain below in more details.

When ${\bf u}=\mbox{\boldmath\ensuremath{\phi}}^{(\epsilon)}$ is
the fundamental breather and $T\neq2\pi n$ is fixed, Theorem \ref{theorem-continuum}
and Lemma \ref{lemma-tail-interaction} imply that
\begin{equation}
\left\{ \begin{array}{lclcl}
u_{0}(t) & = & \varphi(t)-2\epsilon\psi_{1}(t) & + & \mathcal{O}_{H_{{\rm per}}^{2}(0,T)}(\epsilon^{2}),\\
u_{\pm1}(t) & = & \phantom{texttext} \epsilon\varphi_{1}(t) & + & \mathcal{O}_{H_{{\rm per}}^{2}(0,T)}(\epsilon^{2}),\\
u_{\pm n}(t) & = & \phantom{t} & + & \mathcal{O}_{H_{{\rm per}}^{2}(0,T)}(\epsilon^{2}),\quad n\geq2,
\end{array}\right.
\label{fourier-series-0}
\end{equation}
where $\varphi$ can be expanded in the Fourier series,
\begin{equation}
\varphi(t)=\sum_{n\in\N_{{\rm odd}}}c_{n}(T)\cos\left(\frac{2\pi nt}{T}\right),
\label{fourier-series-1}
\end{equation}
and the Fourier coefficients $\{c_{n}(T)\}_{n\in\N_{{\rm odd}}}$
are uniquely determined by the period $T$. The correction terms $\varphi_{1}$
and $\psi_{1}$ are determined by the solution of the linear inhomogeneous
equations (\ref{inhomogen-1}) and (\ref{inhomogen-2}), in particular,
we have
\begin{equation}
\label{fourier-series-0a}
\varphi_{1}(t)=\sum_{n\in\N_{{\rm odd}}}\frac{T^{2}c_{n}(T)}{T^{2}-4\pi^{2}n^{2}}\cos\left(\frac{2\pi nt}{T}\right).
\end{equation}

In what follows, we restrict our consideration of soft potentials
to the case of the quartic potential $V'(u)=u-u^{3}$. We shall assume
that $c_3(6\pi) \neq 0$ and the numerical approximations suggest that
$c_3(6\pi) < 0$ for the quartic potential.

Expansion (\ref{fourier-series-0}) and solution (\ref{fourier-series-0a})
imply that if $T$ is fixed in $(2 \pi, 6\pi)$,
then $\|u_{\pm1}\|_{H_{{\rm per}}^{2}(0,T)}={\cal O}(\epsilon)$
and the cubic term $u_{\pm1}^{3}$ is neglected at the order $\mathcal{O}(\epsilon)$,
where the linear inhomogeneous equation (\ref{inhomogen-1}) is valid.
Near the resonant period $T=6\pi$, the norm $\|u_{\pm1}\|_{H_{{\rm per}}^{2}(0,T)}$
is much larger than ${\cal O}(\epsilon)$ if $c_3(6\pi) \neq 0$. As a result,
the cubic term $u_{\pm1}^{3}$ must be incorporated at the leading order
of the asymptotic approximation.

We shall reduce the discrete Klein--Gordon equation (\ref{KGlattice})
for the fundamental breather near $1 : 3$ resonance to a normal form
equation, which coincides with the nonlinear Duffing oscillator perturbed by a
small harmonic forcing (equation (\ref{normal-form-0}) below). The normal
form equation features the same properties of the pitchfork bifurcation
of $T$-periodic solutions as the discrete Klein--Gordon equation (\ref{KGlattice}).
To prepare for the reduction to the normal form equation, we
introduce the scaling transformation,
\begin{equation}
T=\frac{6\pi}{1+\delta\epsilon^{2/3}},\quad\tau=(1+\delta\epsilon^{2/3})t,\quad
u_{n}(t)=(1+\delta\epsilon^{2/3})U_{n}(\tau),
\label{scal-T}
\end{equation}
where $\delta$ is a new parameter, which is assumed to be $\epsilon$-independent.
The discrete Klein--Gordon equation (\ref{KGlattice}) with $V'(u) = u - u^3$
can be rewritten in new variables (\ref{scal-T}) as follows,
\begin{equation}
\ddot{U}_{n}+U_{n}-U_{n}^{3}=\beta U_{n}+\gamma(U_{n+1}+U_{n-1}),\quad n\in\mathbb{Z},
\label{KGlattice-tilded}
\end{equation}
where
\begin{equation}
\beta=1-\frac{1+2\epsilon}{(1+\delta\epsilon^{2/3})^{2}},\quad
\gamma=\frac{\epsilon}{(1+\delta\epsilon^{2/3})^{2}}.
\label{scaling-normal-form}
\end{equation}
$T$-periodic solutions of the discrete Klein--Gordon equation (\ref{KGlattice})
in variables $\{u_{n}(t)\}_{n\in\Z}$ become
now $6\pi$-periodic solutions of the rescaled Klein--Gordon equation
(\ref{KGlattice-tilded}) in variables $\{U_{n}(\tau)\}_{n\in\Z}$.
To reduce the system of Klein--Gordon equation (\ref{KGlattice-tilded})
to the Duffing oscillator perturbed by a small harmonic forcing near 1:3 resonance,
we consider the fundamental breather, for which $U_n = U_{-n}$ for all $n \in \mathbb{N}$.
Using this reduction, we write equations (\ref{KGlattice-tilded})
separately at $n = 0$, $n = 1$, and $n \geq 2$:
\begin{eqnarray}
\label{system-0}
\ddot{U}_{0}+U_{0}-U_{0}^{3} & = & \beta U_{0}+ 2 \gamma U_{1}, \\
\label{system-1}
\ddot{U}_{1}+U_{1}-U_{1}^{3} & = & \beta U_{1} + \gamma U_{2} + \gamma U_{0}, \\
\label{system-2}
\ddot{U}_{n}+U_{n}-U_{n}^{3} & = & \beta U_{n}+\gamma(U_{n+1}+U_{n-1}), \quad n \geq 2.
\end{eqnarray}

Let us represent a $6\pi$-periodic function $U_{0}$ with the symmetries
\begin{equation}
\label{sym-U-0}
U_0(-\tau) = U_0(\tau) = -U_0(3 \pi -\tau), \quad \tau \in \R,
\end{equation}
by the Fourier series,
\begin{equation}
U_{0}(\tau) = \sum_{n\in\N_{{\rm odd}}} b_{n} \cos\left(\frac{n\tau}{3}\right),
\label{fourier-series-2}
\end{equation}
where $\{b_{n}\}_{n\in\N_{{\rm odd}}}$ are some Fourier coefficients. If
$U_0$ converges to $\varphi$ in $H^2$-norm as $\epsilon \to 0$ (when
$\beta, \gamma \to 0$), then $b_n \to c_n(6\pi)$ as $\epsilon \to 0$ for all $n \in \N_{\rm odd}$, where
the Fourier coefficients $\{c_{n}(6\pi)\}_{n\in\N_{{\rm odd}}}$ are uniquely defined by
the Fourier series (\ref{fourier-series-1}) for $T = 6\pi$. We assume
again that $c_3(6\pi) \neq 0$ and $\delta$ is fixed independently of small $\epsilon > 0$.

We shall now use a Lyapunov--Schmidt reduction method to show that the components
$\{ U_n \}_{n \in \N}$ are uniquely determined from the system (\ref{system-1})--(\ref{system-2})
for small $\epsilon > 0$ if $U_0$ is represented by the Fourier series (\ref{fourier-series-2}).
To do so, we decompose the solution into two parts:
\begin{equation}
U_{n}(\tau)= A_{n} \cos(\tau) + V_{n}(\tau),\quad n\in\N,
\label{substitution-U-V}
\end{equation}
where $V_n(\tau)$ is orthogonal to $\cos(\tau)$ in the sense
$\langle V_n, \cos(\cdot) \rangle_{L^2_{\rm per}(0,6\pi)} = 0$.
Projecting the system (\ref{system-1})--(\ref{system-2}) to
$\cos(\tau)$, we obtain a difference equation for $\{ A_n \}_{n \in \N}$:
\begin{eqnarray}
\label{system-1-A}
\beta A_{1} + \gamma A_{2} + \gamma b_{3} & = & -\frac{1}{3\pi} \int_0^{6\pi} \cos(\tau) (A_1 \cos(\tau) + V_1(\tau) )^3 d \tau, \\
\label{system-2-A}
\beta A_{n} + \gamma (A_{n+1} + A_{n-1}) & = & -\frac{1}{3\pi} \int_0^{6\pi} \cos(\tau) (A_n \cos(\tau) + V_n(\tau) )^3 d \tau, \quad n \geq 2.
\end{eqnarray}
Projecting the system (\ref{system-1})--(\ref{system-2}) to the orthogonal complement of
$\cos(\tau)$, we obtain a lattice differential equation for $\{ V_n(\tau) \}_{n \in \N}$:
\begin{eqnarray}
\nonumber
\ddot{V}_{1}+V_{1} & = & \beta V_{1} + \gamma V_{2} + \gamma \sum_{k\in\N_{{\rm odd}}\backslash\{3\}}b_{k}\cos\left(\frac{k\tau}{3}\right) \\
\label{system-1-V}
& \phantom{t} &  + (A_{1} \cos(\tau) + V_1)^{3} - \cos(\tau)
\frac{\langle \cos(\cdot),(A_1 \cos(\cdot) + V_1)^3 \rangle_{L^2_{\rm per}(0,6\pi)}}{
\langle \cos(\cdot),\cos(\cdot)\rangle_{L^2_{\rm per}(0,6\pi)}}, \\
\nonumber
\ddot{V}_{n}+V_{n} & = & \beta V_{n} +\gamma (V_{n+1}+V_{n-1}) \\
\label{system-2-V}
& \phantom{t} &  + (A_{n} \cos(\tau) + V_n)^{3} - \cos(\tau)
\frac{\langle \cos(\cdot),(A_n \cos(\cdot) + V_n)^3 \rangle_{L^2_{\rm per}(0,6\pi)}}{
\langle \cos(\cdot),\cos(\cdot)\rangle_{L^2_{\rm per}(0,6\pi)}}, \quad n \geq 2.
\end{eqnarray}
Recall that $\beta = \mathcal{O}(\epsilon^{2/3})$ and $\gamma = \mathcal{O}(\epsilon)$
as $\epsilon \to 0$ if $\delta$ is fixed independently of small $\epsilon > 0$.
Provided that the sequence
$\{ A_{n}\}_{n\in\N}$ is bounded and $\| {\bf A} \|_{l^{\infty}(\N)}$ is small as $\epsilon \to 0$,
the Implicit Function Theorem applied to the system (\ref{system-1-V})--(\ref{system-2-V})
yields a unique even solution for ${\bf V} \in l^{2}(\N,H_{e}^{2}(0,6\pi))$ such
that $\langle {\bf V}, \cos(\cdot) \rangle_{L^2_{\rm per}(0,6\pi)} = {\bf 0}$
in the neighborhood of zero solution for small $\epsilon>0$ and ${\bf A} \in l^{\infty}(\N)$.
Moreover, for all small $\epsilon > 0$ and ${\bf A} \in l^{\infty}(\N)$,
there is a positive constant $C > 0$ such that
\begin{equation}
\label{solution-V-bound}
\|{\bf V}\|_{l^{2}(\N,H_{{\rm per}}^{2}(0,6\pi))} \leq C (\epsilon + \| {\bf A} \|^3_{l^{\infty}(\N)}).
\end{equation}
The balance occurs if $\| {\bf A} \|_{l^{\infty}(\N)} = \mathcal{O}(\epsilon^{1/3})$ as $\epsilon \to 0$.

Recall now that $\beta = 2\delta \epsilon^{2/3} - 2 \epsilon + \mathcal{O}(\epsilon^{4/3})$
and $\gamma = \epsilon + \mathcal{O}(\epsilon^{5/3})$ as $\epsilon \to 0$. Substituting the solution
of  the system (\ref{system-1-V})--(\ref{system-2-V}) satisfying (\ref{solution-V-bound})
to the system (\ref{system-1-A})--(\ref{system-2-A}) and using the
scaling transformation $A_n = \epsilon^{1/3} a_n$, $n \in \N$, we
obtain the perturbed difference equation for $\{ a_n \}_{n \in \N}$:
\begin{eqnarray}
\label{system-1-a}
2 \delta a_{1} + \frac{3}{4} a_1^3 + b_3 & = & \epsilon^{1/3} (2 a_1 - a_{2}) + \mathcal{O}(\epsilon^{2/3}), \\
\label{system-2-a}
2 \delta a_{n} + \frac{3}{4} a_n^3 & = & \epsilon^{1/3} (2 a_n - a_{n+1} - a_{n-1}) + \mathcal{O}(\epsilon^{2/3}), \quad n \geq 2.
\end{eqnarray}

At $\epsilon = 0$, the system (\ref{system-1-a}) and (\ref{system-2-a})
is decoupled. Let $a(\delta)$ be a root of the cubic equation:
\begin{equation}
2\delta a(\delta) + \frac{3}{4}a^{3}(\delta) + c_{3}(6\pi) = 0,
\label{cubic-equation}
\end{equation}
where $c_3(6\pi) \neq 0$ is given.
Roots of the cubic equation (\ref{cubic-equation}) are shown on Figure
\ref{fig:a_vs_delta} for $c_3(6\pi) < 0$. A positive root continues across $\delta=0$ and the
two negative roots bifurcate for $\delta<0$ by means of a saddle-node
bifurcation.
\begin{figure}
\begin{centering}
\includegraphics[width=0.47\textwidth]{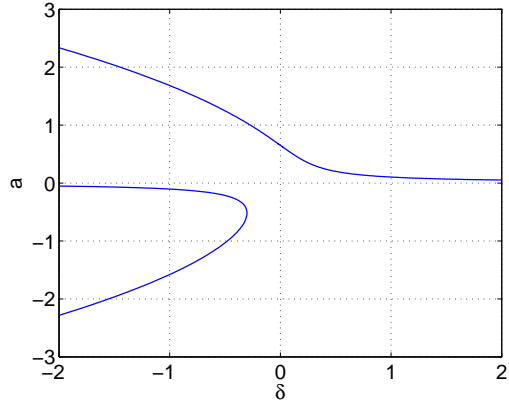}
\par\end{centering}
\caption{Roots of the cubic equation (\ref{cubic-equation}).
\label{fig:a_vs_delta}}
\end{figure}

Let $a(\delta)$ denote any root of cubic equation (\ref{cubic-equation})
such that $8\delta+9a^{2}(\delta)\neq0$. Assuming that
$b_3 = c_3(6 \pi) + \mathcal{O}(\epsilon^{2/3})$ as $\epsilon \to 0$
(this assumption is proved later in Lemma \ref{lemma-pitchfork}), the
Implicit Function Theorem yields a unique continuation of this root
in the system (\ref{system-1-a})--(\ref{system-2-a}) for small $\epsilon>0$
and any fixed $\delta \neq 0$:
\begin{equation}
\label{bound-a}
\left\{ \begin{array}{lclcl}
a_{1} & = & a(\delta)+\epsilon^{1/3}\frac{8a(\delta)}{8\delta+9a^{2}(\delta)} & + & \mathcal{O}(\epsilon^{2/3}),\\
a_{2} & = & \phantom{text}-\epsilon^{1/3}\frac{a(\delta)}{2\delta} & + & \mathcal{O}(\epsilon^{2/3}),\\
a_{n} & = & \phantom{t} & + & \mathcal{O}(\epsilon^{2/3}),\quad n\geq3.\end{array}\right.
\end{equation}
Again, these expansions are valid for any fixed $\delta \neq 0$ such that
$8\delta+9a^{2}(\delta)\neq0$.

\begin{remark}
The condition $8\delta+9a^{2}(\delta)=0$ implies
bifurcations among the roots of the cubic equation (\ref{cubic-equation}),
e.g., the fold bifurcation, when two roots coalesce and disappear
after $\delta$ crosses a bifurcation value. The condition $\delta=0$
does not lead to new bifurcations but implies that the values of
$a_{n}$ for $n\geq2$ are no longer as small as $\mathcal{O}(\epsilon^{1/3})$.
Refined scaling shows that if $\delta=0$, then $a_{1}=a(0)+\mathcal{O}(\epsilon^{1/3})$,
$a_{2}=\mathcal{O}(\epsilon^{1/9})$, and $a_{n}=\mathcal{O}(\epsilon^{4/27})$,
$n\geq3$, where $a(0)$ is a unique real root of
the cubic equation (\ref{cubic-equation}) for $\delta = 0$.
\end{remark}

We can now focus on the last remaining equation (\ref{system-0}) of the
rescaled discrete Klein--Gordon equation (\ref{KGlattice-tilded}).
Substituting $U_1 = \epsilon^{1/3} a(\delta) \cos(\tau) + \mathcal{O}_{H^2_{\rm per}(0,6\pi)}(\epsilon^{2/3})$
into equation (\ref{system-0}), we obtain the perturbed normal
form for 1:3 resonance,
\begin{equation}
\ddot{U}_{0}+U_{0}-U_{0}^{3}=\beta U_{0}+\nu\cos(\tau)+\mathcal{O}_{H_{{\rm per}}^{2}(0,6\pi)}(\epsilon^{5/3}),
\label{normal-form-0}
\end{equation}
where $\nu = 2\gamma\epsilon^{1/3}a(\delta) = \mathcal{O}(\epsilon^{4/3})$
as $\epsilon \to 0$. Because $a(\delta) \neq 0$,
we know that $\nu \neq 0$ if $\epsilon \neq 0$. The perturbed normal form
(\ref{normal-form-0}) coincides with the nonlinear Duffing oscillator
perturbed by a small harmonic forcing. The following lemma summarizes the reduction
of the discrete Klein--Gordon equation to the perturbed Duffing equation,
which was proved above with the Lyapunov--Schmidt reduction arguments.

\begin{lemma}
\label{lemma-reduction}
Let $\delta \neq 0$ be fixed independently of small $\epsilon>0$.
Let $a(\delta)$ be a root of the cubic equation (\ref{cubic-equation}) such
that $8\delta+9a^{2}(\delta)\neq0$. Assume that $c_3(6\pi) \neq 0$
among the Fourier coefficients (\ref{fourier-series-1}). For any
$6 \pi$-periodic solution $U_0$ of the perturbed
Duffing equation (\ref{normal-form-0})
satisfying symmetries (\ref{sym-U-0}) such that
\begin{equation}
\label{assumption-U-0}
U_0(\tau) = \varphi(\tau) + \mathcal{O}_{H_{{\rm per}}^{2}(0,6\pi)}(\epsilon^{2/3})
\quad \mbox{\rm as} \quad \epsilon \to 0,
\end{equation}
there exists a solution of the discrete Klein--Gordon equation (\ref{KGlattice-tilded})
such that
\begin{equation}
\left\{ \begin{array}{lclcl}
U_{\pm1}(\tau) & = & \epsilon^{1/3}a(\delta)\cos(\tau)+\epsilon^{2/3}\frac{8a(\delta)}{8\delta+9a^{2}(\delta)}\cos(\tau) & + & \mathcal{O}_{H_{{\rm per}}^{2}(0,6\pi)}(\epsilon),\\
U_{\pm2}(\tau) & = & \phantom{texttexttextte}-\epsilon^{2/3}\frac{a(\delta)}{2\delta}\cos(\tau) & + & \mathcal{O}_{H_{{\rm per}}^{2}(0,6\pi)}(\epsilon),\\
U_{\pm n}(\tau) & = & \phantom{t} & + & \mathcal{O}_{H_{{\rm per}}^{2}(0,6\pi)}(\epsilon),\quad n\geq3.\end{array}\right.
\label{fourier-series-3}
\end{equation}
 \end{lemma}

\begin{remark}
Figure \ref{fig:a_vs_delta} shows that two negative roots
of the cubic equation (\ref{cubic-equation}) bifurcate at $\delta_* < 0$
via the saddle-node bifurcation and exist for $\delta < \delta_*$.
Negative values of $\delta$ correspond
to $T>6\pi$. As $\epsilon$ is small, this saddle-node bifurcation
gives a birth of two periodic solutions with
$$
u_1(0) = \epsilon^{1/3} a(\delta) + \mathcal{O}(\epsilon^{2/3}) < 0.
$$
This bifurcation is observed on Figure \ref{fig:breather_branches} (right),
one of the two new solutions still satisfies the asymptotic
representation (\ref{fourier-series-0}) as $\epsilon\to0$ for fixed
$T>6\pi$. \label{remark-4}
\end{remark}

In what follows, we shall consider the positive root of the cubic equation
(\ref{cubic-equation}) that continues across $\delta = 0$.
We are interested in $6\pi$-periodic solutions of the perturbed normal
form (\ref{normal-form-0}) in the limit of small $\epsilon > 0$
(when $\beta = \mathcal{O}(\epsilon^{2/3})$ and $\nu = \mathcal{O}(\epsilon^{4/3})$ are small).
Since the remainder term is small as $\epsilon\to0$ and the
persistence analysis is rather straightforward, we obtain main results
by studying the truncated Duffing equation with a small harmonic forcing:
\begin{equation}
\ddot{U}+U-U^{3}=\beta U+\nu\cos(\tau).
\label{normal-form-1}
\end{equation}
The following lemma guarantees the persistence of $6\pi$-periodic solutions
with even symmetry in the Duffing equation
(\ref{normal-form-1}) for all small values of $\beta$ and $\nu$.
Note that this persistence is assumed in
equation (\ref{assumption-U-0}) of the statement of Lemma \ref{lemma-reduction}.

\begin{lemma}
There are positive constants $\beta_{0}$, $\nu_{0}$, and $C$ such that
for all $\beta\in(-\beta_{0},\beta_{0})$ and $\nu\in(-\nu_{0},\nu_{0})$,
the normal form equation (\ref{normal-form-1}) admits a unique $6\pi$-periodic
solution $U_{\beta,\nu}\in H_e^{2}(0,6\pi)$ satisfying
symmetries
\begin{equation}
\label{sym-UU}
U_{\beta,\nu}(-\tau) = U_{\beta,\nu}(\tau) = -U_{\beta,\nu}(3 \pi -\tau), \quad \tau \in \R,
\end{equation}
and bound
\begin{equation}
\|U_{\beta,\nu}-\varphi\|_{H_{{\rm per}}^{2}}\leq C(|\beta|+|\nu|).
\end{equation}
Moreover, the map $\R \times \R \ni (\beta,\nu) \mapsto U_{\beta,\nu} \in H^2_e(0,6\pi)$
is $C^{\infty}$ for all $\beta\in(-\beta_{0},\beta_{0})$ and $\nu\in(-\nu_{0},\nu_{0})$.
\label{lemma-pitchfork}
\end{lemma}

\begin{proof} The proof follows by the Lyapunov--Schmidt reduction
arguments. For $\nu=0$ and small $\beta\in(-\beta_{0},\beta_{0})$, there
exists a unique $6\pi$-periodic solution $U_{\beta,0}$
satisfying the symmetry (\ref{sym-UU}), which is $\mathcal{O}(\beta)$-close to $\varphi$
in the $H_{{\rm per}}^{2}(0,6\pi)$ norm. Because the Duffing oscillator
is non-degenerate, the Jacobian operator $L_{\beta,0}$ has a one-dimensional kernel
spanned by the odd function $\dot{U}_{\beta,0}$, where
\begin{equation}
L_{\beta,\nu}=\partial_{t}^{2}+1-\beta-3U_{\beta,\nu}^{2}(t).
\end{equation}
Therefore, $\langle\dot{U}_{\beta,0},\cos(\cdot)\rangle_{L_{{\rm per}}^{2}(0,6\pi)}=0$,
and the unique even solution persists for small $\nu\in(-\nu_{0},\nu_{0})$.
The symmetry (\ref{sym-UU}) persists for all $\nu \in (-\nu_0,\nu_0)$
because both the Duffing oscillator and the forcing term $\cos(\tau)$ satisfy this symmetry.
\end{proof}

\begin{remark}
Lemma \ref{lemma-pitchfork} excludes the pitchfork bifurcation in the limit $\epsilon \to 0$
for fixed $\delta \neq 0$. This result implies that the period of the pitchfork bifurcation
$T_S$ does not converge to $6 \pi$ as $\epsilon \to 0$. Indeed, we mentioned in the context of
Figure \ref{fig:breather_branches} that $T_S$ gets larger as $\epsilon$ gets smaller.
\label{remark-5}
\end{remark}

By the perturbation theory arguments, the kernel of the Jacobian operator $L_{\beta,\nu}$
is empty for small $\beta$ and $\nu$ provided that $\nu \neq 0$. Indeed, expanding
the solution of Lemma \ref{lemma-pitchfork} in power series in $\beta$
and $\nu$, we obtain
\begin{equation}
U_{\beta,\nu}=\varphi+\beta L_{e}^{-1}\varphi+\nu L_{e}^{-1}\cos(\cdot)+\mathcal{O}_{H_{{\rm per}}^{2}(0,6\pi)}(\beta^{2},\nu^{2}),
\end{equation}
where $L_{e}$ is the operator in (\ref{inhomogen-2}). Although
$L_{e}$ has a one-dimensional kernel spanned by $\dot{\varphi}$,
this eigenfunction is odd in $\tau$, whereas $\varphi$ and $\cos(\cdot)$
are defined in the space of even functions.
Expanding potentials of the operator $L_{\beta,\nu}$, we obtain
\begin{equation}
L_{\beta,\nu}\dot{U}_{\beta,\nu}=\nu\sin(\cdot)+\mathcal{O}_{H_{{\rm per}}^{2}(0,6\pi)}(\beta^{2},\nu^{2}).
\end{equation}
We note that
\[
\langle\dot{\varphi},\cos(\cdot)\rangle_{L_{{\rm per}}^{2}(0,6\pi)} =
\langle\varphi,\sin(\cdot)\rangle_{L_{{\rm per}}^{2}(0,6\pi)}\neq 0
\]
if $c_3(6\pi)\neq0$, where $c_3(T)$ is defined by the Fourier
series (\ref{fourier-series-1}). By the perturbation theory,
the kernel of $L_{\beta,\nu}$
is empty for small $\nu\in(-\nu_{0},\nu_{0})$.

If the linearization operator $L_{\beta,\nu}$ becomes non-invertible
along the curve $\nu=\nu_S(\beta)$ of the codimension one bifurcation,
the symmetry-breaking (pitchfork) bifurcation occurs at $\nu=\nu_S(\beta)$.
This property gives us a criterion to find the pitchfork bifurcation
numerically, in the context of the Duffing equation (\ref{normal-form-1}).
Figure \ref{fig:FMs_vs_nu} (left) shows the behavior of Floquet multipliers
of equation $L_{\beta,\nu}W=0$ with respect to parameter $\nu$ at
$\beta=0$. We can see from this picture that the pitchfork bifurcation occurs
at $\nu\approx0.00015$.

The right panel of Figure \ref{fig:FMs_vs_nu} gives the dependence of the bifurcation
value $\nu_S$ on $\beta$,
for which the operator $L_{\beta,\nu_S(\beta)}$ in not invertible on $L^2_e(0,6\pi)$.
Using the formula for $\beta$ in (\ref{scaling-normal-form}),
we obtain
$$
T = 6\pi  \frac{\sqrt{1-\beta}}{\sqrt{1+2\epsilon}}.
$$
As the coupling constant $\epsilon$ goes to zero, so does parameter
$\nu$. As shown on Figure \ref{fig:FMs_vs_nu} (right), parameter
$\beta$ at the bifurcation curve goes to negative infinity as $\nu\to0$.
This means that the closer we get to the anti-continuum limit, the
further away from $6\pi$ moves the pitchfork bifurcation period $T_S$.
This confirms the early observation that $T_S$ gets larger as $\epsilon$ gets smaller
(see Remark \ref{remark-5}).

Figure \ref{fig:form-1_solutions} shows one solution of Lemma \ref{lemma-pitchfork}
for $0\leq\nu\leq\nu_S(\beta)$ and three solutions for $\nu>\nu_S(\beta)$,
where $\beta=0$. The new solution branches are still given by
even functions but the symmetry $U(\tau) = -U(3\pi-\tau)$ is now broken.
This behavior resembles the pitchfork bifurcation shown on
Figure \ref{fig:solutions}.

Figure \ref{fig:delta_vs_eps} transfers the behavior of Figures \ref{fig:FMs_vs_nu}
and \ref{fig:form-1_solutions} to parameters $T$, $\epsilon$, and $a_0 = u_0(0)$.
The dashed line on the left panel shows the dependence of period $T_S$ at the pitchfork
bifurcation versus $\epsilon$ for the full system (\ref{eq:dKG_trunc}).
The right panel of Figure \ref{fig:delta_vs_eps} can be compared with
the inset on the left panel of Figure \ref{fig:breather_branches}.

\begin{remark}
Numerical results on Figures \ref{fig:form-1_solutions} and \ref{fig:delta_vs_eps}
indicate that the Duffing equation with a small harmonic forcing (\ref{normal-form-1})
allows us to capture the main features of the symmetry-breaking bifurcations
in the discrete Klein--Gordon equation (\ref{eq:dKG_trunc}). Nevertheless,
we point out that the rigorous results of Lemmas \ref{lemma-reduction} and
\ref{lemma-pitchfork} are obtained far from the pitchfork bifurcation, because
parameter $\delta$ is assumed to be fixed independently of $\epsilon$ in these lemmas.
To observe the pitchfork bifurcation on Figures \ref{fig:form-1_solutions} and \ref{fig:delta_vs_eps},
parameter $\delta$ must be sent to $-\infty$ as $\epsilon$ reduces to zero.
\end{remark}

\begin{center}
\begin{figure}
\begin{centering}
\begin{tabular}{cc}
\includegraphics[width=0.45\textwidth]{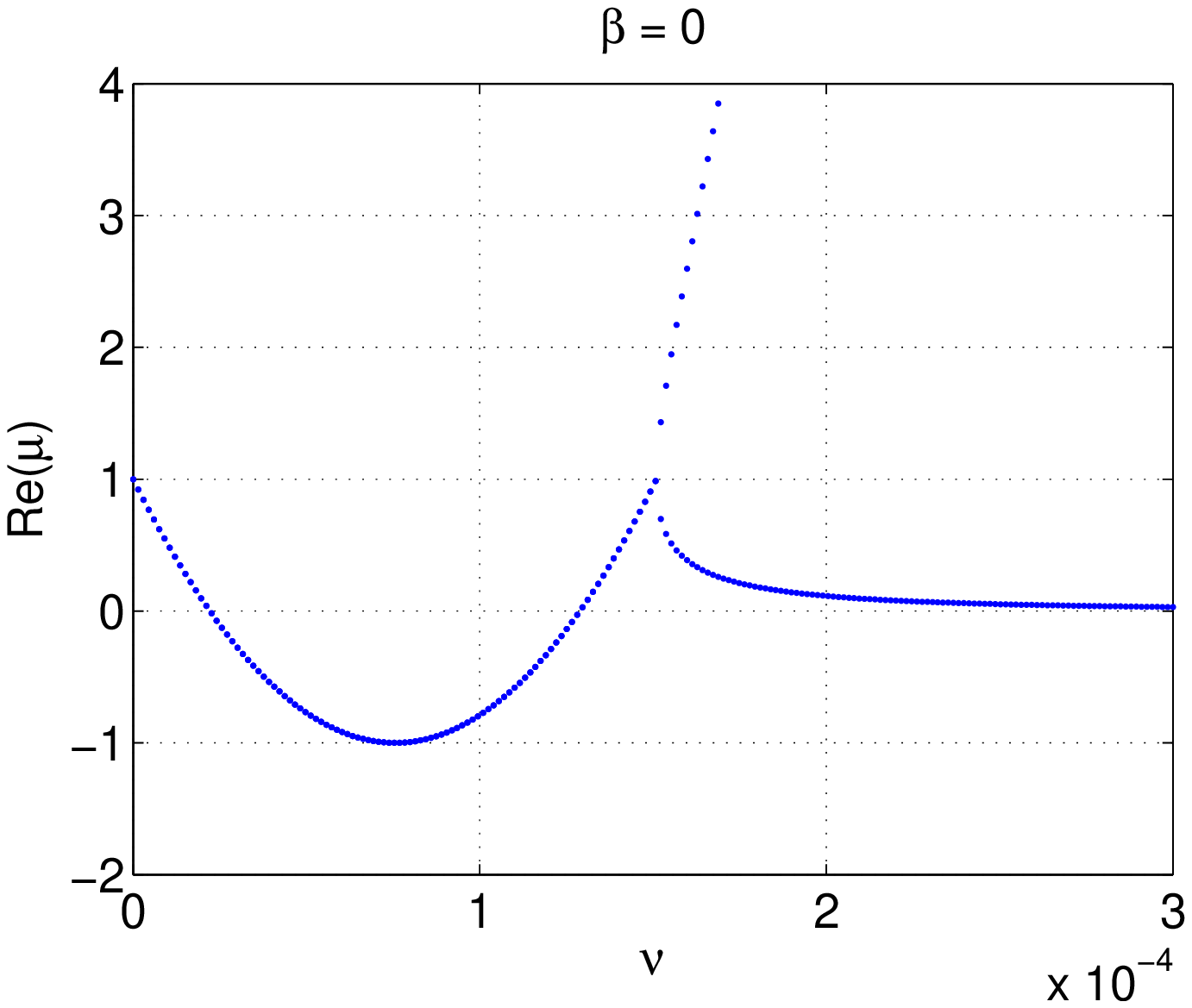}  & \includegraphics[width=0.45\textwidth]{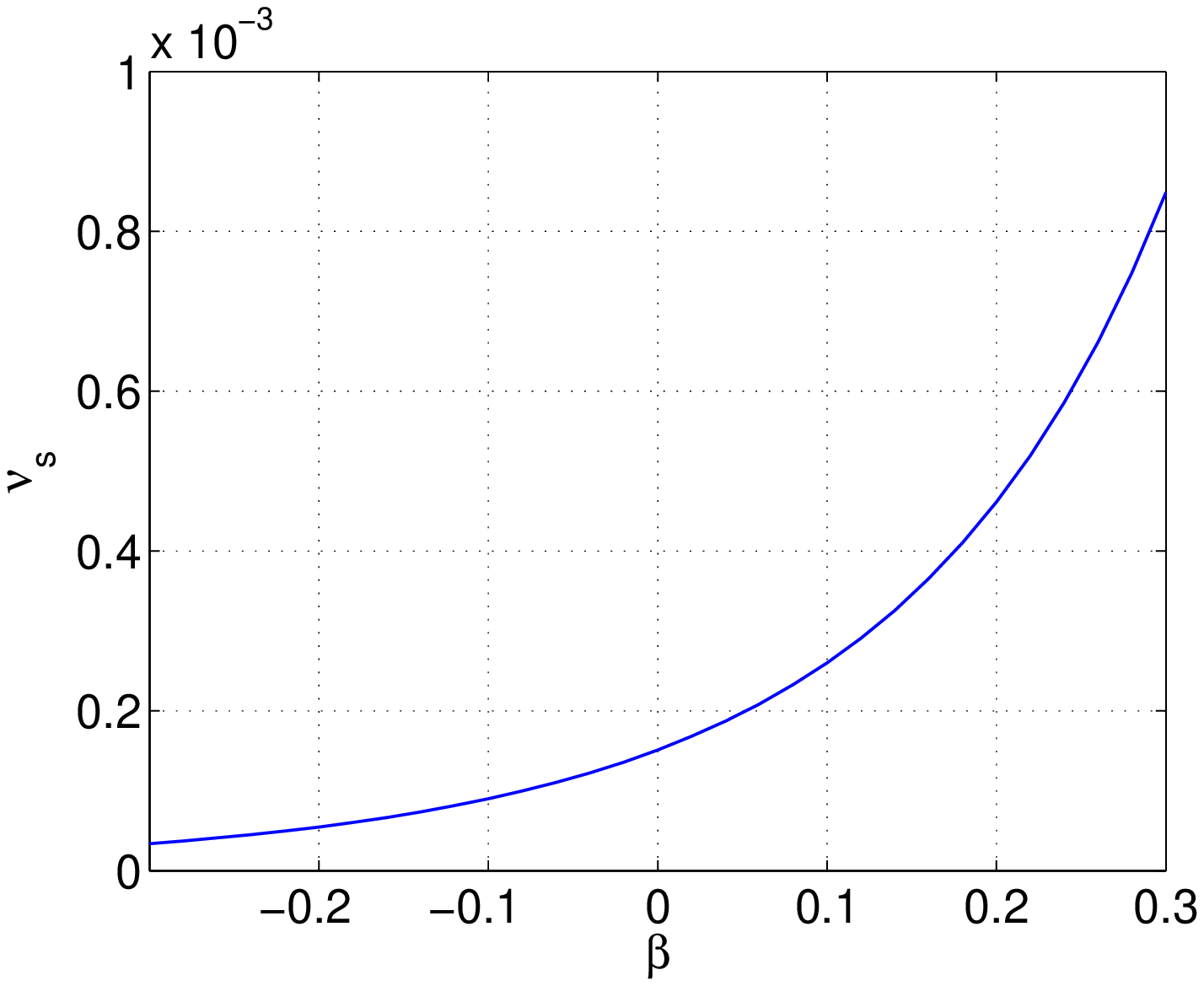}\tabularnewline
\end{tabular}
\par\end{centering}
\caption{Left: Floquet multipliers $\mu$ of equation $L_{\beta,\nu}W=0$.
Right: Parameter $\nu$
versus $\beta$ at the symmetry-breaking bifurcation. \label{fig:FMs_vs_nu}}
\end{figure}
\begin{figure}
\begin{centering}
\begin{tabular}{cc}
\includegraphics[width=0.45\textwidth]{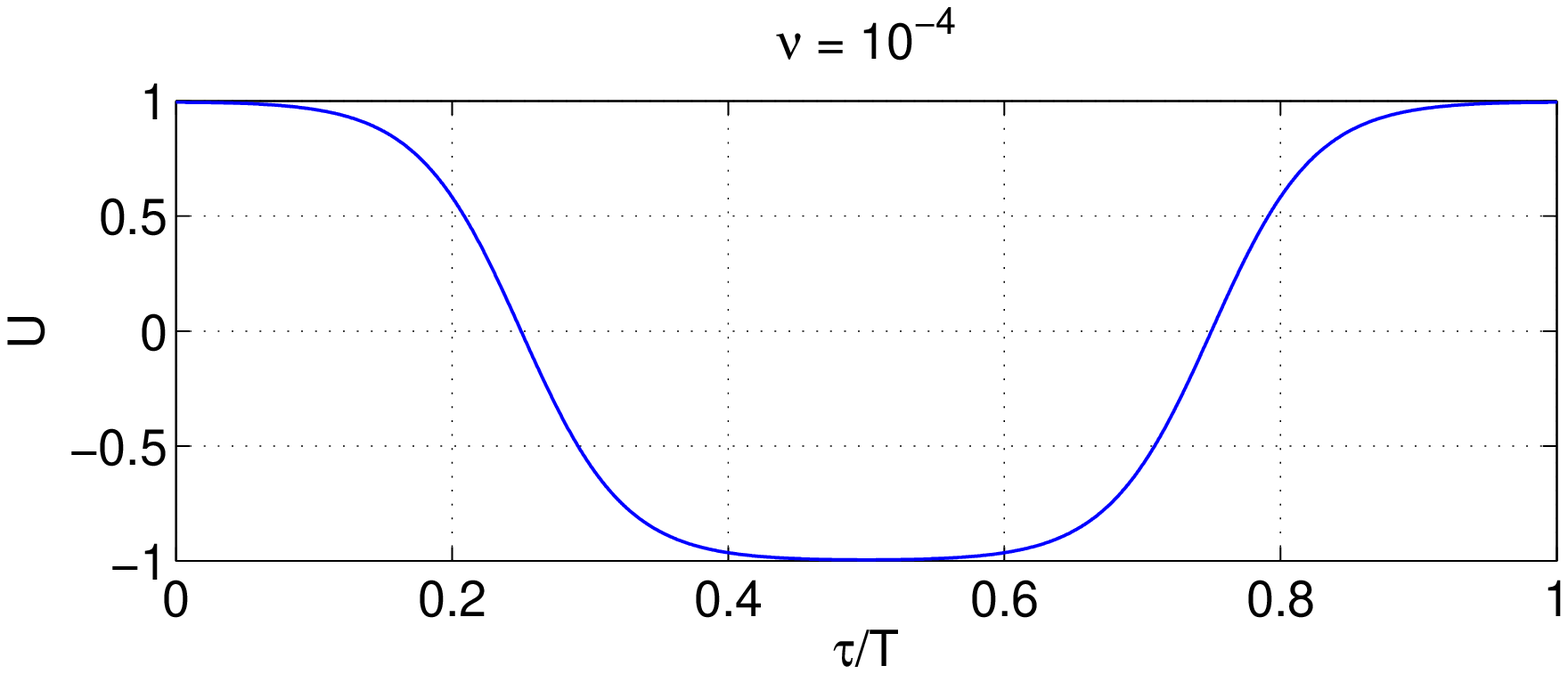} & \includegraphics[width=0.45\textwidth]{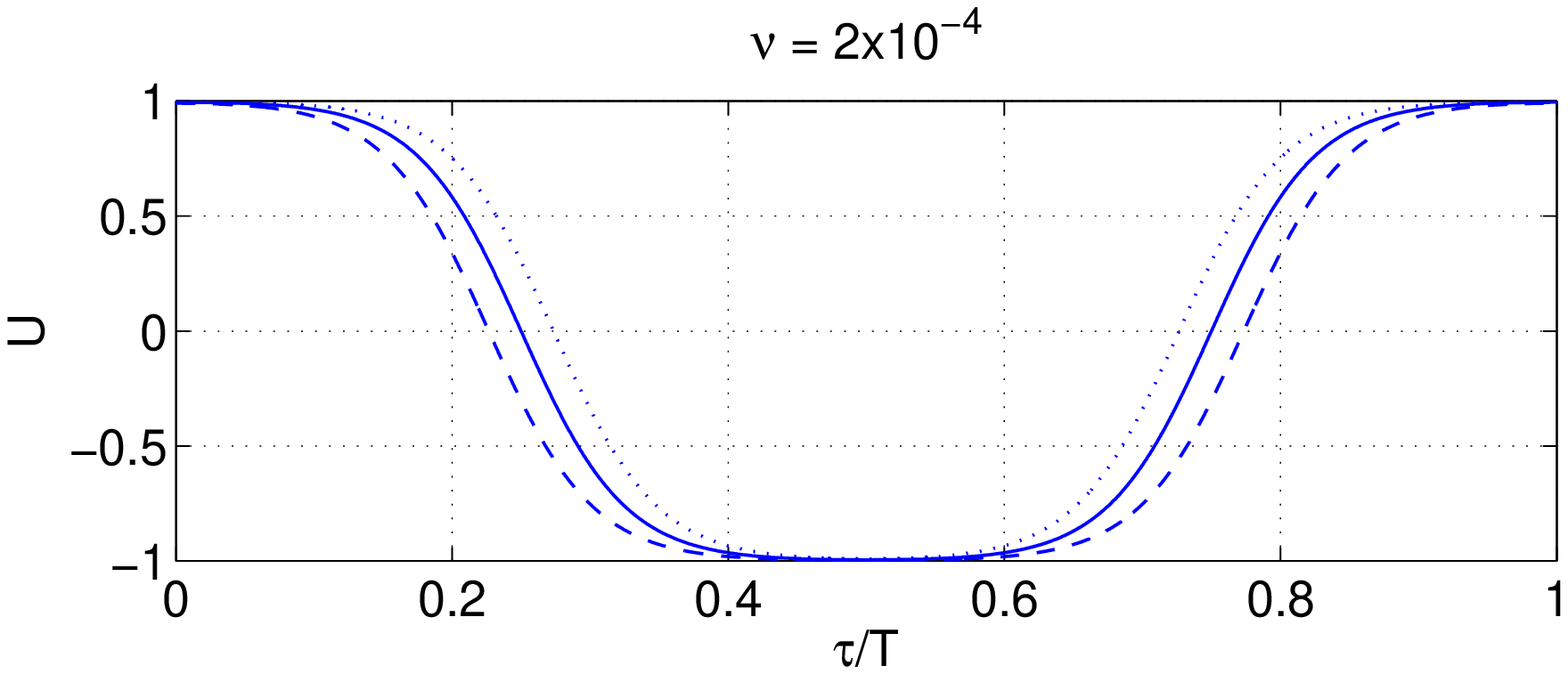}\tabularnewline
\end{tabular}
\par\end{centering}
\caption{Solutions with period $T=6\pi$ to equation (\ref{normal-form-1})
at $\beta=0$ before (left) and after (right) the symmetry-breaking
bifurcation. \label{fig:form-1_solutions}}
\end{figure}
\begin{figure}
\begin{centering}
\begin{tabular}{cc}
\includegraphics[width=0.45\textwidth]{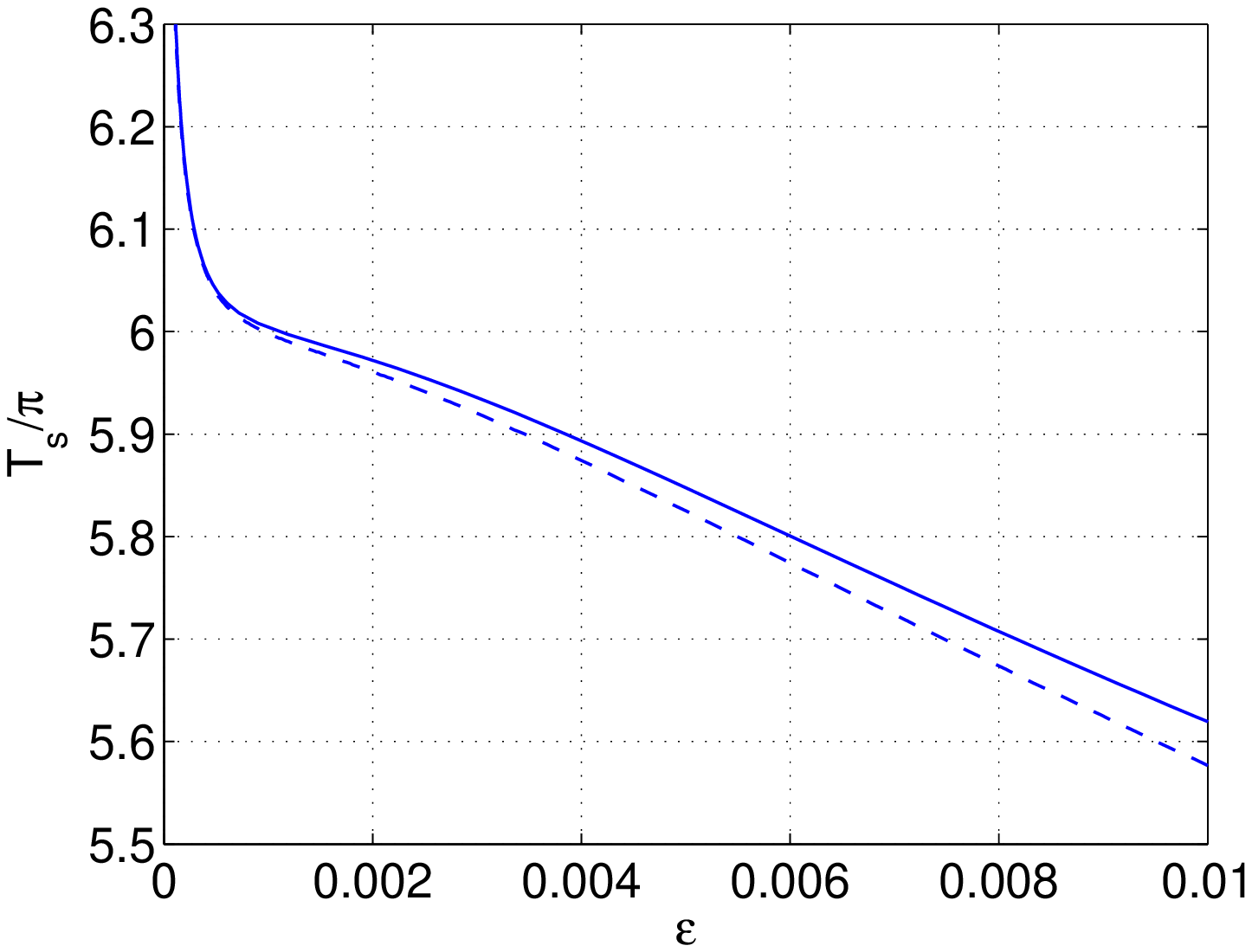}  & \includegraphics[width=0.45\textwidth]{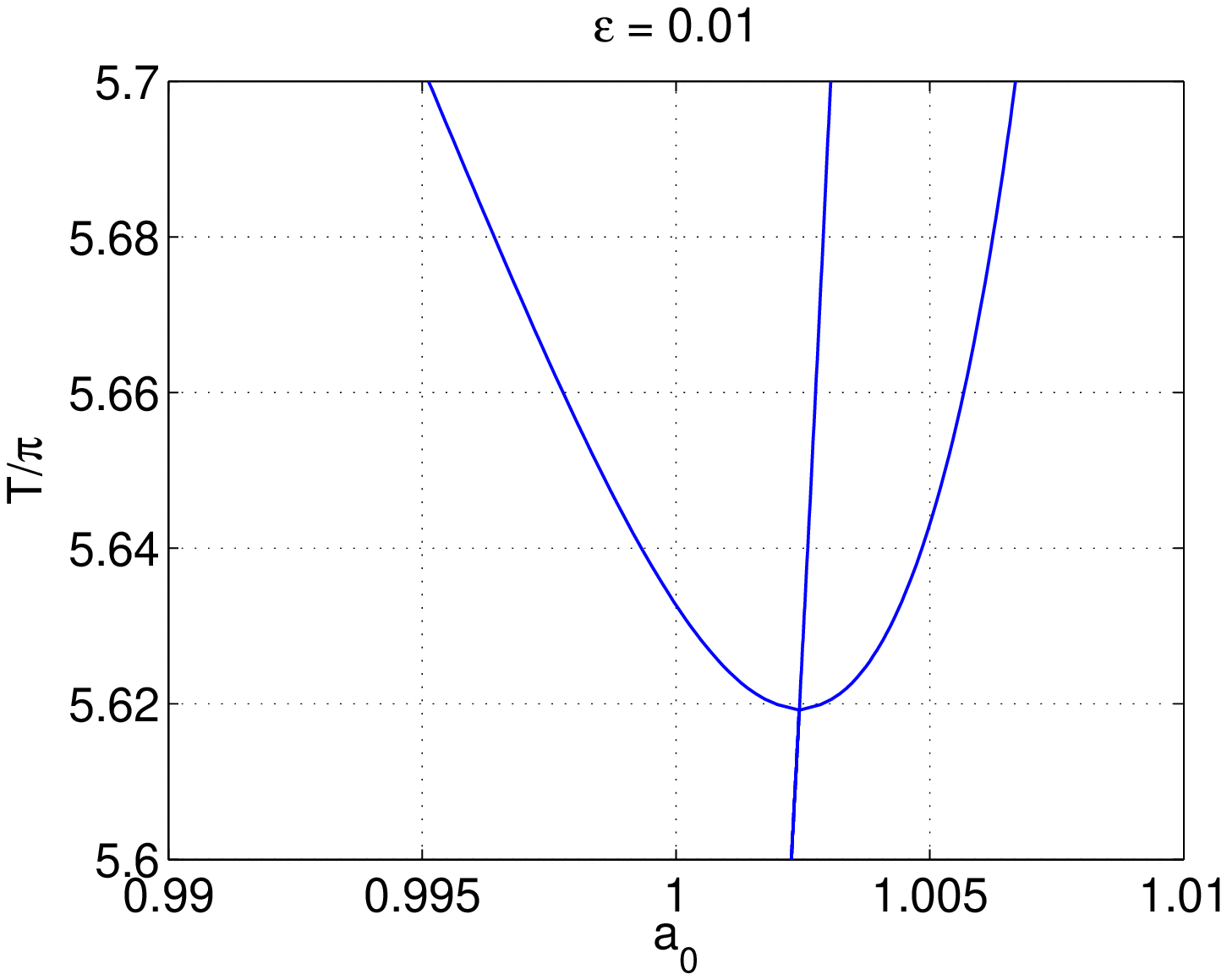}\tabularnewline
\end{tabular}
\par\end{centering}
\caption{Left: Period $T_S$ versus $\epsilon$ at the symmetry--breaking
bifurcation of the fundamental breather modelled by equation (\ref{normal-form-1})
(solid line) and equation (\ref{eq:dKG_trunc}) (dashed line). Right:
Bifurcation diagram for the initial displacement $u_0(0)=a_{0}$
and period $T$ in variables  (\ref{scal-T}) computed from the $6\pi$-periodic
solution to equation (\ref{normal-form-1}). \label{fig:delta_vs_eps}}
\end{figure}
\end{center}

\section{Conclusion}

We have considered existence and stability of multi-site breathers in the Klein--Gordon
lattices with linear couplings between neighboring particles. We have described explicitly
how the stability or instability of a multi-site breather
depends on the phase difference and distance between the excited oscillators.

It is instructive to compare our results to those obtained by Yoshimura \cite{Yoshimura-2} for
the lattices with purely anharmonic coupling:
\begin{equation}
\ddot{u}_{n}+u_{n}\pm u_{n}^{k}=\epsilon(u_{n+1}-u_{n})^{k}-\epsilon(u_{n}-u_{n-1})^{k},\label{eq:Yoshimura}
\end{equation}
where $k\geq3$ is an odd integer. Table II summarizes the result of \cite{Yoshimura-2}
for stable configurations of two-site breathers from the configuration
$$
{\bf u}^{(0)}(t) = \sigma_j \varphi(t) {\bf e}_j + \sigma_k \varphi(t) {\bf e}_k,
$$
where $N = |j-k| \geq 1$.

Note that the original results of \cite{Yoshimura-2}
were obtained for finite lattices with open boundary conditions but can be extrapolated
to infinite lattices, which preserve the symmetry of the multi-site breathers.

\medskip{}

\begin{center}
\begin{tabular}{|>{\centering}m{25mm}|>{\centering}m{25mm}|>{\centering}m{50mm}|}
\hline
 & $N$ odd  & $N$ even\tabularnewline
\hline
hard potential

$V'(u)=u+u^{3}$ & in-phase  & anti-phase\tabularnewline
\hline
soft potential

$V'(u)=u-u^{3}$ & anti-phase \\ in-phase
& anti-phase\tabularnewline
\hline
\end{tabular}
\par\end{center}

\textbf{Table II:} Stable two-site breathers in soft and hard potentials from \cite{Yoshimura-2}.

\medskip{}

Table II is to be compared with Table I summarizing our results. Note that
Table I actually covers $M$-site breathers with equal distance $N$ between the excited sites,
whereas Table II only gives the results in the case $M = 2$.
We have identical results for hard potentials and different results for soft potentials.
First, spectral stability of a two-site breather in the anharmonic potentials
is independent of its period of oscillations and is solely determined by
its initial configuration (Table II). This is different from the transition from
stable anti-phase to stable in-phase breathers for even $N$ in soft potentials (Table I).
Second, both anti-phase and in-phase two-site breathers with odd $N$ are
stable in the anharmonic lattice. The surprising stability of in-phase breathers
is explained by additional symmetries in the anharmonic potentials. The symmetries
trap the unstable Floquet multipliers $\mu$
associated with in-phase breathers for odd $N$ at the point $\mu = 1$. Once the symmetries are
broken (e.g., for even $N$), the Floquet multipliers $\mu$ split along the real axis and the
in-phase two-site breather becomes unstable in soft potentials.

We have also illustrated bifurcations
of breathers near the point of $1:3$ resonance. It is important
to note that a similar behavior is observed near points
of $1:k$ resonance, with $k$ being an odd natural number.
For the non-resonant periods, a breather has large amplitudes
on excited sites and small amplitudes on the other sites. As we increase
the breather's period approaching a resonant point $T = 2\pi k$ for odd $k$,
the amplitudes at all sites become large, a cascade of pitchfork bifurcations
occurs for these breathers, and families of these breathers deviate from the
one prescribed by the anti-continuum limit. However, due to the saddle-node
bifurcation, another family of breathers satisfying Theorem \ref{theorem-continuum}
emerges for periods just above the resonance value.
The period--amplitude curves, similar to those on Figure \ref{fig:breather_branches},
start to look like trees with branches at all resonant points $T = 2\pi k$ for odd $k$.
In the anti-continuum limit, the gaps at the period--amplitude curves vanish while the points
of the pitchfork bifurcations go to infinity. The period--amplitude
curves turn into those for the set of uncoupled anharmonic oscillators.

\textbf{Acknowledgements.} The authors thank P. Kevrekidis and K.
Yoshimura for posing the problem and stimulating discussions, which resulted in this work.
%They also thank G. James and T. Kapitula for refereeing this paper and
%providing a number of useful suggestions.

\appendix
\section{Comparison of Floquet theory with spectral band theory and
Hamiltonian averaging}

\label{sec-equivalence}

We will show here that the criterion of Lemma \ref{lemma-count} for
$N=1$ agrees exactly with the main conclusions of the previous studies
\cite{Archilla,KK09} (see also \cite{ArchKK11}).

The matrix eigenvalue problem (\ref{reduced-eigenvalue}) for $N=1$
can be written in the form
\begin{equation}
A(E) \Lambda^{2}{\bf c}=\mathcal{S}{\bf c},\quad{\bf c}\in\C^{M},
\label{red-eig}
\end{equation}
where
\begin{equation}
A(E)=-\frac{T^{2}(E)}{T'(E)\int_{0}^{T}\dot{\varphi}^{2}dt}.
\label{formula-L}
\end{equation}
We will show that the quantity $A(E)$ arises both in the spectral band
theory used in \cite{Archilla} and in the Hamiltonian averaging
used in \cite{KK09}.

For the spectral band theory \cite{Archilla}, we consider solutions of the spectral
problem
\begin{equation}
Lu=\lambda u,\quad L=\partial_{t}^{2}+V''(\varphi(t)) :
H_{{\rm per}}^{2}(0,T)\to L_{{\rm per}}^{2}(0,T).
\end{equation}
Let $M=\Phi(T)$ be the monodromy matrix computed from the fundamental
matrix solution $\Phi(t)$ of the system
\begin{equation}
\frac{d}{dt}\Phi(t)=\left[\begin{array}{cc}
0 & 1\\
\lambda-V''(\varphi(t)) & 0\end{array}\right]\Phi(t),
\label{system-ODEs}
\end{equation}
subject to the initial condition $\Phi(0)=I\in\mathbb{M}^{2\times2}$.
Since ${\rm det}(M)=1$, the Floquet multipliers $\mu_{1}$ and $\mu_{2}$
satisfy
\[
\mu_{1}\mu_{2}=1,\quad\mu_{1}+\mu_{2}={\rm tr}(M)\equiv F(\lambda).
\]
In particular, $\mu_{1}=\mu_{2}=1$ if ${\rm tr}(M)=2$, which is
true at $\lambda=0$ thanks to the exact solution (\ref{exact-solution-2}),
\[
\Phi(t)=\left[\begin{array}{cc}
\frac{\partial_{E}\varphi(t)}{a'(E)} & \frac{\dot{\varphi}(t)}{\ddot{\varphi}(0)}\\
\frac{\partial_{E}\dot{\varphi}(t)}{a'(E)} & \frac{\ddot{\varphi}(t)}{\ddot{\varphi}(0)}\end{array}\right]\quad\Rightarrow\quad
\Phi(T)=\left[\begin{array}{cc}
1 & 0\\
T'(E)[V'(a)]^{2} & 1\end{array}\right].
\]
Hence, we have $F(0) = 2$. We will show that $A(E)$ in (\ref{formula-L}) determines the sign of
$F'(0)$. Denote elements of $M=\Phi(T)$ by $M_{i,j}$ for $1\leq i,j\leq2$.
Since ${\rm det}(M)=1$ for all $\lambda$, we obtain
\[
F'(0)=\partial_{\lambda}(M_{11}+M_{22})|_{\lambda=0} =
M_{21}\partial_{\lambda}M_{12}|_{\lambda=0}=T'(E) [V'(a)]^{2}\partial_{\lambda}M_{12}|_{\lambda=0}.
\]
Let $U(t)$ be a solution of
\begin{equation}
\ddot{U}(t)+V''(\varphi(t))U(t)=\dot{\varphi}(t),
\label{second-order-equation-U}
\end{equation}
subject to the initial condition $U(0)=\dot{U}(0)=0$. Then, $U(T)=\partial_{\lambda}M_{12}|_{\lambda=0}\ddot{\varphi}(0)$.
Solving the second-order equation (\ref{second-order-equation-U}),
we obtain an explicit solution
\[
U(t)=\dot{\varphi}(t)\left(1-\int_{0}^{t}\dot{\varphi}(s)
\partial_{E}\varphi(s)ds\right)+\partial_{E}\varphi(t)\int_{0}^{t}\dot{\varphi}^{2}(s)ds,
\]
from which we find that
\[
F'(0)=-T'(E)\int_{0}^{T}\dot{\varphi}^{2}(t)dt=\frac{T^{2}(E)}{A(E)}.
\]

If $T'(E)<0$ (for the hard potentials), we have $F'(0)>0$, which
implies that the spectral band of the purely continuous spectrum of
operator $L$ in $L^{2}(\R)$ is located to the left of $\lambda=0$.
If $T'(E)>0$ (for the soft potentials), we have $F'(0)<0$, which
implies that the spectral band of $L$ in $L^{2}(\R)$ is located
to the right of $\lambda=0$. If $\lambda=\omega(k)$ is the dispersion
relation of the spectral band for $k\in\left[-\frac{\pi}{T(E)},\frac{\pi}{T(E)}\right]$,
then near $\lambda=0$, we have
\[
\omega(k)=-\frac{T^{2}(E)}{F'(0)}k^{2}+{\cal O}(k^{4})\quad\mbox{{\rm as}}\quad k\to0\quad\Rightarrow\quad\omega''(0)=-2A(E).
\]
The identity $\omega''(0) = -2A(E)$ establishes the equivalence of the
matrix eigenvalue problem (\ref{red-eig}) with the spectral band
theory used in \cite{Archilla}.

For the Hamiltonian averaging \cite{KK09}, we consider the action variable
for the nonlinear oscillator (\ref{nonlinear-oscillator}),
\[
J=4\int_{0}^{a(E)}\sqrt{2(E-V(\varphi))}d\varphi.
\]
Explicit computation shows that
\[
\frac{dJ}{dE}=2\sqrt{2}\int_{0}^{a(E)}\frac{d\varphi}{\sqrt{E-V(\varphi)}}=T(E).
\]
If $\omega=\frac{1}{T(E)}$ is the frequency of oscillations, then
\[
\frac{dE}{dJ}=\omega(J)\quad\Rightarrow\quad\frac{d^{2}E}{dJ^{2}}= \frac{d\omega}{dJ}=-\frac{T'(E)}{T^{3}(E)}=\frac{1}{T(E) A(E) \int_{0}^{T}\dot{\varphi}^{2}dt}.
\]
Therefore, the signs of $A(E)$ and $E''(J)$ coincide and this establishes
the equivalence of the matrix eigenvalue problem (\ref{red-eig})
with the Hamiltonian averaging used in \cite{KK09}.

\end{document}